\documentclass{statsoc}
\usepackage{amsmath}
\usepackage{graphicx,psfrag,epsf}
\usepackage{enumerate}
\usepackage{url} 
\usepackage{float}
\usepackage{etoolbox}
\usepackage[table,xcdraw]{xcolor}

\makeatletter
\patchcmd{\@makecaption}
{\parbox}
{\advance\@tempdima-\fontdimen2} 
{}{}
\makeatother 

\usepackage[]{hyperref}

\usepackage{algorithm,algorithmic}
\usepackage[algo2e]{algorithm2e}

\usepackage{multirow}
\usepackage{lipsum} 
\usepackage{amsfonts}
                  
\usepackage{graphicx}   

\usepackage[utf8]{inputenc}

\usepackage{amsfonts}
\usepackage{amssymb}
\usepackage{graphicx}
\usepackage{subfigure}
\usepackage{transparent}
\usepackage{multirow}
\usepackage{color}
\usepackage{blkarray}
\usepackage{booktabs}

\hypersetup{colorlinks=true, 
	linkcolor=red,       
    citecolor=blue,        
    filecolor=magenta,      
    urlcolor=cyan           
}  
\usepackage[a4paper]{geometry}

\newcommand{\bm}[1]{\boldsymbol{#1}}
\newcommand{\E}{\mathbb{E}}

\renewcommand{\d}{\mathrm{d}}
\renewcommand{\P}{\mathbb{P}}

\newcommand{\Y}{\bm{Y}}
\newcommand{\bG}{\bm{\Gamma}}
\newcommand{\mE}{\mathcal{E}}

\newcommand{\mA}{\mathcal A}

\newcommand{\mX}{\mathcal{X}}
\newcommand{\mT}{\mathcal{T}}
\newcommand{\wt}[1]{\widetilde{#1}}
\def\mS{\mathcal{S}}
\def\mC{\mathcal{C}}

\newcommand{\abs}[1]{\left\vert#1\right\vert}

\newcommand{\1}{\mathbb{I}}
\newcommand{\iid}{\stackrel{iid}{\sim}}

\newcommand{\wh}[1]{\smash{\widehat{#1}}}

\usepackage{array}

\def\C {\,|\:}

\def\C {\,|\:}

\def\G{\bm{\Gamma}}
\def\mF{\mathcal{F}}
\def\mI{\mathcal{I}}
\def\B{\bm{B}}
\def\b{\bm{\beta}}
\def\Y{\bm{Y}}

\def\x{\bm{x}}

\def\bg{\bm{\gamma}}

\def\mV{\mathcal{V}}

\def\b{\bm{\beta}}

\renewcommand{\d}{\mathrm{d}\,}
\newcommand{\e}{\mathrm{e}}

\newcommand{\N}{\mathbb{N}}
\newcommand{\R}{\mathbb{R}}
\newcommand{\Ha}{\mathcal{H}^\alpha}

\newtheorem{definition}{Definition}[section]
\newtheorem{lemma}{Lemma}[section]

\newtheorem{theorem}{Theorem}[section]
\newtheorem{example}{Example}[section]
\newtheorem{remark}{Remark}[section]

  \title[Variable Selection with ABC Bayesian Forests]{Variable Selection with ABC Bayesian Forests}
\author[Liu {\it et al.}]{Yi Liu} 
\address{Department of  Statistics, University of Chicago } 
\author[Liu {\it et al.}]{Veronika Ro\v{c}kov\'{a}}
\address{Booth School of Business, University of Chicago}
\author[Liu {\it et al.}]{Yuexi Wang}
\address{Booth School of Business, University of Chicago}

\begin{document}


  
\maketitle

\bigskip
\begin{abstract}
Few problems in statistics are as perplexing as variable selection in the presence of very many redundant covariates. The variable selection problem is most familiar in parametric environments such as the linear model or additive variants thereof. In this work, we abandon the linear model framework, which can be quite detrimental when the covariates impact the outcome in a non-linear way, and turn to tree-based methods for variable selection. Such variable screening is traditionally done by pruning down large trees or by ranking variables based on some importance measure. Despite heavily used in practice, these ad-hoc selection rules are not yet well understood from a theoretical point of view.  In this work, we devise a Bayesian tree-based probabilistic method  and show that it  is {consistent} for variable selection when the regression surface is a smooth mix of $p>n$ covariates. 
These results are the first model selection consistency results for Bayesian forest priors. 
Probabilistic assessment of variable importance is made feasible by a spike-and-slab wrapper around  sum-of-trees priors. 
Sampling from posterior distributions over trees is inherently very difficult.
As an alternative to MCMC, we propose  {ABC Bayesian Forests}, a new ABC sampling method based on data-splitting that achieves higher ABC acceptance rate.  We show that the method is robust and successful at finding variables with high marginal inclusion probabilities. Our ABC  algorithm provides a new avenue towards approximating the median probability model in non-parametric  setups where the marginal likelihood is intractable.
\end{abstract}

\keywords{Approximate Bayesian Computation, BART,  Consistency, Spike-and-Slab, Variable Selection}



\section{Perspectives on Non-parametric Variable Selection}
In its simplest form, variable selection is most often carried out in the context of linear regression \citep{tibshirani1996regression, george1993variable, fan2001variable}.
However, confinement to linear parametric forms can be quite detrimental for variable importance screening, when the covariates impact the outcome in a non-linear way \citep{turlach2004least}.
Rather than first selecting a parametric model to filter out variables,  another strategy is to first select variables  and then build a model. Adopting this reversed point of view, we focus on developing methodology for the so called ``model-free" variable selection \citep{chipman2001practical}.

There is a long strand of literature on the fundamental problem of non-parametric variable selection.  One line of research focuses on capturing non-linearities and interactions with basis expansions and performing grouped shrinkage/selection on sets of coefficients \citep{scheipl2011spikeslabgam,ravikumar2009sparse,lin2006component,radchenko2010variable}. 
\citet{lafferty2008rodeo} propose the RODEO method for sparse non-parametric function estimation through regularization of the derivative expectation operator and provide a consistency result for the selection of  the optimal bandwidth. \cite{candes2018panning}  propose a model-free knock-off procedure, controlling FDR in settings when the conditional distribution of the response is arbitrary.   In the Bayesian literature, \cite{savitsky2011variable} deploy spike-and-slab priors on covariance parameters of Gaussian processes to erase variables. 
In this work, we focus on other non-parametric  regression techniques, namely  trees/forests which have been ubiquitous throughout machine learning  and statistics \citep{breiman2001random, chipman2010bart}.  The question we wish to address is whether one can leverage the flexibility of regression trees for effective (consistent) variable importance screening.   

While trees are routinely deployed for data exploration,  prediction and causal inference \citep{hill,taddy,gramacy}, they have also been used for dimension reduction and variable selection.
 This is traditionally done by pruning out variables or by ranking them based on some importance measure. The notion of variable importance was originally proposed for CART using overall improvement in node impurity involving surrogate predictors \citep{breiman1984classification}. In random forests, for example, the importance measure consists of a difference between prediction errors before and after noising  the covariate through a permutation in the out-of-bag sample. 
 However, this continuous  variable importance measure is on an arbitrary scale, rendering variable selection ultimately  ad-hoc.  Principled selection of the importance threshold (with theoretical guarantees such as FDR control or model selection consistency)  is still an open problem. Simplified variants of importance measures have begun to be understood theoretically for variable selection only very recently \citep{ishwaran2007variable,kazemitabar2017variable}. 
 
Bayesian trees and forests select variables based on probabilistic considerations. The BART procedure \citep{chipman2010bart} can be adapted for variable selection by forcing the number of available splits (trees) to be small, thereby introducing competition between predictors.
 BART then keeps track of predictor inclusion frequencies and outputs a probabilistic importance measure: an average proportion of all splitting rules inside a tree ensemble that split on a given variable, where the average is taken over the MCMC samples. This measure cannot be directly interpreted as the posterior variable inclusion probability in anisotropic regression surfaces, where wigglier directions require more splits.
 \citet{bleich2014variable} consider a permutation framework for obtaining the null distribution of the importance weights. 
\citet{zhu2015reinforcement} implement reinforcement learning for selection of  splitting variables during  tree construction to encourage splits on fewer more important variables. 
All these developments point to the fact that regularization is key to enhancing  performance of trees/forests in high dimensions.
Our approach differs in that we impose regularization from {\sl outside} the tree/forest through a spike-and-slab wrapper.



Spike-and-slab  variable selection  consistency results have relied  on analytical tractability (approximation availability) of the marginal likelihood \citep{narisetty2014bayesian,johnson2012bayesian,castillo2015bayesian}. Nicely tractable marginal likelihoods are ultimately unavailable in our framework, rendering the majority of the existing theoretical tools inapplicable.  For these contexts, \citet{yang2017bayesian} characterized general conditions for model selection consistency, extending the work of \citet{lember2007universal} to non {\sl iid} setting.  Exploiting these developments, we show variable selection consistency of our non-parametric spike-and-slab approach when the regression function is a smooth mix of covariates.  Building on \citet{rockova2017posterior},  our paper continues the investigation of missing theoretical  properties of Bayesian CART and BART.
We show model selection consistency when the smoothness is known as well as joint consistency for both the regularity level {\sl and} active variable set when the smoothness is not known and when $p>n$.
These results are the first model selection consistency results for Bayesian forest priors.

The absence of a tractable marginal likelihood  complicates not only theoretical analysis, but also computation.  We turn to Approximate Bayesian Computation (ABC) \citep{plagnol2004approximate,marin2012approximate,csillery2010approximate} and propose a procedure for model-free variable selection. Our ABC method {\sl does not} require the use of low-dimensional summary statistics and, as such, it {\sl does not} suffer from the known difficulty of ABC model choice \citep{robert2011lack}. Our  method is based on sample splitting where at each iteration (a) a random subset of data is used to come up with a proposal draw and (b) the rest of the data is used for ABC acceptance.
This new data-splitting  approach increases ABC  effectiveness by increasing its acceptance rate. ABC Bayesian forests relate to the recent line of work on combining machine learning with ABC \citep{pudlo2015reliable,sinica}.  We propose dynamic plots that describe the evolution of marginal inclusion probabilities as a function of the ABC selection threshold. 

The paper is structured as follows. Section 2 introduces the spike-and-slab wrapper around tree priors. Section 3 develops the ABC variable selection algorithm. Section 4 presents model selection consistency results. Section 5 demonstrates the usefulness of the ABC method on simulated data and Section 6 wraps up with a discussion.



\vspace{-0.5cm}

\subsection{Notation}
With $\|\cdot\|_n$ we denote the empirical $L^2$ norm. 
The class of functions $f(\x):[0,1]^p\rightarrow\R$ such that  $f(\cdot)$ is  constant in all directions excluding $\mS_0\subseteq\{1,\dots,p\}$ is denoted with $\mC(\mS_0)$. With $\Ha_p$, we denote $\alpha$-H\"{o}lder continuous functions with a smoothness coefficient $\alpha$.  $a\lesssim b $ denotes $a$ is less or equal to $b$, up to a multiplicative positive constant, and $a\asymp b$ denotes $a\lesssim b$ and $b\lesssim a$. The $\varepsilon$-covering number of a set $\Omega$ for a semimetric $d$, denoted by $N(\varepsilon; \Omega; d),$ is the minimal number of $d$-balls of radius $\varepsilon$ needed to cover set $\Omega$.

\vspace{-0.5cm}
\section{Bayesian Subset Selection with Trees}
We will work within the purview of non-parametric regression, where a vector of continuous responses $\Y^{(n)}=(Y_1,\dots,Y_n)'$ is  linked to fixed (rescaled) predictors $\x_i=(x_{i1},\dots,x_{ip})'\in[0,1]^p$ for $1\leq i\leq n$ through
\begin{equation}\label{model}
Y_i=f_0(\x_i)+\varepsilon_i\quad\text{with}\quad \varepsilon_i\sim\mathcal{N}(0,\sigma^2)\quad\text{for}\quad 1\leq i\leq n,
\end{equation}
where $f_0(\cdot)$ is  the regression mixing function and $\sigma^2>0$ is a scalar. It is often reasonable to expect that only a small  subset  $\mS_0$ of $q_0=|\mS_0|$ predictors actually exert influence on $\Y^{(n)}$ and contribute to the mix. The subset $\mS_0$ is seldom known with certainty and we are faced with the problem of variable selection. Throughout this paper, we assume   that the regression surface is  smoothly varying ($\alpha$-H\"{o}lder continuous) along the active directions $\mS_0$ and constant otherwise, i.e. we write $f_0\in\Ha_p\cap \mC(\mS_0)$.


Unlike linear models that capture the effect of a single covariate with a single coefficient, we permit non-linearities/interactions and capture variable  importance with (additive) regression trees. By doing so, we hope to recover non-linear signals that could be otherwise missed by linear variable selection techniques.


As with any other  non-parametric regression method, regression trees are vulnerable to the curse of dimensionality, where prediction performance deteriorates dramatically as the number of variables $p$ increases.   If an oracle were to isolate the active covariates $\mS_0$, the fastest achievable estimation rate would be $n^{-\alpha/(2\alpha+|\mS_0|)}$. This rate depends only on the intrinsic dimensionality $q_0=|\mS_0|$, not the actual dimensionality $p$ which can be much larger than $n$. Recently, \citet{rockova2017posterior}  showed that with {\sl suitable regularization}, the posterior distribution for Bayesian CART and BART actually concentrates at this fast rate (up to a log factor), adapting to the intrinsic dimensionality and smoothness.
 Later in Section \ref{sec:consist}, we continue their theoretical investigation and focus on consistent {\sl variable selection}, i.e. estimation of $\mS_0$ rather than $f_0(\cdot)$.  Spike-and-slab regularization  plays a key role in obtaining these theoretical guarantees.


\vspace{-0.5cm}
 \subsection{Trees with Spike-and-Slab Regularization}
Many applications offer a plethora of predictors and some form of 
 redundancy penalization has to be incurred to cope with  the curse of dimensionality. Bayesian regression trees were originally conceived for prediction rather than variable selection. Indeed, original tree implementations of Bayesian CART \citep{denison1998bayesian,chipman1998bayesian} do not seem to penalize inclusion of redundant variables aggressively enough. As noted by \citet{linero2018bayesian}, the prior expected number of active variables under  the Bayesian CART prior of \citet{chipman1998bayesian} satisfies  $\lim_{p\rightarrow\infty}\mathbb E [q]=K-1$
 as $p\rightarrow\infty$ where $K$ is the fixed number of bottom leaves.
 This behavior suggests that (in the limit) the prior forces  inclusion of the maximal number of variables while splitting on them only once. This is far from ideal.
 To alleviate this issue, we deploy  the so-called {\sl spike-and-forest priors}, i.e. spike-and-slab wrappers around sum-of-trees priors \citep{rockova2017posterior}. As with the traditional spike-and-slab priors, the specification starts with   a prior distribution over the $2^p$ active variable sets:
\begin{equation}\label{prior:S_orig}
\mS\sim \pi(\mS)\quad\text{for each}\quad \mS\subseteq\{1,\dots,p\}.
\end{equation}
 We elaborate on the specific choices of $\pi(\mS)$ later in Section \ref{sec:ABC} and Section \ref{sec:consist}. 
 
Given the pool of variables $\mS$, a regression tree/forest is grown using {\sl only} variables inside $\mS$. This prevents the trees from using too many variables and thereby from overfitting. 
 Recall that each individual regression tree is characterized by two components: (1)  a tree-shaped $K$-partition of $[0,1]^p$, denoted with $\mT$, and  (2)  bottom node parameters (step heights), denoted with $\b\in\R^K$. Starting with a parent node $[0,1]^p$, each  $K$-partition is grown by recursively dissecting  rectangular cells at chosen internal nodes along one of the active coordinate axes, all the way down to  $K$ terminal nodes. Each tree-shaped  $K$-partition  $\mT=\{\Omega_k\}_{k=1}^K$ consists of $K$ partitioning rectangles $\Omega_k\subset[0,1]^p$. 

While Bayesian CART approximates $f_0(\x)$ with a single tree mappings $f_{\mT,\b}(\x)=\sum_{k=1}^K\1(\x\in \Omega_k)\beta_k$,
 Bayesian Additive Regression Trees (BART) use an aggregate of $T$ mappings
$$
f_{\mE,\B}(\x)=\sum_{t=1}^Tf_{\mT^t,\b^t}(\x)
 $$
where $\mE=\{\mT^1,\dots,\mT^T\}$ is an ensemble of tree partitions and $\B=[\b^1,\dots,\b^T]$ is an ensemble of step coefficients.
 In a fully Bayesian approach, prior distributions have to be specified over the set of tree structures $\mE$ and over terminal node heights $\B$. 
 The spike-and-forest construction can accommodate various tree prior options.

To assign a prior over  $\mE$ for a given $T$, one possibility is to first pick the number of bottom nodes,  independently for each tree,  from a prior
  \begin{equation}\label{prior:K_orig}
  K^t\sim\pi(K)\quad\text{for}\quad K=1,\dots, n,
  \end{equation}
such as the Poisson distribution \citep{denison1998bayesian}. 
{ Given the vector of tree sizes $\bm K=(K^1,\dots, K^T)'$ and a set of covariates $\mS$, we assign a prior over so-called valid ensembles/forests $\mV\mE_{\mS}^{\bm K}$. We say that a tree ensemble $\mE$ is valid if it consists of trees that have  non-empty bottom leaves.} One can pick a tree partition  ensemble from a uniform prior over {\sl valid} forests $\mE\in\mV\mE^{\bm K}_\mS$, i.e.
\begin{equation}\label{prior:tree}
\pi(\mE \C \mathcal{S}, \bm K)=\frac{1}{\Delta(\mV\mE_\mS^{\bm K})}\1\left(\mE\in\mV\mE_\mS^{\bm K}\right),
\end{equation}
where $\Delta(\mV\mE_{\mS}^{\bm K})$ is the number of valid tree ensembles  characterized by $\bm K$  bottom leaves and split directions $\mS$. The prior \eqref{prior:K_orig} and \eqref{prior:tree}  was deployed in the Bayesian CART implementation of \citet{denison1998bayesian} (with $T=1$) and it was studied theoretically by \citet{rockova2017posterior}. 
Another related Bayesian forest prior  (implemented in the BART procedure and studied theoretically by \citet{rockova_saha} consists of an independent product of branching process priors (one for each tree) with decaying split probabilities \citep{chipman1998bayesian}.  The implementation is very similar to the one of  \citet{denison1998bayesian}. 

Finally, given the partitions $\mT^t$ of size $K^t$ for $1\leq t\leq T$, one assigns (independently for each tree) a Gaussian product prior on the step heights 
\begin{equation}\label{prior:beta2}
\pi(\b^t\C K^t)=\prod_{k=1}^{K^t}\phi(\beta_k^t;\sigma_{\b}^2),
\end{equation}
where $\phi(x;\sigma_{\b}^2)$ denotes a Gaussian density with mean zero and variance $\sigma_{\b}^2=1/T$ (as suggested by \citet{chipman2010bart}). The prior for $\sigma^2$ can be chosen as inverse chi-squared with hyperparameters chosen based on an estimate of the residual standard deviation of the data \citep{chipman2010bart}.

The most crucial component in the spike-and-forest construction, which sets it apart from existing BART implementations, is
the active set $\mS$ which  serves to mute variables by restricting the pool of predictors available for splits. The goal is to learn which set $\mS$ is most likely (a posteriori) and/or how likely  each variables is  to have contributed to $f_0$. Unlike related tree-based variable selection criteria,  the spike-and-slab  envelope makes it possible to 
 perform variable selection directly by evaluating posterior model probabilities $\Pi(\mS\C\Y^{(n)})$ or marginal inclusion probabilities $\Pi(j\in\mS_0\C\Y^{(n)})$ for $1\leq j\leq p$. Random forests \citep{breiman2001random} also mute variables, but they do so from within the tree by randomly choosing a small subset of variables for each split.  The spike-and-slab approach mutes variables externally rather than internally.
 \citet{bleich2014variable} note that when the number of trees is small, the Gibbs sampler for BART can get trapped in local modes which can destabilize the estimation procedure.  On the other hand, when the number of trees is large, there are ample opportunities for the noise variables to enter the model without necessarily impacting the model fit, making variable selection very challenging. Our spike-and-slab wrapper is devised to get around this problem.

The problem of variable selection is fundamentally challenged by the sheer size  of possible variable subsets. For  linear regression, (a) MCMC implementations exist that capitalize on the availability of marginal likelihood {\citep{narisetty2014bayesian, guan_stephens}, (b) optimization strategies exist  for both continuous \citep{rockova2016spike,rockova17_pem} and point-mass spike-and slab priors  \citep{carbonetto}. These techniques do not directly translate to tree models, for which tractable marginal likelihoods $\pi(\Y^{(n)}\C\mS)$ are unavailable.
To address this computational challenge, we explore ABC techniques as a new promising avenue for non-parametric spike-and-slab methods.

\vspace{-0.5cm}
\section{ABC for Variable Selection}
Performing (approximate) posterior inference in complex models is often complicated by the analytical intractability of the marginal likelihood. Approximate Bayesian Computation (ABC)  is a simulation-based inference framework that obviates the need to compute the likelihood directly by  evaluating the proximity  of (sufficient statistics of) observed data and pseudo-data simulated from the likelihood.  Simon Tavar{\'e} first proposed the ABC algorithm for posterior inference \citep{tavare1997inferring} in  the 1990's and since then it has widely been used in population genetics, systems biology, epidemiology and phylogeography\footnote{The study of how human beings migrated throughout the world in the past.}.


Combined with  a probabilistic structure over models, marginal likelihoods give rise to posterior model probabilities, a standard  tool for Bayesian model choice.
When the marginal likelihood is unavailable (our case here), ABC offers a unique computational solution. 
However, as pointed out by \citet{robert2011lack}, ABC cannot be trusted for model comparisons when model-wise sufficient summary statistics are not sufficient across models.  The ABC approximation to  Bayes factors then does not converge to exact Bayes factors, rendering  ABC model choice fundamentally untrustworthy. A fresh new perspective to ABC model choice was offered in \citet{pudlo2015reliable}, who rephrase  model selection as a classification problem that can be tackled with machine learning tools. Their idea is to treat the ABC reference table (consisting of samples from a prior model distribution and high-dimensional vectors of  summary statistics of pseudo-data obtained from the prior predictive distribution) as an actual data set, and  to train a random forest classifier that predicts a model label using the summary statistics as predictors. 
Their goal is to produce a stable model decision based on a classifier rather than on an estimate of posterior model probabilities. Our approach has a similar flavor in the sense that it combines machine learning with ABC, but the concept is fundamentally very different. Here, the fusion of Bayesian forests and ABC is tailored to non-parametric variable selection towards obtaining posterior variable inclusion probabilities. Our model selection approach does not suffer from the difficulty of ABC model choice as we {\sl do not} commit to any summary statistics and use random subsets of observations to generate the ABC reference table.

\vspace{-0.5cm}
\subsection{Naive ABC Implementation}
\vspace{0.2cm}
For its practical implementation,  our Bayesian variable selection method requires sampling from  the analytically intractable posterior distribution over subsets $\Pi(\mS\C\Y^{(n)})$ under the {\sl spike-and-forest} prior
 \eqref{prior:tree}, \eqref{prior:K_orig} and \eqref{prior:S_orig}.  Given a single tree partition $\mT$, the (conditional) marginal likelihood $\pi(\Y^{(n)}\C\mT,\mS)$ is available in closed form, facilitating  implementations of Metropolis-Hastings  algorithms \citep{chipman1998bayesian,denison1998bayesian} (see Section \ref{sec:mcmc}). However, such MCMC schemes can suffer from poor mixing. Taking advantage of the fact that, despite being intractable, one can  {\sl simulate from} the marginal likelihood  $\pi(\Y^{(n)}\C\mS)$,  we will explore the potential of ABC as a complementary development to  MCMC implementations. 

The principle at the core of ABC is to perform approximate posterior inference from a given dataset by simulating from a prior distribution and by comparisons with numerous synthetic datasets.
In its standard form, an ABC implementation of model choice creates a reference table, recording a large number of datasets simulated from the model prior and the prior predictive distribution under each model. Here, the table consists of $M$  pairs $(\mS_m,\Y_m^{\star})$ of model indices $\mS_m$, simulated from the prior $\pi(\mS)$, and pseudo-data $\Y^{\star}_m\in\R^{n}$, simulated  from the marginal likelihood $\pi(\Y^{(n)}\C\mS_m)$. To generate $\Y^\star_m$ in our setup, one can hierarchically decompose the marginal
likelihood 
\begin{equation}\label{marginal_lik}
\pi(\Y^{(n)}\C\mS)=\int_{(f_{\mE,\B},\sigma^2)}\pi(\Y^{(n)}\C f_{\mE,\B},\sigma^2)\d\pi(f_{\mE,\B},\sigma^2\C\mS)
\end{equation}
and first  draw $({f^m_{\mE,\B}},\sigma^{2}_m)$ from the prior $\pi({f_{\mE,\B}},\sigma^2\C\mS)$ and obtain  $\Y^\star_m$ from  \eqref{model}, given $({f^m_{\mE,\B}},\sigma^{2}_m)$.
ABC sampling is then followed by an ABC rejection step, which extracts pairs $(\mS_m,\Y^{\star}_m)$ such that $\Y^{\star}_m$ is close enough to the actual observed data. In other words, one  trims the reference table by keeping only model indices $\mS_m$ paired with pseudo-observations that are at most   $\epsilon$-away from the observed data, i.e.
$\|\Y^{obs}-\Y^{\star}_m\|_2\leq \epsilon$ for some tolerance level $\epsilon$. These extracted values comprise an approximate ABC sample from the posterior $\pi(\mS\C\Y^{(n)})$,
 which should be informative for the relative  ordering of the competing models, and thus variable selection \citep{grelaud2009abc}. 
 Note that this particular ABC implementation does not require any use of low-dimensional summary statistics, where rejection is based solely on $\Y^{obs}$.  While theoretically justified, this  ABC variant has two main drawbacks.
 
 First, with very many predictors, it will be virtually impossible to sample from all $2^p$ model combinations at least once, unless the reference table is huge. Consequently, relative frequencies of occurrence of a model $\mS_m$ in the trimmed ABC reference table {\sl may not}  be a good estimate of the posterior model probability $\pi(\mS_m\C\Y^{(n)})$.
 While   the model with the highest posterior probability $\pi(\mS_m\C\Y^{(n)})$ is commonly conceived as the right model choice, it may not be the  optimal model for prediction. Indeed,
 in nested correlated designs and orthogonal designs, it is  the median probability model  that is predictive optimal  \citep{barbieri2004optimal}. The median probability model (MPM)
  consists of those variables whose {\sl marginal} inclusion probabilities $\P(j\in \mS_0\C\Y^{(n)})$ are at least $0.5$.
 While simulation-based estimates of posterior model probabilities $\P(\mS\C\Y^{(n)})$ can be imprecise,  we argue (and show) that ABC estimates of marginal inclusion probabilities $\P(j\in \mS_0\C\Y^{(n)})$ are far more robust and stable. 

 The second difficulty is purely computational and relates to the issue of coming up with good proposals $f_{\mE,\B}^m$ such that the  pseudo-data are sufficiently close to $\Y^{obs}$. Due to the vastness of the tree ensemble space, it would be naive to think that one can obtain solid guesses of $f_0$ purely by sampling from non-informative priors. This is why we call this ABC implementation naive. These considerations lead us to a new data-splitting ABC modification  that uses a random portion of the data to train the prior and to generate pseudo-data with more affinity to the left-out observations.

\vspace{-0.5cm}
 \subsection{ABC Bayesian Forests}\label{sec:ABC}
 \vspace{0.2cm}
By sampling directly from  noninformative priors over tree ensembles $\pi(f_{\mE,\B},\sigma^2\C\mS)$, the acceptance rate of the naive ABC can be prohibitively small where huge reference tables  would be required to obtain only a few approximate samples from the posterior.

To address this problem, we suggest a sample-splitting approach to come up with draws  that are less likely to be rejected by the ABC method. 
At each  ABC iteration, we first draw a random subsample $\mI\subset\{1,\dots,n\}$ of size $|\mI|=s$ with no replacement. 
Then we split the observed data $\Y^{(n)}$ into two groups, denoted with $\Y^{(n)}_{\mI}$ and $\Y^{(n)}_{\mI^c}$, and instead of  \eqref{marginal_lik} we consider the marginal likelihood conditionally on 
$\Y^{(n)}_{\mI}$
{\begin{equation}\label{marglik1}
\pi(\Y^{(n)}\C\Y^{(n)}_{\mI},\mS)=\int_{(f_{\mE,\B},\sigma^2)}\pi(\Y^{(n)}_{\mI^c}\C f_{\mE,\B},\sigma^2)\d\pi_{\mI}(f_{\mE,\B},\sigma^2\C\mS)
\end{equation}}
where
\begin{equation}\label{marglik2}
\pi_{\mI}(f_{\mE,\B},\sigma^2\C\mS)=\pi(f_{\mE,\B},\sigma^2\C\Y^{(n)}_{\mI},\mS).
\end{equation}
This simple decomposition unfolds new directions  for ABC sampling based on data splitting.
Instead of using all observations $\Y^{obs}$ to accept/reject each draw,
we  set aside a random subset of  data  $\Y^{obs}_{\mI^c}$ for ABC rejection and use $\Y^{obs}_{\mI}$ to ``train the prior".
The key observation is that the samples from the prior $\pi_{\mI}(f_{\mE,\B},\sigma^2\C\mS)$, i.e. the {\sl posterior} $\pi(f_{\mE,\B},\sigma^2\C\Y^{(n)}_{\mI},\mS)$, will have seen a part of the data and will produce more realistic guesses  of $f_0$. Such guesses are more likely to yield pseudo-data that match  $\Y^{obs}_{\mI^c}$ more closely, thereby increasing the acceptance rate of ABC sampling. Note that the acceptance step is based solely on  the left-out sample $\Y^{obs}_{\mI^c_m}$, not the entire data. {Similarly as the naive ABC outlined in the previous section, we first sample the subset $\mS$ from the prior $\pi(\mS)$ and then obtain draws from the conditional marginal likelihood under an updated prior $\pi_{\mI}(f_{\mE,\B},\sigma^2\C\mS)$. This corresponds to an ABC strategy for sampling from
$\pi(\mS\C \Y^{(n)}_{\mI^c})$ under the priors \eqref{prior:S_orig} and \eqref{marglik2}.  As will be seen later, this posterior is effective for assessing variable importance. Moreover, if $\pi(\mS)$ is a good proxy for $\pi(\mS\C \Y^{(n)}_{\mI})$ (when the training set is small relative to the ABC rejection set), this ABC will produce approximate samples from the original target $\pi(\mS\C\Y^{(n)})$.}

{

The idea of using a portion of the data for training the prior and the rest for model selection goes back to at least \citet{good1950probability}.  
The most common prescription for choosing training samples in Bayesian analysis 
is to convert improper priors into propers ones for meaningful model selection with Bayes factors \citep{lempers1971posterior,o1995fractional}.
\citet{berger1996intrinsic} advocated  choosing the training set as small as possible subject to yielding proper posteriors (so called minimal training samples). 
\citet{berger2004training}  argue that data can vary widely in terms of their information content and the use of single minimal training samples can be inadequate/ suboptimal. Since there are many possible training samples, it is natural to average the resulting Bayes factors over the training samples in some fashion. 
While intrinsic Bayes factors  \citep{berger1996intrinsic} average Bayes factors over all possible minimal training samples, expected posterior priors \citep{perez2002expected} average the prior first. In particular, the empirical expected-posterior prior for model $\mS$  \citep{ghosh2002nonsubjective, perez2002expected} writes as
\begin{equation}\label{eq:expected_prior}
\pi(f_{\mE,\B},\sigma^2\C\mS)=\frac{1}{L}\sum_{l=1}^L\pi_{\mathcal I_l}(f_{\mE,\B},\sigma^2\C\mS),
\end{equation}
where $\pi_{\mathcal I_l}(f_{\mE,\B},\sigma^2\C\mS)$ was defined in (8) and where $L$ is the number of all minimal training samples $\mI_l$. The marginal likelihood under this prior can be then written as (equation (3.5) in \citet{perez2002expected})
$
m(\Y^{(n)}\C\mS)=\frac{1}{L}\sum_{l=1}^L \pi(\Y^{(n)}\C\Y^{(n)}_\mI,\mS), 
$
where  $\pi(\Y^{(n)}\C\Y^{(n)}_\mI,\mS)$ was defined in (7).  Our ABC analysis with internal data splitting can be thus regarded as arising from the empirical expected posterior prior \eqref{eq:expected_prior}. While the motivation for using training samples in Bayesian analysis has been largely to make improper priors proper, here we use this idea in a different context to increase ABC acceptance rate. 

}

The ABC Bayesian Forests algorithm is formally summarized in Table \ref{ABC}. It starts by splitting the dataset into two subsets at each  ($m^{th}$) iteration:  $\Y^{obs}_{\mI_m}$ for fitting and $\Y^{obs}_{\mI_m^c}$ for ABC rejection. The algorithm then proceeds by sampling an active set $\mS$ from $\pi(\mS)$. Using the spike-and-slab construction, one can draw Bernoulli indicators $\bg=(\gamma_1,\dots,\gamma_p)'$ where $\P(\gamma_j=1\C\theta)=\theta$ for some prior inclusion probability $\theta\in(0,1)$ and set $\mS_m=\{j:\gamma_j=1\}$. 
When sparsity is anticipated, one can choose $\theta$ to be small or to arise from a beta prior  $\mathcal{B}(a,b)$ for some $a>0$ and $b>0$ (yielding the beta-binomial prior). We discuss other suitable prior model choices in Section \ref{sec:consist}.

{\scriptsize
\begin{algorithm}[t]\small
 \KwData{ Data ($Y^{obs}_i$, $\x_i$) for $1\leq i\leq n$}
 \KwResult{$\pi_j(\epsilon)$ for $1\leq j\leq p$ where $\pi_j(\epsilon) = \wh\P(j\in\mS_0\C\Y^{(n)})$}
 \textbf{Set} $M$: the number of ABC simulations; $s$: the subsample size; $\epsilon$: the tolerance threshold; $m=0$ the counter \\
 \While{$m\leq M$}{
 {\vspace{-0.3cm} \hspace{5cm}\color{white} \tiny ahoj}\\
  \textbf{(a) Split}  data $\Y^{obs}$ into $\Y^{obs}_{\mI_m}$ and $\Y^{obs}_{\mI^c_m}$, where  $\mI_m\subset\{1,\dots,n\}$ of size $|\mI_m|=s$ is obtained by sampling with no replacement.\\
  \textbf{(b) Pick} a subset $\mS_m$ from $\pi(\mS)$.\\
  \textbf{(c) Sample}  $(f_{\mE,\B}^m,\sigma^2_m)$ from $\pi_{\mI_m}(f_{\mE,\B},\sigma^2\C\mS_{m})=\pi(f_{\mE,\B},\sigma^2\C\Y^{obs}_{\mI_m},\mS_{m})$.\\
  \textbf{(d) Generate} pseudo-data  $\Y^\star_{\mI_m^c}$ by sampling white noise $\varepsilon_i\iid\mathcal{N}(0,\sigma_m^2)$  and setting \\
  $Y^\star_i=f_{\mE,\B}^m(\x_i)+\varepsilon_i$
   for each $i\notin I_m$. \\
  \textbf{(e) Compute}  discrepancy $\epsilon_m=\|\Y^\star_{\mI_m^c}-\Y^{obs}_{\mI_m^c}\|_2$.\\
  \eIf{$\epsilon_m<\epsilon$
}{Accept $(\mS_m,f_{\mE,\B}^m)$ and set $m=m+1$}
{Reject $(\mS_m,f_{\mE,\B}^m)$  and set $m=m+1$}
}
\textbf{Compute} $\pi_j(\epsilon)$ as the proportion of times $j^{th}$ variable is used in the accepted $f_{\mE,\B}^m$'s. 
 \caption{\bf : ABC Bayesian Forests}\label{ABC}
\end{algorithm}}

In the {\sl (c) step} of ABC Bayesian Forests, one obtains a sample from the posterior of $(f_{\mE,\B},\sigma^2)$, given $\Y^{obs}_{\mI_m}$. For this step, one can leverage existing implementations of Bayesian CART and BART (e.g. the \texttt{BART} R package of \citet{mcculloch2018package}). A single draw from the posterior is obtained after a sufficient burn-in.
In this vein, one can view ABC Bayesian Forests as a computational envelope around BART to restrict the pool of available variables. 
 The {\sl (d) step} then consists of predicting the outcome $\Y^\star_{\mI_m^c}$ for left-out observations  $\x_i$ using \eqref{model} for each $i\in\mI_{m}^c$. The last step is ABC rejection based on the discrepancy between $\Y^\star_{\mI_m^c}$ and $\Y^{obs}_{\mI_m^c}$.

 For the computation of marginal inclusion probabilities $\pi_j(\epsilon)$, one could conceivably report the proportion of ABC accepted  samples  $\mS_m$ that contain the $j^{th}$ variable. However, $\mS_m$ is a pool of {\sl available} predictors and not all of them are necessarily used in $f_{\mE,\B}^m$. Thereby, we report the proportion of ABC accepted samples $f_{\mE,\B}^m$  that use the $j^{th}$ variable at least once, i.e. 
 \begin{equation}\label{pis}
 \pi_j(\epsilon)=\frac{1}{M(\epsilon)}\sum\limits_{m: \epsilon_m<\epsilon}\1(j \,\text{used in}\, f_{\mE,\B}^m),
 \end{equation}
 where $M(\epsilon)$ is the number of accepted ABC samples at $\epsilon$. Each tree ensemble $f_{\mE,\B}^m$ thus performs its own variable selection by picking variables from $\mS_m$ rather than from $\{1,\dots,p\}$.  {Limiting the pool of predictors prevents from too many false positives. In addition, the inclusion probabilities \eqref{pis} do use the training data $\Y^{(n)}_{\mI}$ to shrink and update the subset $\mS$ by leaving out covariates not picked  by $f_{\mE,\B}^m$. In this way,  the  mechanism for selecting the subsets $\mS$ is not strictly sampling from the prior $\pi(\mS)$ but it seizes   the information in the training set $\mI$. In this way, $\mS_m$'s can be regarded as approximate samples from  $\pi(\mS\C\Y^{obs})$. When $\mI=\emptyset$, we recover the naive ABC as a special case.
 }

\vspace{-0.5cm}
\subsubsection{Dynamic ABC}
\vspace{0.2cm}
The estimates of  marginal inclusion probabilities $\pi_j(\epsilon)$ obtained with ABC Bayesian Forests  unavoidably depend on the level of approximation accuracy $\epsilon$.
The acceptance threshold $\epsilon$ can be difficult to determine in practice, because it has to accommodate random variation of data around  $f_{0}$ as well as the error   when approximating smooth surfaces $f_0$ with trees.  As $\epsilon\rightarrow 0$,  the approximations $\pi_j(\epsilon)$ will be more accurate, but the acceptance rate will be smaller. It is customary to pick $\epsilon$ as an empirical quantile of $\epsilon_m$ \citep{grelaud2009abc}, keeping only the top few closest samples. Rather than choosing one value $\epsilon$, we suggest a dynamic strategy by considering a sequence  of decreasing values $\epsilon_N>\epsilon_{N-1}>\dots>\epsilon_1>0$. By filtering out the ABC samples with stricter thresholds,
we track the evolution of each $\pi_j(\epsilon)$ as $\epsilon$ gets smaller and smaller. This gives us a dynamic plot that is similar in spirit to the  Spike-and-Slab LASSO \citep{rockova2016spike} or EMVS \citep{rovckova2014emvs} coefficient evolution plots. However, our plots depict approximations to posterior inclusion probabilities rather than coefficient magnitudes. Other strategies for selecting the threshold $\epsilon$ are discussed in \citep{sunnaaker2013approximate,marin2012approximate,csillery2010approximate}.

\vspace{-0.5cm}
\subsection{ABC Bayesian Forests in Action}\label{sec:demonstrate}
\vspace{0.2cm}
We  demonstrate the usefulness of ABC Bayesian Forests on the benchmark Friedman dataset \citep{friedman1991multivariate},  where the observations are generated from \eqref{model} with  $\sigma=1$ and
\begin{equation}\label{mean_friedman}
f_0(\x_i)= 10\,\sin(\pi \,x_{i1}\, x_{i2}) + 20\,(x_{i3} -0.5)^2 + 10\, x_{i4} + 5\, x_{i5},
\end{equation} 
where $x_i\in[0,1]^p$ are $iid$ from a uniform distribution on a unit  cube.  Because the outcome depends on $x_1,\dots,x_p$, the predictors $x_6,\dots,x_p$ are irrelevant, making it more challenging to find $f_0(\x)$. We begin by illustrating the basic features of ABC Bayesian Forests with $p=100$ and $n=500$, assuming the beta-binomial prior $\pi(\mS\C\theta)$ with $\theta\sim\mathcal{B}(1,1)$ (see Section \ref{sec:ABC}).
At the $m^{th}$ ABC iteration, we draw one posterior sample $f_{\mE,\B}^{m}$ after $100$ burnin iterations using the BART MCMC algorithm \citep{chipman2001practical} with $T=10$ trees.  We generate $M=1\,000$ ABC samples (with $s=n/2$) and we keep track of variables used in $f_{\mE,\B}^{m}$'s to estimate the marginal posterior inclusion probabilities $\pi_j(\epsilon)$. 
It is worth pointing out that unlike MCMC, ABC Bayesian Forests are embarrassingly parallel, making distributed implementations readily available.

Following the dynamic ABC strategy, we plot the estimates of posterior inclusion indicators $\pi_j(\epsilon)$ as a function of $\epsilon$ (Figure \ref{fig1}). The true signals are depicted in blue, while the noise covariates are in red. The estimated inclusion probabilities clearly segregate the active and non-active variables, even for large $\epsilon$ values. This is because BART itself performs variable selection to some degree, where not all variables in  $\mS_m$ end up contributing to $f_{\mE,\B}^m$. For small enough $\epsilon$, the inclusion probabilities of true signals eventually cross the $0.5$ threshold.
Based on the median probability model rule \citep{barbieri2004optimal}, one thereby selects the true model when  $\epsilon$ is sufficiently small.
Because the inclusion probabilities  get a bit unstable as $\epsilon$ gets smaller (they are obtained from smaller reference tables),
we  excluded the $10$ smallest $\epsilon$ values from the plot.

\begin{figure}[!t]
     \subfigure[$T=10$]{
\includegraphics[width=0.45\textwidth,height=0.4\textwidth]{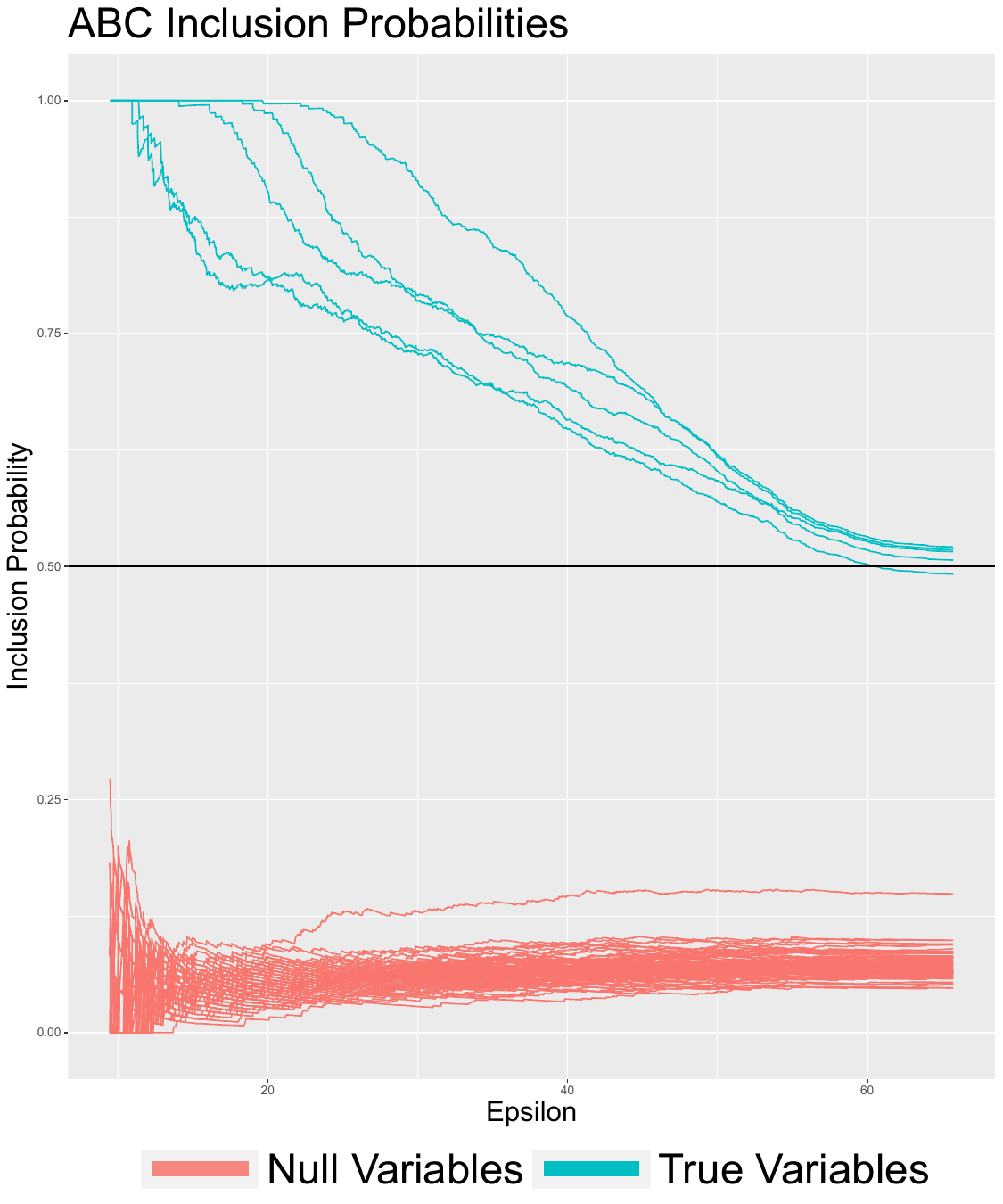} \label{fig1:dyna1}
}
\subfigure[Random Forests versus  Bayesian Forests]{
\includegraphics[width=0.45\textwidth,height=0.4\textwidth]{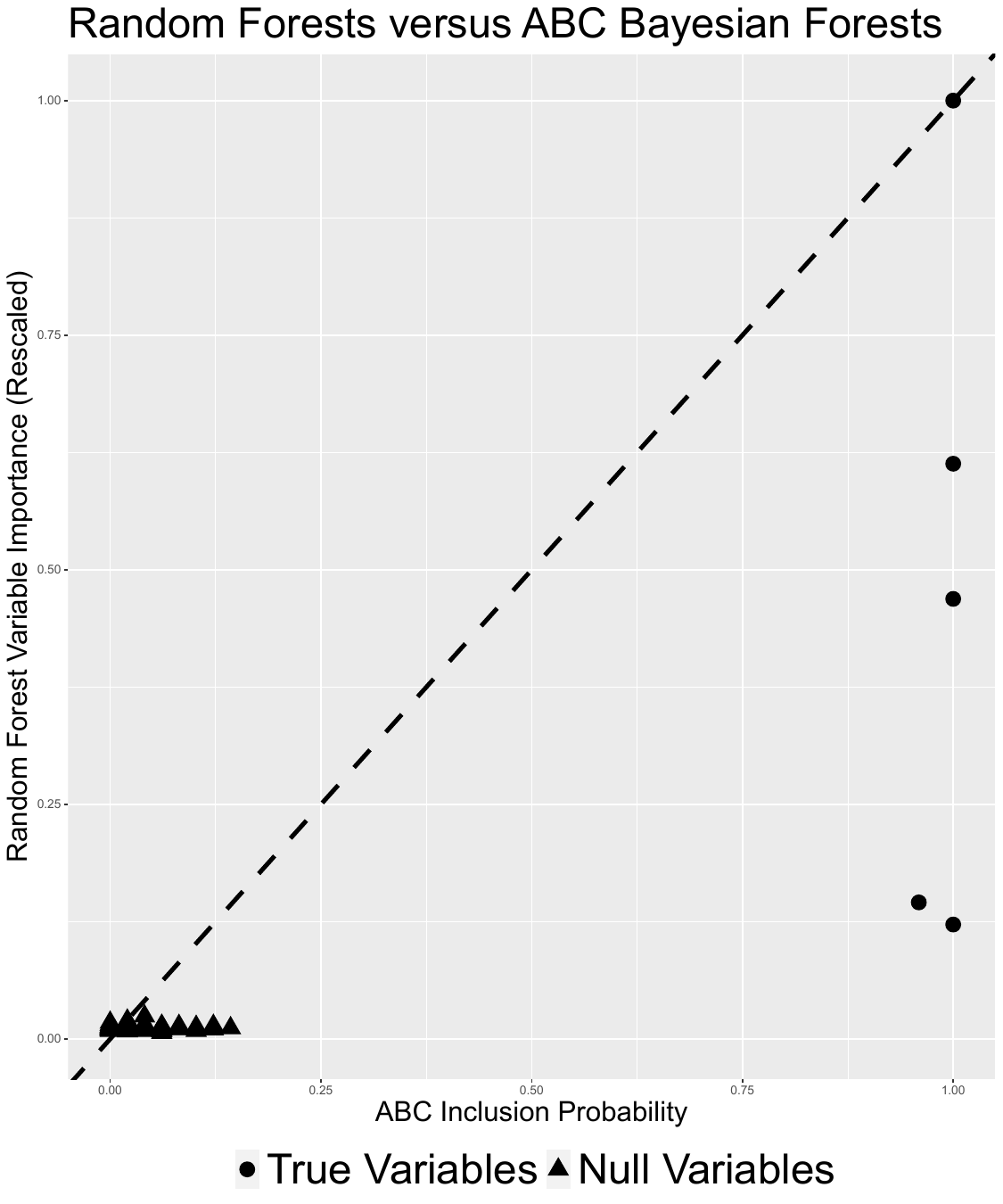}\label{fig1:dyna2}
}
    \caption{\small(Left) Dynamic ABC plots for evolving inclusion probabilities as $\epsilon$ gets smaller. (Right) Plot of $\pi_j(\epsilon)$ obtained with ABC Bayesian Forests ($\epsilon$ is   the $5\%$ quantile of $\epsilon_m$'s) and the variable importance measure   from Random Forests (rescaled to have a maximum at 1).}\label{fig1}
\end{figure}

We repeated the experiment with more trees ($T=50$) and a single tree ($T=1$). Using more trees, one still gets the separation between signal and noise. However, many more noisy covariates would be included by the MPM rule. This is  in accordance with  \citet{chipman2001practical} who state that BART can over-select with many trees.  With a single tree, on the other hand, one may miss some of the low-signal predictors, where deeper trees and more ABC iterations would be needed  to obtain a clearer separation. 

 In this simulation, we observe a curious empirical connection between $\pi_j(\epsilon)$, obtained with ABC Bayesian Forests (taking top $5\%$ ABC samples), and  rescaled variable importances obtained with Random Forests (RF). From  Figure \ref{fig1:dyna2}, we see that  the two measures largely agree, separating the signal coefficients (triangles) from the noise coefficients (dots). However, the RF measure is a bit more conservative, yielding smaller normalized importance scores for true signals.  While variable importance for RF is yet not understood theoretically, in the next section we provide conditions under which the posterior distribution  is consistent for variable selection.
\smallskip

\vspace{-0.5cm}
\section{Model-Free Variable Selection Consistency}\label{sec:consist}
\vspace{0.2cm}
In this section,  we develop large sample model selection theory for spike-and-forest priors.
As a jumping-off point, we first assume that $\alpha$ (the regularity of $f_0$) is known, where model selection essentially boils down to finding the active set $\mS_0$. Later in this section, we  investigate {\sl joint} model selection consistency, acknowledging uncertainty about $\mS_0$ and, at the same time, the regularity $\alpha$.

Several consistency results for non-parametric regression already exist \citep{zhu2015reinforcement,yang2017bayesian}. 
\citet{comminges2012tight} characterized tight conditions on $(n,p,q_0)$,  under which it is possible to consistently estimate the sparsity pattern in two regimes.
  For fixed $q_0$, consistency is attainable when $(\log p)/n\leq c$ for some $c>0$. When $q_0$ tends to infinity as $n\rightarrow\infty$, consistency is achievable when $c_1q_0+\log\log(p/q_0)-\log n\leq c_2$ for some $c_1,c_2>0$. Throughout this section, we will treat $q_0$ as fixed and show variable selection consistency when $q_0\log p\leq n^{q_0/(2\alpha+q_0)}$.
As an overture to our main result, we start with a simpler case when $T=1$ (a single tree) and  when $\alpha$ is  known. The full-fledged result for Bayesian forests and unknown $\alpha$ is presented in Section \ref{sec:consist:forest}. Throughout this section, we will assume   $\sigma^2=1$.

\vspace{-0.5cm}
\subsection{The Case of Known $\alpha$}\label{sec:known:alpha}
\vspace{0.2cm}
Spike-and-forest mixture priors are constructed in two steps by (1) first specifying a conditional prior $\Pi_\mS(f)$ on tree (ensemble) functions expressing a qualitative guess on $f_0$, and then (2)  attaching a prior weight $\pi(\mS)$ to each ``model" (i.e. subset) $\mS$. 
The posterior distribution $\Pi(f\C\Y^{(n)})$ can be viewed as a mixture of individual posteriors for various models $\mS$ with  weights given by posterior model probabilities $\Pi(\mS\C\Y^{(n)})$, i.e.
$$
\Pi(f\C\Y^{(n)})=\sum_{\mS}\Pi(\mS\C\Y^{(n)})\Pi_{\mS}(f\C\Y^{(n)}).
$$
Our aim is to establish ``model-free" variable selection consistency in the sense that
$$
\Pi(\mS=\mS_0\C\Y^{(n)})\rightarrow 1\quad \text{in $\P_{f_0}^{(n)}$-probability} \quad\text{as}\quad n\rightarrow\infty,
$$
where $\P_{f_0}^{(n)}$ is the distribution of $\Y^{(n)}$ under \eqref{model}.
The adjective ``model-free" merely refers to the fact that we are selecting subsets in a non-parametric regression environment without necessarily committing to a linear model.
We start by defining the model index set  $\G=\big\{\mS:\mS\subseteq\{1,\dots,p\}\big\}$, consisting of all $2^p$ variable subsets,
and we partition it into (a) the true model $\mS_0$, (b) models that {\sl overfit} $\G_{\mS\supset\mS_0}$ (i.e. supersets of the true subset $\mS_0$) and (c) models that {\sl underfit} $\G_{\mS\not\supset\mS_0}$ (i.e. models that miss at least one active covariate). 
Each model $\mS\in\G$ is accompanied by a convergence rate $\varepsilon_{n,\mS}$ that reflects the inherent difficulty of the estimation problem. For each model $\mS$ of size $|\mS|$, we define
\begin{equation}\label{minimax_rate}
\varepsilon_{n,\mS}=C_\varepsilon\,n^{-\alpha/(2\alpha+|\mS|)}\sqrt{\log n}\quad\text{for some}\quad C_\varepsilon>0,
\end{equation}
the $\|\cdot\|_n$-near-minimax  rate of estimation of a $|\mS|$-dimensional $\alpha$-smooth function. 

\vspace{-0.5cm}
\subsubsection{Prior Specification}
\vspace{0.2cm}
Prior distribution on the model index $\Pi(\mS)$ has to be chosen carefully for model selection consistency to hold  when $p>n$ \citep{moreno2015posterior}.
Traditional spike-and-slab priors introduce $\Pi(\mS)$ through a prior inclusion probability $\theta=\Pi(i\in\mS_0\C\theta)$, independently for each $i=1,\dots, p$. This prior mixing weight  is  often  endowed with a prior, such as  the uniform prior $\pi(\theta)=\mathcal{B}(1,1)$ \citep{scott2010bayes}, yielding a uniform prior on the model size, or the ``complexity prior"  $\pi(\theta)=\mathcal{B}(1,p^c)$ for $c>2$ \citep{castillo2012needles}, yielding an exponentially decaying prior on the model size.
We propose a  different approach, directly assigning a prior on model weights through
\begin{equation}\label{model_prior}
\pi(\mS)\propto {\e^{-C\,\left(n^{|\mS|/(2\alpha+|\mS|)}\log n\vee |\mS|\log p\right)}}
\end{equation}
where $C>0$ is a suitably large constant. When $|\mS|\log p\leq n^{|\mS|/(2\alpha+|\mS|)}$, this prior is proportional to $\e^{-C/C_\varepsilon^2\,n\,\varepsilon_{n,\mS}^2}$  and, as such, it puts more mass on models that yield faster rates convergence (similarly as in  Lember and van der Vaart (2007)). 
When $|\mS|\log p>n^{|\mS|/(2\alpha+|\mS|)}\log n$, the implied prior on the effective dimensionality 
$\pi(|\mS|)={p\choose |\mS|}\pi(\mS)$
 will be exponentially decaying in the sense that $\pi(|\mS|)\lesssim \e^{-(C-1)|\mS|\log p}$ for $C>1$. It was recently noted by \citet{castillo2018empirical} that the complexity prior  ``penalizes slightly more than necessary". With our prior specification \eqref{model_prior}, however, the exponential decay kicks in {\sl only} when $|\mS|$ is sufficiently large. 

Assuming that the level of smoothness $\alpha$ is known,  the optimal number of steps (i.e. tree bottom leaves $K$) needed to achieve the rate-optimal performance for estimating  $f_0$ should be of the order $n^{q_0/(2\alpha+q_0)}=1/C_\varepsilon^2\, n\,\varepsilon_{n,\mS_0}^2/\log n$ \citep{rockova2017posterior}.
For our toy setup with a known  $\alpha$, we thus assume a point-mass prior on $K$ with an atom near the optimal number of steps for each given $\mS$, i.e.
\begin{equation}\label{prior:K}
\pi(K\C\mS)=\1[K=K_{\mS}],\quad\text{where}\quad K_{\mS}=\lfloor C_K/C_\varepsilon^2\,n\,\varepsilon_{n,\mS}^2/\log n\rfloor
\end{equation}
for some $C_K>0$ such that $K_{\mS_0}=2^{q_0s}$ for some $s\in\N$. In Section \ref{sec:alpha_unknown}, we allow for more flexible trees with variable sizes.

\vspace{-0.5cm}
\subsubsection{Identifiability}
\vspace{0.2cm}
The active variables ought to be sufficiently relevant in order to make their identification possible.
To this end, we introduce a non-parametric signal strength assumption, making sure that  $f_0$  is not too flat in active directions \citep{yang2017bayesian,comminges2012tight}. 

We first introduce the notion of an approximation gap.
For any given model $\mS$, we denote with $\mF_\mS$ a set of approximating functions  (only single trees $f_{\mT,\b}$ with $K_{\mS}$ leaves for now) and
 define the approximation gap as follows:
\begin{equation}\label{gap}
\delta^{\mS}_n\equiv \inf\limits_{f_{\mT,\b}\in\mF_{\mS}}\|f_0-f_{\mT,\b}\|_n=\|f_0-f_{\wh\mT,\wh\b}^{\mS}\|_n,
\end{equation}
where $f_{\wh\mT,\wh\b}^{\mS}$ is the $\|\cdot\|_n$-projection of $f_0$ onto $\mF_\mS$.
For identifiability of $\mS_0$, we require that those models that miss one of the active covariates have a large separation gap.
\begin{definition}(Identifiability)\label{ass:identify}
We say that $\mS_0$ is $(f_0,\varepsilon)$-{\sl identifiable} if, for some $M>0$,
\begin{equation}\label{identify}
\inf_{i\in\mS_0} \delta^{\mS_0\backslash i}_n>2M\varepsilon.
\end{equation}
\end{definition}
We provide a more intuitive explanation of  \eqref{identify} in terms of directional variability of $f_0$. The best approximating tree $f_{\wh\mT,\wh\b}^{\mS}$ can be written as
$$
 f_{\wh\mT,\wh\b}^{\mS}(\x)=\sum_{k=1}^{K_\mS}\1(\x\in\wh\Omega_{k}^{\mS})\wh\beta_k \,\,\text{with}\,\, \wh \beta_k= {\bar f_0(\wh\Omega_k^{\mS})}\equiv\frac{1}{n(\wh\Omega_{k}^{\mS})}\sum_{\x_i\in\wh{\Omega}_{k}^{\mS}}f_0(\x_i),
$$
where
$\wh\mT=\{\wh\Omega_k^{\mS}\}_{k=1}^{K_\mS}$ is the {tree-shaped partition of the $\|\cdot\|_n$-projection of $f_0$ defined in \eqref{gap}}  with $K_\mS$ leaves and where $n(\wh\Omega_k^{\mS})=\sum_{i=1}^n\1(\x_i\in\wh\Omega_k^{\mS})\equiv n\,\mu(\wh\Omega_k^{\mS}).$
The separation gap in \eqref{gap} can be then re-written as
\begin{align*}
\delta^{\mS}_n
=\sqrt{\sum_{k=1}^{K_\mS}\mu(\wh\Omega_{k}^{\mS})V[f_0\C\wh\Omega_{k}^{\mS}]},
\end{align*}
where
$$
{V[f_0\C\wh\Omega_k^{\mS}]}\equiv\frac{1}{n(\wh\Omega_k^{\mS})}\sum_{\x_i\in\wh\Omega_k^{\mS}}\left(f_0(\x_i)-{\bar f_0(\wh\Omega_k^{\mS})}\right)^2
$$
is the local variability of $f_0$ inside $\wh\Omega_k^{\mS}$. Given this characterization,  \eqref{identify} will be satisfied, for instance, when variability of $f_0$ inside best approximating cells that miss an active direction is too large, i.e.
$
\inf_{i\in\mS_0}\inf\limits_{k} V[f_0\C\wh\Omega_k^{\mS_0\backslash i}]> 4M^2\,\varepsilon^2.
$
 
 {
 Our identifiability condition is a theoretical assumption on $f_0$ which indicates how large signal in each direction should be in order to be capturable. 
 It generalizes the more traditional sufficient ``beta-min conditions" \citep{castillo2015bayesian,zhao2006model}  for variable selection consistency (see Remark \ref{remark:identify}). Here,
we gauge the amount of signal  in terms of local variation in cells that {\em do not split} on an active covariate. Intuitively, if we do not split on $i\in\mS_0$, the 
``variation" of $f_0$ inside the cells of the best tree we can get without $i$ will be too large. The following example links our identifiability assumption with beta-min conditions.

\begin{example}
Assume for now that $p=2$ and that $f_0$ is linear, i.e. 
$$
f_0(\bm x_i)=a+bx_{i1}+cx_{i2}.
$$
Moreover, assume that $n=16$ predictor observations are  located on a regular grid  $\mathcal X=\{k/4: 1\leq k\leq 4\}\times \{j/4: 1\leq j\leq 4\}$, where $\times$ denotes the Cartesian product. 
Suppose $\mS_0=\{1,2\}$ and  set $\mS=\mS_0\backslash\{2\}=\{1\}$ and $K_{\mS}=2$. It can be verified that the partition $\wh\mT$ of the best approximating tree  that {\em does not} split on the covariate $x_2$ consists of two rectangles $\wh\Omega^\mS_1=[0,1/2)\times [0,1]$ and $\wh\Omega^\mS_2=[1/2,1]\times [0,1]$.
Then we have
$$
\bar f_0(\wh\Omega^\mS_1)=a+\frac{3}{2}\left(\frac{b}{4}\right)+ \frac{5}{2}\left(\frac{c}{4}\right)\quad\text{and}\quad 
\bar f_0(\wh\Omega^\mS_2)=a+\frac{7}{2}\left(\frac{b}{4}\right)+ \frac{5}{2}\left(\frac{c}{4}\right)
$$
and thereby
\begin{equation}\label{eq:delta}
(\delta_m^{\mS})^2=V (f_0 | \wh\Omega^\mS_1)=V (f_0 | \wh\Omega^\mS_2)=\frac{1}{4}\frac{b^2}{16}+\frac{5}{4}\frac{c^2}{16}.
\end{equation}
From the expression \eqref{eq:delta} we can immediately see the connection to the beta-min conditions.
When the signal in the direction of $x_2$ is large enough, i.e. $c>16/\sqrt{5}M\varepsilon$, our identifiability condition will be satisfied.  
\end{example}
}
 
{The second  sufficient condition needed for methods such as the LASSO to fully recover $\mS_0$ is ``irrepresentability" \citep{zhao2006model, van2009conditions}. This condition restricts the amount of correlation between (active and non-active) covariates by imposing a regularization constraint on the magnitudes of regression coefficients of the inactive predictors onto the active ones.}
Here, we generalize the notion of irrepresentability to the non-parametric setup. 
Consider an underfitting model $\mS=\mS_1\cup\mS_2\not\supset\mS_0$, where $\mS_1\subset\mS_0$ are true positives and $\mS_2$ is a possibly empty set of  false positives, i.e. $\mS_2\cap\mS_0=\emptyset$. We define 
\begin{equation}\label{rho}
\rho_n^\mS\equiv
\frac{1}{n}\sum_{i=1}^n[f_0(\x_i)-f^{\mS_1}_{\wh\mT,\wh\b}(\x_i)][f^{\mS}_{\wh\mT,\wh\b}(\x_i)-f^{\mS_1}_{\wh\mT,\wh\b}(\x_i)],
\end{equation}
the sample covariance  between the surplus signals in $f_0$ and $f^{\mS}_{\wh\mT,\wh\b}$ obtained by removing the effect of $f^{\mS_1}_{\wh\mT,\wh\b}$. This quantity will be large if
noise covariates inside $\mS_2$ can  compensate for the missed true covariates in $\mS_0\backslash\mS_1$, i.e. when the true and fake covariates are strongly correlated.  To obviate this substitution effect, we introduce the following  { nonparametric ``irrepresentability"condition.
Similarly as in \cite{zhao2006model}, we require that 
``the total amount of an irrelevant covariate represented by the covariates in the true model" is small.
}
 \begin{definition}(Irrepresentability)\label{ass:irrep}
We say that $\varepsilon$-{\sl  irrepresentability} holds for $f_0$ and $\mS_0$ if, for some $M>0$, we have
$
\sup_{\mS\not\supset\mS_0}|\rho_n^\mS|<\frac{M}{2}\varepsilon,
$
where $\rho_n^{\mS}$ was defined in \eqref{rho}.
 \end{definition}

It follows from Lemma \ref{lemma:approx} (Appendix) that under the irrepresentability and identifiability conditions (Definition \ref{ass:identify} and \ref{ass:irrep}), we obtain
\begin{equation}\label{gap_final}
\inf\limits_{\mS\not\supset\mS_0}\inf_{f_{\mT,\b}\in\mF_\mS}\|f_{\mT,\b}-f_0\|_n>M\,\varepsilon.
\end{equation}
This condition essentially states that {\sl all} models that miss {\sl at least one} active covariate (i.e. not only subsets of the true model) have a large separation gap.

The following theorem characterizes variable selection consistency of spike-and-tree posterior distributions. Namely, the posterior distribution over the model index is shown to concentrate on the true model $\mS_0$.
One additional assumption is needed to make sure that the (fixed) design $\mX=\{\x_1,\dots,\x_n\}$  is sufficiently regular.
\citet{rockova2017posterior} define the notion of a fixed $\mS_0$-regular design in terms of cell diameters of a $k$-$d$ tree partition  (Definition 3.3). This assumption essentially excludes outliers, making sure that the data cloud is spread evenly in active directions (while permitting correlation between covariates). 

\begin{theorem}\label{thm:consist}
Assume $f_0\in\Ha_p\cap\mC(\mS_0)$ for some $\alpha\in(0,1]$ and $\mS_0\subset\{1,\dots,p\}$ with $q_0=|\mS_0|$ and $\|f_0\|_\infty\lesssim B$.  Denote with 
$\wt\varepsilon_n= C_\varepsilon\,n^{-\alpha/(2\alpha+q_n)}\sqrt{\log n}$,
 where $q_n=C_q  \lceil n\,\varepsilon_{n,\mS_0}^2/\log p\rceil$ for some $C_q>0$, and  assume  $q_0\log p\leq n^{q_0/(2\alpha+q_0)}$ with $2\leq q_0=\mathcal{O}(1)$ as $n\rightarrow\infty$. Assume
  that (a) $\mS_0$ is $(f_0,\wt\varepsilon_n)$-{\sl identifiable},  (b) $\wt\varepsilon_n$-irrepresentability holds and that (c) the design $\mX$ is $\mS_0$-regular. 
Under the {\sl spike-and-tree} prior comprising   (with $T=1$)
\eqref{prior:tree},\eqref{prior:beta2},\eqref{model_prior} with $C>2$ and  \eqref{prior:K}, we  have
 $$
 \Pi[\mS=\mS_0\C \Y^{(n)}]\rightarrow 1\quad\text{in $\P_{f_0}^{(n)}$-probability as $n\rightarrow\infty$}.
 $$
\end{theorem}
\proof Section \ref{sec:proof:tm1}

{
\begin{remark}\label{remark:identify}
The assumption  of $(f_0,\wt\varepsilon_n)$-identifiability pertains to the more traditional sufficient beta-min conditions for variable selection consistency in sparse high-dimensional models. For example,  \cite{castillo2015bayesian} in their Corollary 1 require that
$
\min_{i\in \mS_0} |\beta_i^0|\geq M\sqrt{\frac{{q_0\log p}}{n}},
$ for some ``large enough constant" $M>0$ that depends on the compatibility number  (see e.g. Definition 2.1 in \cite{castillo2015bayesian} of the design matrix $X$ (rescaled to have an $\|\cdot\|_2$ norm $\sqrt{n}$).  Our identifiability threshold also depends on the rate of convergence $\varepsilon_n$ (similarly as in \cite{castillo2015bayesian}). However, unlike in the linear models we measure the signal strength in a non-parametric way.
Lastly, note that the identifiability gap $\wt\varepsilon_n$ in Theorem \ref{thm:consist} is a bit larger than the near-minimax rate $\varepsilon_{n,\mS_0}$.  This requirement  will be relaxed in the next section, where  $\alpha$ will be treated as  unknown.
\end{remark}
}

For $iid$ models, \citet{ghosal2008nonparametric} considered the problem of nonparametric Bayesian model selection and averaging and characterized conditions under which the posterior achieves adaptive rates of convergence. The authors also study the posterior distribution of the model index, showing that  it puts a negligible weight on models that are bigger than the optimal one.   \citet{yang2017bayesian} characterized similar conditions for the non-$iid$ case, see Section \ref{sec:proof:tm1} for more details.

{
\begin{remark}(Theory for ABC)
It is worth pointing out that Theorem \ref{thm:consist} is obtained for the {\em actual} posterior $\pi(\mS\C\Y^{(n)})$, not the ABC posterior. Theory for  ABC recently started emerging with the first results focussing on ABC bias \citep{barber2015rate}, consistency and asymptotic normality \citep{martin2014approximate, frazier2018asymptotic, frazier2020model} and on convergence of the posterior mean \citep{li2018convergence}.  For our non-parametric regression scenario, we can conclude (variable selection) consistency for ABC Bayesian forests under the assumption that the residual variance $\sigma^2$ decreases with the sample size (as is typical in the Gaussian sequence model). In particular, Theorem \ref{thm:abc1} in Supplemental Materials  (Section \ref{sec:theory_ABC})  shows that the ABC posterior concentrates at the rate  $\lambda_n=4\epsilon^T_n/3+1/\sqrt{n}$, where $\epsilon^T_n=\sqrt{2\log n/n}$ is the ABC tolerance level.
This result implies that the ABC posterior will {\em not} reward underfitting model as long as our identifiability and irrepresentability conditions are satisfied with $\varepsilon=\lambda_n$. Regarding over-fitting models,  an ABC analogue of Lemma 1.1 (Section 1.1.2 in Supplemental Materials)  implies that the ABC posterior probability of over-fitting models goes to zero, which concludes variable selection consistency of a (naive) ABC method. These considerations can be extended to ABC Bayesian Forests  with data splitting using the empirical expected posterior prior justification  in \eqref{eq:expected_prior}. More details are in Supplemental Materials (Section \ref{sec:theory_ABC}). 
\end{remark}}

{
\begin{remark}(Consistency of the Median Probability Model)\label{rem:mpm_consist}
In Section \ref{sec:demonstrate}, we used the median probability model rule which may not the same as the highest-posterior model whose consistency we have shown in Theorem \ref{thm:consist}. However, even when $p\rightarrow\infty$  it can be verified (as in Corollary 4.1 in \cite{narisetty2014bayesian}) that the median probability model is {\em also} consistent  under the same assumptions as Theorem \ref{thm:consist}.
In particular, $\P_{f_0}^{(n)}[\cap_{i=1}^{p} E_i]\rightarrow 1$ as $n\rightarrow\infty$ where $E_i=\{\Pi(\gamma_i=\gamma_i^0\C\Y^{(n)})>0.5\}$ and where $\gamma_i=\mathbb I(i\in\mS)$ are binary inclusion indicators and $\gamma_i^0=\mathbb I(i\in\mS_0)$.
\end{remark}}

\vspace{-0.7cm}
\subsection{The Case of Unknown $\alpha$}\label{sec:alpha_unknown}
\vspace{0.2cm}
The fact that the level $\alpha$ has to be known for the consistency to hold makes the result in Theorem \ref{thm:consist} somewhat theoretical.
In this section, we provide a joint consistency result for the unknown regularity level $K$ and, at the same time, the unknown subset $\mS_0$.
 Finding the optimal regularity level $K$, given $\mS_0$,  is a model selection problem of independent interest \citep{lafferty2001iterative}. Here, we  acknowledge uncertainty about {\sl both} $K$ and $\mS_0$ by assigning a joint prior distribution on $(K,\mS)$. Namely, we consider an analogue of \eqref{model_prior}, where  $n^{|\mS|/(2\alpha+|\mS|)}$  is now replaced with $K\log n$ (according to \eqref{prior:K}), i.e.
\begin{equation}\label{prior:joint}
\pi(K,\mS)\propto \e^{-C(K\log n\, \vee\,  |\mS|\log p)}\quad\text{for}\quad 1\leq K\leq n\quad\text{and}\quad \mS\subseteq\{1,\dots,p\}.
\end{equation}
This prior penalizes models with too many splits or too many covariates. We now  regard each model as a {\sl pair of indices} $(K,\mS)$, where
the ``true" model is characterized by $\bG_0=(K_{\mS_0},\mS_0)$ with $K_{\mS_0}$ defined in \eqref{prior:K}.
Again, we partition the model index set $\bG=\{(K,\mS):\mS\subseteq\{1,\dots,p\},1\leq K\leq n\}$ into (a) the true model $\bG_0$, (b)
 models that underfit $\bG_{\{\mS\not\supset\mS_0\}\cup \{K< K_{\mS_0}\}}$ (i.e. miss at least one covariate {
 \sl or} use less than the optimal number of splits),
and (c)  models that overfit  $\bG_{\{\mS\supset\mS_0\}\cap \{K\geq K_{\mS_0}\}}$ (i.e. use too many variables and splits).

We combine the identifiability and irrepresentability conditions into one as follows:
\begin{equation}\label{gap_new}
\inf\limits_{\{\mS\not\supset\mS_0\}\cup \{K< K_{\mS_0}\}}\inf_{f_{\mT,\b}\in\mF_\mS(K)}\|f_{\mT,\b}-f_0\|_n>M\,\varepsilon_{n,\mS_0}
\end{equation}
for some $M>1$, where $\mF_\mS(K)$ consists of all trees with $K$ bottom leaves and splitting variables $\mS$. 
This condition is an analogue of \eqref{gap_final}, essentially stating that one cannot approximate  $f_0$ with an error  smaller than  a multiple of the near-minimax rate  using underfitting models.

\begin{theorem}\label{thm:consist2}
Assume $f_0\in\Ha_p\cap\mC(\mS_0)$ for some $\alpha\in(0,1]$ and $\mS_0\subset\{1,\dots,p\}$ such that $|\mS_0|=q_0$ and $\|f_0\|_\infty\lesssim B$.  
Assume  $q_0\log p\leq n^{q_0/(2\alpha+q_0)}$ and $2\leq q_0=\mathcal{O}(1)$ as $n\rightarrow\infty$. Furthermore, assume
  that the design $\mX$ is $\mS_0$-regular and that  \eqref{gap_new} holds.
Under the {\sl spike-and-tree} prior comprising  (with $T=1$)
  \eqref{prior:tree}, \eqref{prior:beta2} and \eqref{prior:joint} for $C>3$, we  have
 $$
 \Pi\left[\{\mS=\mS_0\}\,\cap\, \{K_{\mS_0}\leq K\leq K_n\} \,\Big|\, \Y^{(n)}\right]\rightarrow 1\quad\text{in $\P_{f_0}^{(n)}$-probability as $n\rightarrow\infty$},
 $$
 where $K_{\mS_0}$ was defined in \eqref{prior:K} and  $K_n=\lceil\bar C\,n\,\varepsilon_{n,\mS_0}^2/\log n\rceil$ for some $\bar C>C_K/C_\varepsilon^2$.
\end{theorem}
\proof Section \ref{sec:proof:thm2}

Note that both $K_{\mS_0}$ and $K_n$ are of the same  (optimal) order, where the marginal posterior distribution  $\Pi(K\C\Y^{(n)})$ squeezes inside these two quantities as $n\rightarrow\infty$.
\citet{lafferty2001iterative} provide a similar result for their RODEO method,  without the variable selection consistency part. \citet{yang2017bayesian} also provide a similar result for Gaussian processes,  without the regularity selection consistency part. Here,  we characterize {\sl joint}  consistency for both subset and regularity model selection. 

\vspace{-0.5cm}
\subsection{Variable Selection Consistency with Bayesian Forests}\label{sec:consist:forest}
\vspace{0.2cm}
Finally, we provide a variant of Theorem \ref{thm:consist2} for tree ensembles.
Each Bayesian forest (i.e. additive regression tree) model is characterized by a triplet $(\mS,T,\bm K)$, where $\mS$ is the active variable subset, $T\in\N$ is the number of trees and $\bm{K}=(K^1,\dots,K^T)'\in\N^T$ is a vector of the  bottom leave counts for the $T$ trees. Rate-optimality of Bayesian forests can be achieved for a wide variety of  priors, ranging from many weak learners (large $T$  and small $K^t$'s) to a few strong learners (small $T$ and large $K^t$'s) \citep{rockova2017posterior}. The optimality requirement is that   the {\sl total} number of leaves in the ensemble $\sum_{t=1}^TK^t$ behaves like $K_{\mS_0}$, defined earlier in \eqref{prior:K}.

 We thereby define  models in terms of equivalence classes  rather than individual triplets $(\mS,T,\bm K)$. We construct each equivalence class $E(Z)$ by  combining ensembles with the same number $Z$ of total leaves, i.e.
\begin{equation}\label{equivalance_class}
E(Z)=\bigcup_{T=1}^{\min\{Z,n\}}\left\{\bm K\in\N^T:\sum_{t=1}^TK^t=Z\right\}.
\end{equation}
The cardinality of $E(Z)$, denoted with $\Delta(E(Z))$, satisfies 
$\Delta(E(Z))\leq  Z!\, p(Z),
$ where $p(Z)$ is the partitioning number (i.e. the number of ways one can write $Z$ as a sum of positive integers). 
The ``true" model  $\bG_0=(\mS_0,E(K_{\mS_0}))$ consists of an equivalence class of forests that split on variables inside $\mS_0$ with a total number of  $K_{\mS_0}$ leaves.
Similarly as before, we define underfitting model classes $\bG_{\{\mS\not\supset\mS_0\}\cup \{E(Z): Z< K_{\mS_0}\}}$ and overfitting model classes $\bG_{\{\mS\supset\mS_0\}\cap \{E(Z): Z\geq K_{\mS_0}\}}$. Regarding the prior on $T$, similarly as \citet{rockova2017posterior}, we consider 
\begin{equation}\label{prior:T}
\pi(T)\propto \e^{-C_T\, T},\quad T=1,\dots, n,\quad\text{for}\quad C_T>0.
\end{equation}
Given $T$, we assign a joint prior over $\mS_0$ and  $\bm{K}\in\N^T$ as follows:
\begin{equation}\label{prior:joint2}
\pi(\mS,\bm K\C T)\propto \e^{-C\,\max\left\{|\mS|\log p\,;\, \sum_{t=1}^TK^t\log n\right\}} \quad\text{for}\quad C>1.
\end{equation}
We conclude this section with a model selection consistency result for Bayesian forests under the following identifiability condition
\begin{equation}\label{gap_new2}
\inf\limits_{\{\mS\not\supset\mS_0\}\cup \{E(Z): Z< K_{\mS_0}\}}\inf_{f_{\mE,\B}\in\mF_\mS(\bm K)}\|f_{\mE,\B}-f_0\|_n> M\,\varepsilon_{n,\mS_0},
\end{equation}
where $\mF_\mS(\bm K)$ denotes all forests $f_{\mE,\B}$ that split on variables $\mS$ and consist of $T$ trees with $\bm K=(K^1,\dots, K^T)'$ bottom leaves.

\begin{theorem}\label{thm:consist3}
Assume $f_0\in\Ha_p\cap\mC(\mS_0)$ for some $\alpha\in(0,1]$ and $\mS_0\subset\{1,\dots,p\}$ such that $|\mS_0|=q_0$ and $\|f_0\|_\infty\lesssim B$.  Assume  $q_0\log p\leq n^{q_0/(2\alpha+q_0)}$, where $2\leq q_0=\mathcal{O}(1)$ as $n\rightarrow\infty$. Furthermore, assume
  that the design is $\mS_0$-regular and that \eqref{gap_new2} holds.
Under the {\sl spike-and-forest} prior comprising  
 \eqref{prior:tree}, \eqref{prior:beta2}, \eqref{prior:T} and \eqref{prior:joint2}, we  have
 $$
 \Pi\left[\{\mS=\mS_0\}\,\cap\, \left\{K_{\mS_0}\leq \sum_{t=1}^TK^t\leq K_n\right\} \,\Big|\, \Y^{(n)}\right]\rightarrow 1\quad\text{in $\P_{f_0}^{(n)}$-probability as $n\rightarrow\infty$},
 $$
 where $K_{\mS_0}$ was defined in \eqref{prior:K}  and  $K_n=\lceil\bar C\,n\,\varepsilon_{n,\mS}^2/\log n\rceil$ for some $\bar C>C_K/C_\varepsilon^2.$
\end{theorem}
\proof Section \ref{sec:proof:tm3}


\vspace{-0.3cm}
\section{Simulation Study}
\vspace{0.2cm}
{
We evaluate the performance of ABC Bayesian Forests on simulated data. We consider the following performance criteria: Precision $= 1-\text{FDP}= \frac{\text{TP}}{\text{TP}+\text{FP}}$, Power $= \frac{\text{TP}}{\text{TP}+\text{FN}}$(defined as the proportion of true signals discovered as such), 
{Hamming Distance (HD)= FP+FN} (where FP and FN denotes the number of false positives and false negatives, respectively) 
and the area under the ROC curve  (AUC). 
Traditionally, AUC assesses how well a classification method can differentiate between two classes in the absence of a clear decision boundary.  We use this criterion to assess variable importance since many of the considered selection methods are based on an importance measure and, as such,  do not have a clear decision boundary.

The synthetic data are generated from the model  \eqref{model}, where  $\bm x_i$'s for $i=1,\ldots, n$ are drawn independently from $N_p(0,\Sigma)$ with $\Sigma=(\rho_{ij})_{i,j=1}^{p,p}$. We make our comparisons under different combinations of $f_0$, $\sigma$ and $\Sigma$. 
In particular, we consider a relatively large noise level with  $\sigma=5$ ($\sigma=\sqrt 5$ for the linear setup) and
\begin{enumerate}
\item  medium equi-correlation $\rho_{ij}=0.5$ for $i\neq j$ with $\rho_{ii}=1$,
\item  high auto-correlation $\rho_{ij}=0.9^{\abs{i-j}}$.
\end{enumerate}
Regarding the mean function $f_0$, we consider four choices: (1) a linear setup with $f_0(\x_i)=x_{i1} +2 x_{i2}+3x_{i3}-2x_{i4}-x_{i5}$; (2) the Friedman setup as described in \eqref{mean_friedman}; 
{(3) a CART (tree-based) function $f_0(\x_i)$ generated from the first 5 covariates  using the \texttt{rpart} function in R; (4)  a simulated example from \citet{liang2018bayesian} (denoted with LLS hereafter) with $f_0(\x_i)=\frac{10x_{i2}}{1+x_{i1}^2}+5\sin (x_{i3}x_{i4}+2x_{i5})$.
For the auto-correlation case, we  permuted the covariates so that signals are not next to each other.

\begin{figure}[!t] 
\begin{center}
\includegraphics[width=0.9\textwidth,height=0.7\textwidth]{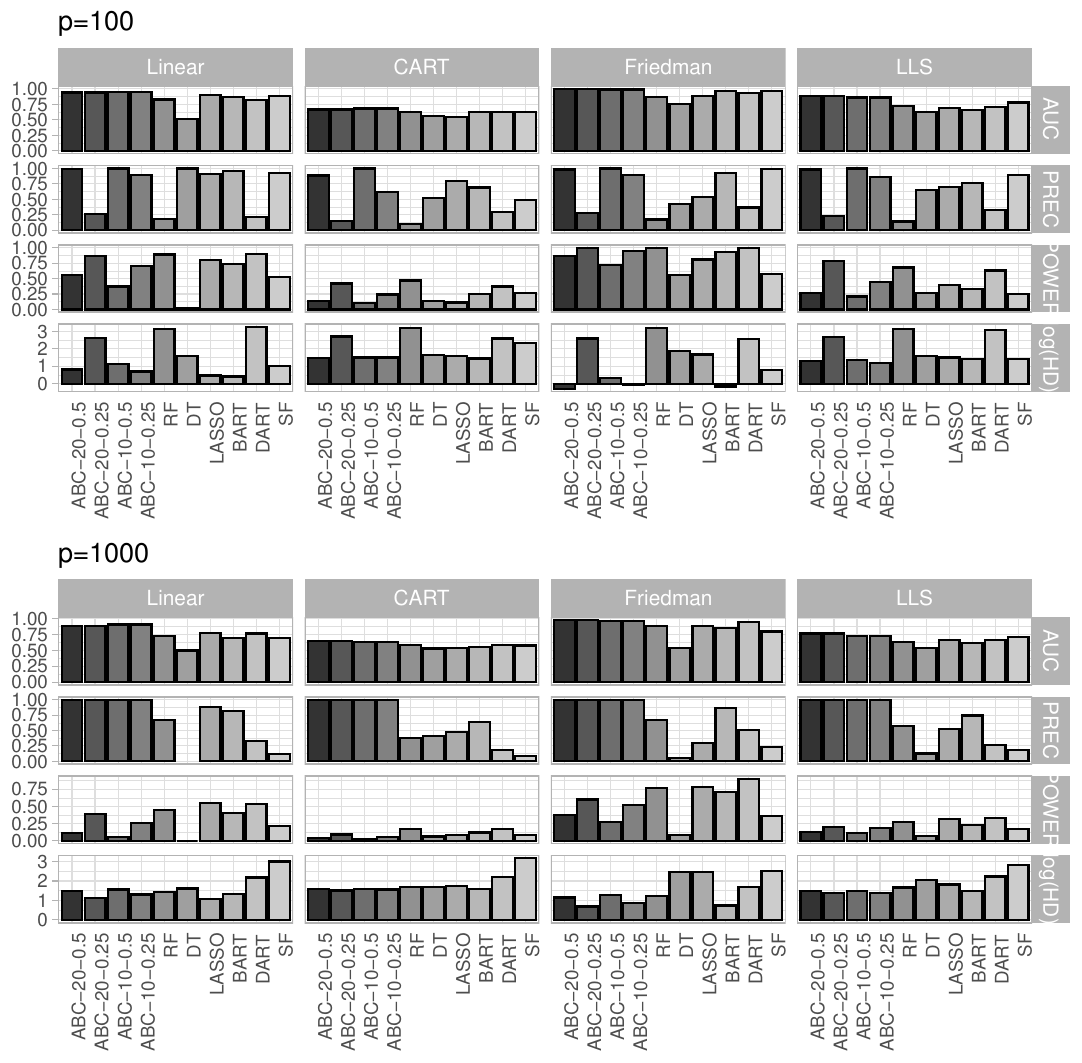}
\end{center}  
\noindent
\caption{\small Average variable selection performance under equicorrelation  $\rho_{ij}=0.5$ over $20$ simulations. Each panel corresponds to a different dimension $p\in\{100, 1000\}$. Each row reports a different statistic:  AUC is the area under the ROC curve, PREC $= 1-\text{FDP}= \frac{\text{TP}}{\text{TP}+\text{FP}}$, POWER $= \frac{\text{TP}}{\text{TP}+\text{FN}}$, $\log(\text{HD})=\log(\text{FP}+\text{FN})$.  ABC is run for $T\in\{10,20\}$ and cutoff $\in\{0.5,0.25\}$.
 Each column indicates a different data generating process. 
 } \label{fig:med_rho}
\end{figure}

 \begin{figure}[!t] 
\begin{center}
\includegraphics[width=0.9\textwidth,height=0.7\textwidth]{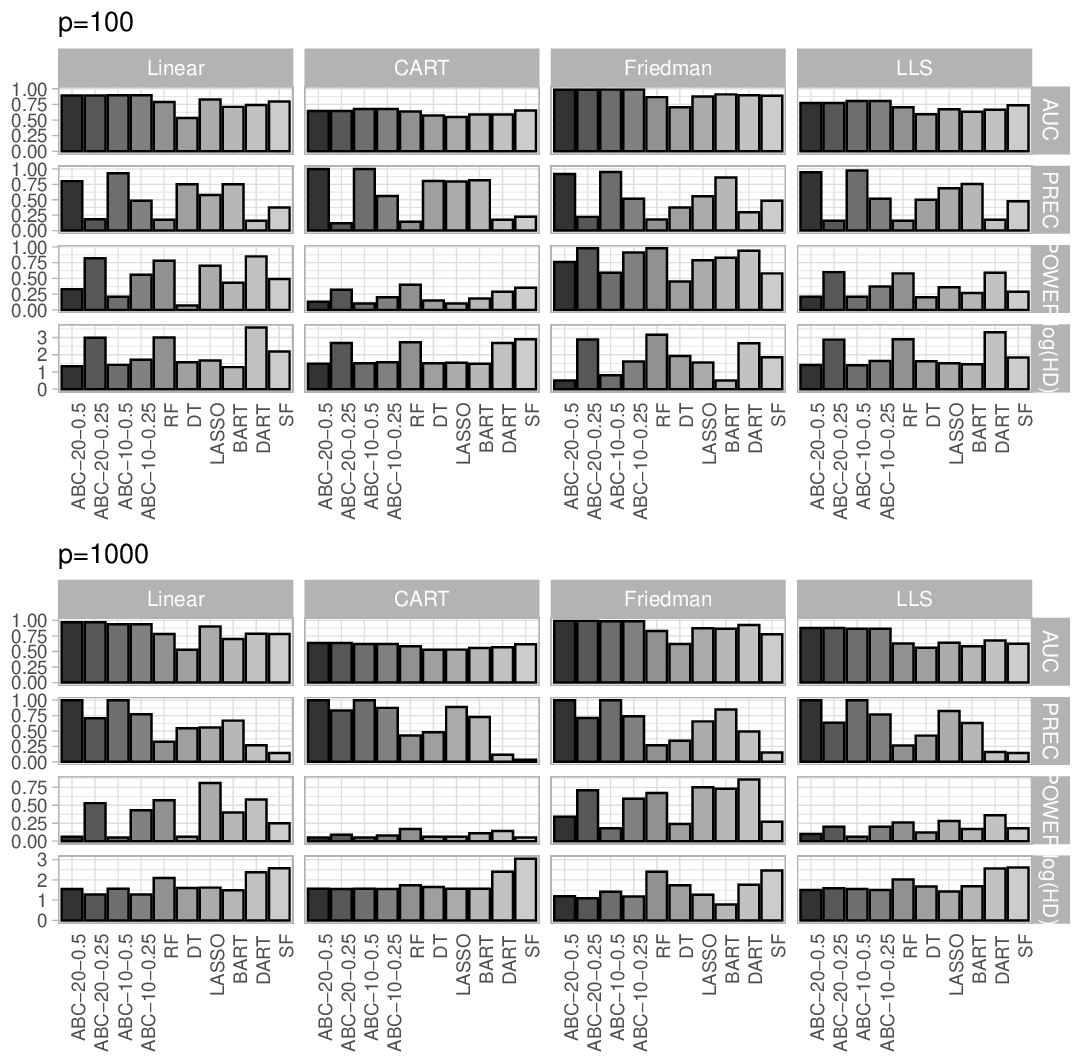}
\end{center}  
\noindent
\caption{\small Average variable selection performance under autocorrelation  $\rho_{ij}=0.9^{|i-j|}$ over $10$ simulations. Each panel corresponds to a different dimension $p\in\{100, 1000\}$. Each row reports a different statistics:  AUC is the area under the ROC curve, PREC $= 1-\text{FDP}= \frac{\text{TP}}{\text{TP}+\text{FP}}$, POWER $= \frac{\text{TP}}{\text{TP}+\text{FN}}$, $\log(\text{HD})=\log(\text{FP}+\text{FN})$. ABC is run for $T\in\{10,20\}$ and cutoff $\in\{0.5,0.25\}$.  
 Each column indicates a different data generating process. } \label{fig:high_rho}
\end{figure}

For each combination of settings, we repeat our simulation over $20$ different datasets assuming $n=500$ and $p\in\{100,1\,000\}$.  
We compare ABC Bayesian Forests with Random Forests (RF), 
Dynamic Trees (DT) of \citet{taddy2011dynamic}, BART \citep{chipman2010bart}, DART of  \citet{linero2018bayesian}, LASSO and Spike-and-Forests (the MCMC counterpart of ABC Bayesian Forests outlined in Section \ref{sec:mcmc} of the Supplemental Materials). ABC Bayesian Forests are trained with $M = 1\, 000$ ABC samples, where only a fraction of ABC samples (top 10\%) are kept in the reference table. 
The prior  $\pi(\mS)$ is the usual beta-binomial prior with $\theta\sim\mathcal{B}(1,1)$.   
Inside each ABC step, we sample a subset of size $s=n/2$ and draw a tree ensemble using the default Bayesian CART prior \citep{chipman1998bayesian} and $T\in\{10,20\}$ trees. 
{For each ABC sample, we draw the last BART sample after $B=200$ burnin MCMC iterations.} 
 A sensitivity analysis to the choice $s, T, B$ and $M$ is reported in the Supplemental Materials (Section  4). Two versions of BART (without ABC) were deployed using the  \texttt{R} package \texttt{BART}: (1) the standard BART from \citet{chipman2010bart} with  $T=20$ (as recommended in \citet{bleich2014variable}), and (2) the sparse version  DART of  \citet{linero2018bayesian} with a Dirichlet prior (\texttt{sparse=TRUE, a=0.5, b=1}) with $T=200$. Both versions are run with  $10\,000$ MCMC samples after $10\,000$ burn-in.  For LASSO, we use the \texttt{glmnet} package in \texttt{R} \citep{friedman2010regularization} using  the $1$-se rule to select the penalty $\lambda$. For Random Forests, we deploy the \texttt{randomForest} package in \texttt{R} \citep{liaw2002classification} using the default number of $500$ trees where variable importance is based on the difference in predictions (with and without each covariate) in out-of-bag samples. 

{To select variables with random forests, there are at least three commonly used strategies: (1) Recursive Feature Elimination (RFE)  implemented in the \texttt{caret} package with 5-fold cross-validation (as suggested in \citet{linero2018bayesian}); (2) truncating importance at  the $1-\alpha$ quantile of a standard normal distribution (as suggested by \citet{breiman2013online}); (3) truncating importance at the Bonferroni-corrected $(1-\alpha/p)$ quantile of a standard normal distribution \citep{bleich2014variable}. We report  the third method, which was seen to perform the best.  For BART and DART, we select those variables which have been split on inside a forest at least once on average. {Alternative strategies based on truncating inclusion probabilities \citep{linero2018bayesian} using data-adaptive thresholds \citep{bleich2014variable} did not perform better, in general.}  {For ABC, we report results for  two selection thresholds $0.5$ and $0.25$. For Spike-and-Forest (SF), we report the median probability model. }

{
The performance comparisons for variable selection are summarized in Figure \ref{fig:med_rho} (equi-correlation $\rho_{ij}=0.5$) and Figure \ref{fig:high_rho} (autocorrelation $\rho_{ij}=0.9^{|i-j|}$). These figures show that ABC has an advantage in terms of AUC, suggesting that ABC can rank variables more efficiently. While RF tend to have a higher power, they are plagued with false discoveries (i.e. smaller precision).
ABC Bayesian Forests, on the other hand, are seen to yield fewer false discoveries (i.e. higher precision) relative to the other procedures. The ABC threshold $0.5$ yields higher precision whereas $0.25$ yields higher power. }

}

\begin{table}
\caption{\label{tab:oos_mse} Average out-of-sample mean squared prediction error over $20$ independent validation datasets. 
ABC1 denotes predictions using ABC samples $f_{\mS,\bm{B}}^m$ and ABC2 uses ABC variable selection and runs BART ($T=200$) on the selected subset. $T$ designates the number of trees and $c$ is the selection threshold. The best performing method for each row is denoted in bold. }
\centering
\scalebox{0.6}{
\begin{tabular}{l | *{11}{c}}
\toprule\toprule
   &\bf  ABC2 &\bf  ABC1 & \bf ABC1  &\bf  ABC2 &\bf  ABC1 &\bf  ABC1 &\bf  RF  &\bf  RLT  & \bf DT    & \bf BART &\bf  DART \\
      & $T=20$ & $T=20,c=0.5$ & $T=20,c=0.25$ & $T=10$ & $T=10,c=0.5$ & $T=10,c=0.25$ &  &&&& \\
            &   \multicolumn{11}{c}{  \cellcolor[gray]{0.8}\bf  Equi-correlation  $ \rho_{ij}=0.5$ for $i\neq j$}  \\
    & \multicolumn{11}{c}{ \cellcolor[gray]{0.9} Linear }  \\ 
 $p=100$  & 5.56   & 5.58       & 5.84        & 5.60   & 5.84       & 5.55        & 5.63 & 5.45 & 5.92    & 5.49  & \bf{5.40}  \\
$p=1\,000$ & 5.79   & 6.15       & 5.73        & 5.86   & 6.28       & 5.95        & 5.83 & 5.70 & 6.04   & 5.82  & \bf{5.62}  \\
[0.05in] 
    & \multicolumn{11}{c}{\cellcolor[gray]{0.9} CART}  \\ 
$p=100$  & 34.21  & 34.63      & 37.19       & \bf{34.00}  & 36.10      & 35.81       & 34.21 & 34.64 & 34.61  & 35.48 & 35.57  \\
$p=1\,000$ & 32.00  & 34.27      & 35.72       & \bf{31.99}  & 33.93      & 33.17       & 32.30 & 32.40 & 33.08  & 33.77 & 34.04  \\
[0.05in] 
    & \multicolumn{11}{c}{\cellcolor[gray]{0.9} Friedman}  \\        
$p=100$  & 30.32  & 29.28      & 31.59       & 30.52  & 30.30      & \bf{29.03}       & 31.84 & 30.17 & 41.41  & 31.31 & \bf{29.03}  \\
$p=1\,000$ & 33.14  & 35.97      & 31.54       & 33.54  & 38.42      & 32.71       & 34.35 & 32.22 & 45.69  & 32.99 & \bf{29.42} \\
[0.05in] 
    & \multicolumn{11}{c}{\cellcolor[gray]{0.9} LLS}  \\    
$p=100$    & \bf{26.23}  & 27.00      & 28.70       & 26.25  & 26.90      & 27.36       & 26.80 & 26.46 & 28.51 &  27.42 & 27.42 \\
$p=1\,000$ & 27.37  & 26.98      & \bf{26.94}       & 27.38  & 27.07      & 27.02       & 27.18 & 26.68 & 30.66 &  28.21 & 27.49  \\
        & \multicolumn{11}{c}{\cellcolor[gray]{0.8}\bf  Auto-correlation  $\rho_{ij}=0.9^{|i-j|}$}  \\ 
         & \multicolumn{11}{c}{\cellcolor[gray]{0.9}Linear}  \\ 
$p=100$  & 6.17   & 6.29       & 6.37        & 6.20   & 6.25       & 6.18        & 6.37 & 6.09 & 6.77   & 6.17  & \bf{5.91}  \\
$p=1\,000$ & 6.39   & 6.44       & \bf{6.00}        & 6.47   & 6.21       & 6.13        & 6.55 & 6.20 & 7.06   & 6.53  & 6.42 \\    [0.05in] 

    & \multicolumn{11}{c}{\cellcolor[gray]{0.9}CART}  \\ 
$p=100$  & 33.80  & 37.72      & 37.28       & 33.83  & 36.78      & 36.61       & \bf{33.57} & 34.40 & 35.05 & 35.61 & 35.81  \\
$p=1\,000$ & 31.57  & 33.55      & 37.21       & \bf{31.52} & 33.52      & 37.43       & 31.63 & 31.88 & 32.22 &  33.11 & 33.43 \\[0.05in] 

  & \multicolumn{11}{c}{\cellcolor[gray]{0.9}Friedman}  \\
$p=100$  & 34.09  & 32.51      & 34.65       & 34.27  & 34.97      & 32.77       & 36.88 & 33.83 & 48.64 &  34.21 & \bf{30.36}  \\
$p=1\,000$ & 39.09  & 39.57      & 32.58       & 40.58  & 43.05      & 33.46       & 41.80 & 37.38 & 49.51 &  35.96 & \bf{30.81} \\
[0.05in] 
  & \multicolumn{11}{c}{\cellcolor[gray]{0.9}LLS}  \\
  
$p=100$   & 28.57  & \bf{27.94}     & 30.71       & 28.45  & 28.03      & 29.12       & 28.88 & 27.87 & 30.69 &  28.83 & 28.81  \\
$p=1\,000$ & 29.98  & \bf{28.25}      & 28.96       & 30.14  & 28.40      & 28.38       & 30.19 & 28.56 & 32.29 &  31.76 & 29.28 \\
\bottomrule
\bottomrule
\end{tabular}}
\end{table}

While ABC Bayesian Forests were designed to explore the posterior distribution over models, it is natural to ask whether they also yield reasonable prediction. There are various ways to perform prediction with our ABC method. One natural strategy is to save each draw $f_{\mS,\bm B}^m$ at the $m^{th}$ ABC iteration when $\epsilon_m<\epsilon$ and average out individual predictions obtained from these single draws. Alternatively, one could first select variables based on ABC Bayesian Forests and then run a separate BART method (using the default number of $T=200$ trees which is recommended for prediction) with the selected variables. 
Using both strategies, we report average out-of-sample mean squared prediction error, where the average is taken over $20$ independent validation samples generated from the same data generating process (Table \ref{tab:oos_mse}). We include both ABC predictions described above and denote them as ABC1 and ABC2, respectively, for the two different thresholds ($c\in\{0.5,0.25\}$)  and for the two choices of the number of trees ($T\in\{10,20\}$). 

The best method under each simulation setting is marked in bold.  When the data becomes more non-linear  (CART and LLS setups) and the correlation among variables gets stronger, ABC  tends to outperform the other methods. DART, on the other hand, works better for more linear datasets. 
Note that our default ABC implementation internally uses only a {\em small} number of $B=200$ burn-in iterations and a small number of trees. 
For prediction, it has been recommended that BART is deployed with a larger number of trees  \citep{chipman2010bart}.
In addition,  the ABC computation produces forest samples  $f_{\mS,B}^m$ which are from an {\em approximate} posterior. These two facts may affect resulting predictions which may not necessarily outperform BART (DART) across-the-board.  }


}
{

\section{HIV Data}
\vspace{0.2cm}
To further illustrate the usefulness of our approach, we consider a dataset described and analyzed in \cite{rhee_genotypic_2006} and \cite{barber_controlling_2015}. 	The data consists of genotype
and resistance measurements (log-decrease in susceptibility)  for three drug classes, i.e. protease inhibitors (PIs), nucleoside reverse transcriptase inhibitors (NRTIs) and non-nucleoside reverse transcriptase inhibitors (NNRTIs). The data is publicly available from the Stanford HIV Drug Resistance Database.\footnote{  \url{https://hivdb.stanford.edu/pages/published_analysis/genophenoPNAS2006/}} 
	
	\begin{figure}[!t]
		\begin{center}
				\subfigure[ABC]{
				
				\includegraphics[width=4.5cm, height=5cm]{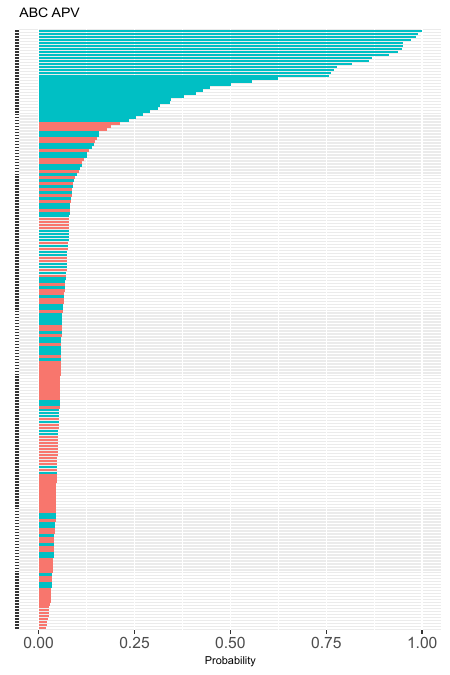}  }
			\subfigure[DART]{
				\centering
				\includegraphics[width=4.5cm, height=5cm]{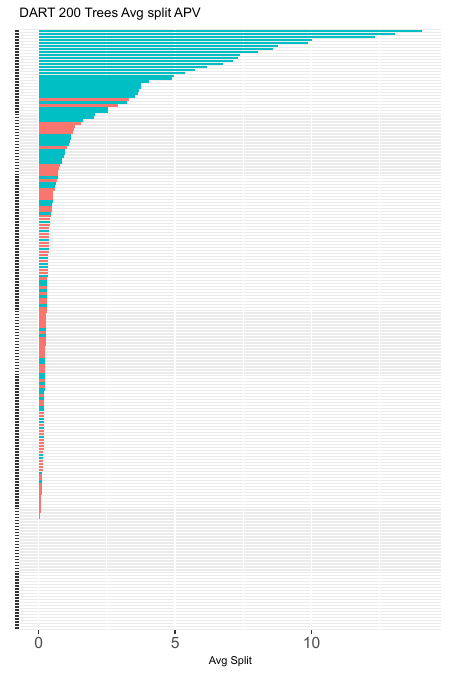} } 
			\subfigure[RF]{
				
				\includegraphics[width=4.2cm, height=5cm]{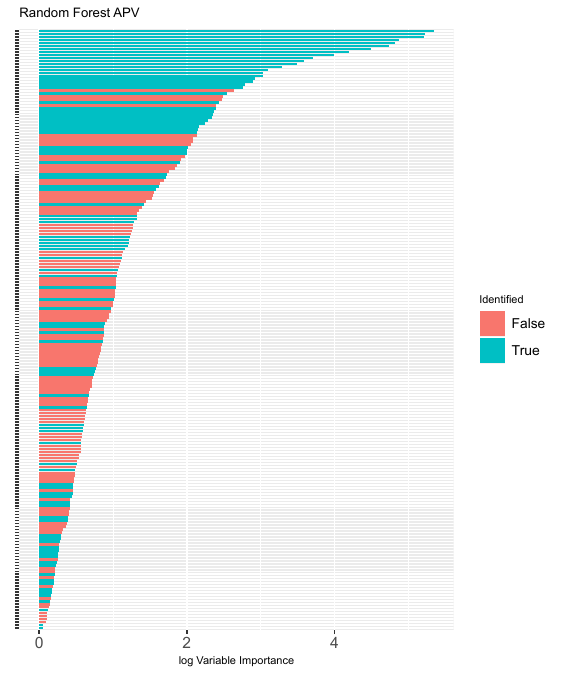} 
			} 
		\end{center}
		\caption{\small \label{fig:X3TC} A barplot of ordered importance measures (inclusion probabilities for ABC, importance measures for  DART and RF) for each of the $p=201$ mutations for the drug APV, where blue represents mutations found in \cite{soo-yon_rhee_hiv-1_2005}. (a) Inclusion probabilities are computed using  the top $1\,000$ out of $M=10\,000$ ABC samples; (b) Average split of DART with $20\,000$ MCMC iterations; (c) log variable importance of Random Forest with $500$ trees. }
	\end{figure}

The goal of this analysis is to identify possible non-polymorphic mutation positions
which result in  a log-fold increase of lab-tested drug resistance. 
The design matrix $X=(x_{ij})_{i,j=1}^{n,p}$ consists of  binary indicators  $x_{ij}\in \{0,1\}$  for  whether or not the $j^{th}$ mutation  occurred in the $i^{th}$ sample. As  in \cite{barber_controlling_2015}, only mutations that appear at least $3$  times  are taken into consideration.  
One appealing feature of this dataset is  the availability of a proxy to the `ground truth'. Indeed, in an independent experimental study, \cite{soo-yon_rhee_hiv-1_2005} identified mutations that are present at a significantly higher frequency in  patients who have been treated with each drug.  Similarly as \cite{barber_controlling_2015}, we treat this experimental data as an approximation to the truth for comparisons and for validation of our findings. 
	
	We run ABC with $M =10\,000$ iterations, where each internal BART sample is obtained after  $200$ burnin iterations with $20$ trees.
	The top 
	$1\,000$ ABC samples with the smallest $\epsilon_m$ are kept and used to compute  inclusion probabilities for each mutation. 
	For illustration,  we visualize results for  one of the PI drugs (APV)  and report  the results for all the drugs  in the  Supplemental Material (Section \ref{supplement:hiv}). The inclusion probabilities have been ordered  and plotted  in Figure \ref{fig:X3TC}, where
	the mutations experimentally validated by \cite{soo-yon_rhee_hiv-1_2005} (a proxy for true signals) are denoted in blue and the rest is in red. 
	For comparisons, we also included the importance measure (the average number of splits on each variable) from  DART  run with $20\,000$ MCMC iterations and $T=200$ trees as well as the importance measure (on a log scale) from Random Forests (RF) run with $500$ trees.

	\begin{figure}[t!]
		\begin{center}
	\subfigure[Adaptive Cut-off]{
		
		\includegraphics[width=.30\linewidth,height=5cm]{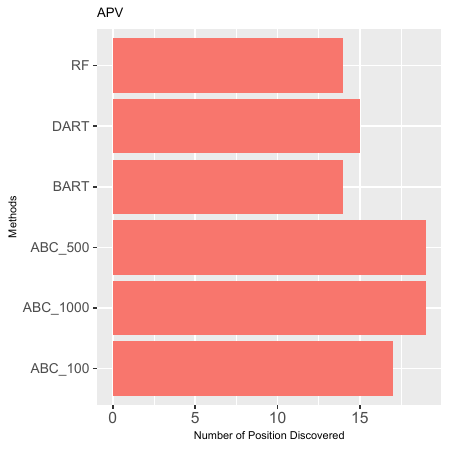}\label{fig2a}  }
	\subfigure[Automated Cut-off]{
		\centering
		\includegraphics[width=.33\linewidth,height=5cm]{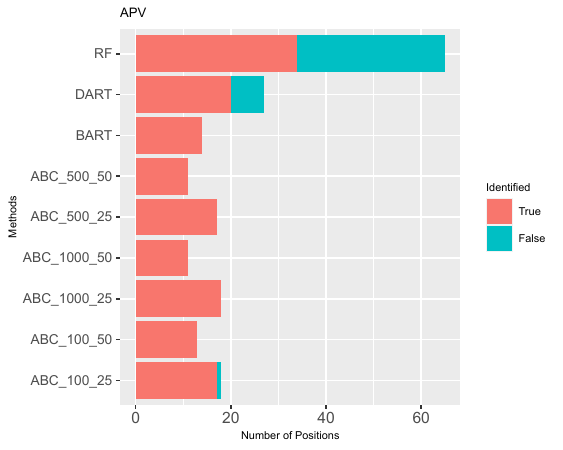}\label{fig2b}  } 
	\subfigure[AUC]{
		
		\includegraphics[width=.30\linewidth,height=5cm]{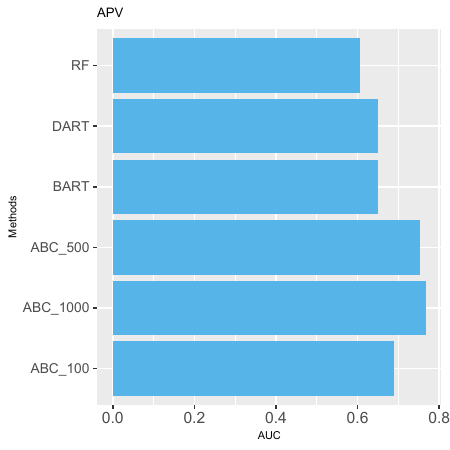} \label{fig2c} 
	} 
\end{center}
		\caption{\label{fig:X3TC_N} (a) The number of true discoveries using an adaptive cut-off;  (b) The number of true (red) and false (blue) discoveries using an automated cut-off; (c) The AUC of each method.  }
	\end{figure}	
	Figure \ref{fig:X3TC} reveals that ABC Bayesian Forests have a strong separation power, where experimentally validated mutations  generally have a  higher inclusion probability. Compared to DART and RF, ABC clearly stands out as being more effective in weeding out `noise'.  We gauge the strength of the signal/noise separation using several descriptive statistics. In these comparisons, we also consider plain BART method (using $T=20$ trees and $20\,000$ MCMC iterations) and ABC using the top $100$ and $500$ samples  with the smallest tolerance level $\epsilon_m$.
	Since the selection of the cut-off point is not obvious for BART and RF, we first select variables based on an adaptive cut-off point so that there are no false discoveries (i.e. the cut-off is the largest importance weight of  a {\em not} experimentally validated  mutation). From the plot of the number of `True' locations selected (displayed in Figure \ref{fig2a}) we can see that all three ABC implementations find more signal variables.
	Next, we choose the cut-off point in an automated way, where ABC importance probabilities are truncated at $0.5$ and $0.25$, BART and DART measures are truncated at one (i.e. the variable has been used on average at least once), and RF select variables using recursive feature elimination as explain in the previous section. Similarly to  \cite{barber_controlling_2015}, we report the number of 'True' locations and 'False' locations (Figure \ref{fig2b}).  RF selection is plagued with false discoveries and DART is not free from false identifications either. The ABC selection cutoff $0.5$ results in a more conservative selection, where lowering the cutoff point to $0.25$ yields more discoveries.
	Finally, from the plot of  the AUC values for all considered methods (Figure \ref{fig2c}), we conclude that ABC is better at separating the experimentally validated mutations from the rest even using a very few filtered ABC samples.

	}
	
\vspace{-0.5cm}
\section{Discussion}
This paper makes advancements at two fronts. One is the proposal of ABC Bayesian Forests for variable selection based on a new idea of data splitting, where a fraction of data is first used for ABC proposal and  the rest for ABC rejection. This new strategy increases ABC acceptance rate. We have shown that ABC Bayesian Forests are  highly competitive with (and often better than) other tree-based variable selection procedures. 
The second development is theoretical and concerns consistency for  variable and regularity selection.  Continuing the theoretical investigation of BART by \citet{rockova2017posterior}, we proposed new complexity priors which jointly penalize model dimensionality and tree size. We have shown joint consistency for variable  {\sl and} regularity selection  when the level of smoothness is unknown and no greater than 1. Our results are the first model selection consistency results for BART priors.

Our ABC sampling routine has the potential to be extended in various ways.
Sampling from  $\pi(f_{\mE,\B},\sigma^2\C\Y^{obs}_{\mI_m},\mS_m)$ in ABC Bayesian Forests is one way  of distilling $\Y^{obs}_{\mI_m}$ to propose a candidate ensemble  $f_{\mE,\B}^m$.
We noticed that the ABC acceptance rate can be further improved by replacing a randomly sampled tree with a fitted tree. Indeed,  instead of drawing from $\pi(f_{\mE,\B},\sigma^2\C\Y^{obs}_{\mI_m},\mS)$, one can {\sl fit} a tree $\wh f_{\mT,\b}^m$ to $\Y^{obs}_{\mI_m}$ using recursive partitioning algorithms (such as the \texttt{rpart} R package of \citet{therneau2018package} or  with BART (by taking  the posterior mean estimate $\wh f_{\mE,\B}^m=\E[f_{\mE,\B}\C\Y^{obs}_{\mI_m},\mS]$). This variant, further referred to as ABC  Forest  Fit, is indirectly linked to other model-selection methods based on resampling.

\citet{felsenstein1985confidence} proposed a ``first-order bootstrap" to assess confidence of an estimated tree phylogeny. The idea was to construct a tree from each bootstrap sample and record the proportion of bootstrap trees that have a feature of interest (for us, this would be variables used for splits). \citet{efron1998problem} embedded this approach within a parametric bootstrap framework, linking the bootstrap  confidence level to both frequentist $p$-values  and Bayesian a posteriori model probabilities. The authors proposed a second-order extension by reweighting the first-order resamples according to a simple importance sampling scheme. This second-order variant performs frequentist calibration of the a-posteriori probabilities and amounts to performing Bayesian analysis with Welch-Peers uninformative priors. \citet{efron2012bayesian} further develops the connection between parametric Bootstrap and posterior sampling through reweighting in exponential family models. Using non-parametric bootstrap ideas, \citet{newton1994approximate} introduce the weighted likelihood bootstrap (WLB) to sample from approximate posterior distributions. The WLB samples are obtained by maximum reweighted likelihood estimation with random weights. Such posterior sampling can be beneficial when, for instance, maximization is easier than Gibbs sampling from conditionals. 
In a similar spirit, our  ABC Forest Fit variant would perform optimization (instead of sampling)  on a random subset of the dataset to obtain a candidate tree/ensemble.

	It is worth pointing out that  $\wh f_{\mE,\B}^m$ does not  necessarily have to be a tree/forest. 
	We suggest trees because they are are easily trainable and produce stable results using traditional software packages. In principle, however, this method could be deployed in tandem with other non-parametric methods, such as deep learning, to perform variable selection.

\section*{Acknowledgments}
This work was supported by the James S. Kemper Foundation Faculty Research Fund at the University of Chicago Booth School of Business.

\bibliographystyle{rss}
\bibliography{tree}

\pagebreak
\
\\

\noindent \textbf{\huge Supplemental Materials}

\bigskip
%

\newcommand{\beginsupplement}{%
\setcounter{section}{0}
\renewcommand{\thesection}{S.\arabic{section}}
        \setcounter{table}{0}
        \renewcommand{\thetable}{S\arabic{table}}%
        \setcounter{figure}{0}
        \renewcommand{\thefigure}{S\arabic{figure}}%
     }

\beginsupplement

\section{Theory}
\subsection{Proof of Theorem \ref{thm:consist}}\label{sec:proof:tm1}
We first review some notation used throughout this section and adapted from \cite{rockova2017posterior}.
Recall that $\Pi_{\mS}(\cdot)$ denotes the conditional distribution given the model $\mS$. 
Next,  $\mF_{\mS}(K)$ denotes a set of all  step functions $f_{\mT,\b}(\cdot)$ with $K$ steps that split on  covariates $\mS$ and $\|f_{\mT,\b}\|_\infty\leq B$. A tree partition is called valid when each tree splits on observed values $\mX=\{\x_1,\dots,\x_n\}$ and has nonempty cells. We denote with $\mV_\mS^K$ all valid trees obtained by splitting $K-1$ times along coordinates inside $\mS$. The number of such valid trees is denoted with $\Delta(\mV_\mS^K)$.
For a valid tree partition $\mT\in \mV_\mS^K$, we denote with $\mF(\mT)\subset \mF_{\mS}(K)$ all step functions supported on $\mT$.
We  prove Theorem \ref{thm:consist} by verifying conditions B1-B4 in Theorem 4 of \cite{yang2017bayesian} (further referred to as YP17).  We build on tools developed in \cite{rockova2017posterior} (further referred to as RP17).

\subsubsection{Prior Concentration Condition}
The first condition  pertains to prior concentration and consists of two parts:  (a)  the model prior mass condition and (b)  the prior concentration condition  in the parameter space under the true model.
Namely, we want to show that 
\begin{equation}\label{prior:mass:model}
\pi(\mS_0)\geq \e^{-n\,\varepsilon_{n,\mS_0}^2}
\end{equation}
and
\begin{equation}\label{prior:mass}
 \Pi_{\mS_0}\left(f_{\mT,\b}\in\mF_{\mS_0}(K):\|f_{\mT,\b}-f_0\|_n\leq\varepsilon_{n,\mS_0}\right)\geq \e^{-d\,n\,\varepsilon_{n,\mS_0}^2}
\end{equation}
for some $d>2$. The prior concentration \eqref{prior:mass:model} follows directly from the definition of model weights \eqref{model_prior} for $C\leq C_\varepsilon^2$ under our assumption  $q_0\log p<n^{q_0/(2\alpha+q_0)}$.

Regarding \eqref{prior:mass}, a variant of this condition  is verified in Section 8.2 of RP17 assuming that $K$ is random with a prior. It follows from their proof, however, that \eqref{prior:mass} holds if we fix $K$ at $K_{\mS_0}=\lfloor C_K/C_\varepsilon^2\, n\,\varepsilon_{n,\mS_0}^2/\log n\rfloor=2^{q_0s}$ for some $s\in\N$. The proof consists of (a) constructing a single approximating tree (i.e. the $k$-$d$ tree with $s=(\log_2 K_{\mS_0})/q_0$ cycles of splits on each coordinate in $\mS_0$) and showing that it has enough prior support. 
This tree exists under the assumption that the design is $\mS_0$-regular. From (8.5) of RP17, such tree approximates $f_0$ with an error bounded by a constant multiple of $\varepsilon_{n,\mS_0}$. The verification of \eqref{prior:mass} then follows directly from  RP17.

\subsubsection{Entropy Condition}
The second condition (B4 in the notation of YP17) entails controlling the complexity of over/underfitting models. In the sequel, we focus only on models with up to $q_n$ covariates, where $q_n=C_q\lceil n\,\varepsilon_{n,\mS_0}^2/\log p\rceil$. This restriction is justified by the following lemma.

\begin{lemma}\label{lemmadim}
Denote with $q_n=C_q\lceil n\,\varepsilon_{n,\mS_0}^2/\log p\rceil$. Under the assumptions of Theorem \ref{thm:consist}, we have
 \begin{equation}\label{lemma:dim1}
 \Pi(q\geq q_n\C\Y^{(n)})\rightarrow0
 \end{equation}
 in $\P_{f_0}^{(n)}$-probability as $n\rightarrow\infty$.
\end{lemma}
\begin{proof}
First, we show that $\Pi(q\geq q_n)\e^{(d+2)n\,\varepsilon_{n,\mS_0}^2}\rightarrow 0$, where  $d>2$ is as in \eqref{prior:mass}. 
We can write 
\begin{align*}
\Pi(q>q_n)\e^{(d+2)n\varepsilon_{n,\mS_0}^2}&\lesssim\e^{(d+2)n\,\varepsilon_{n,\mS_0}^2} \sum_{k=q_n}^{p}{p\choose k}\e^{-C\,\times\max\{n^{k/(2\alpha+k)}\log n,k\log p\}}\\
&\leq  \e^{(d+2)n\,\varepsilon_{n,\mS_0}^2-(C-2)\, q_n\log p}=\e^{-n\,\varepsilon_{n,\mS_0}^2[(C-2)C_q- (d+2)]}. 
\label{sum2}
\end{align*}
The right hand side above goes to zero when $(C-2)C_q- (d+2)>0$.  This can be satisfied with $C>2$ and $C_q$ large enough. 
This fact, together with prior mass conditions \eqref{prior:mass} and \eqref{prior:mass:model}, yields \eqref{lemma:dim1} according to Lemma 1 of \citet{ghosal2007convergence}.
\end{proof}

Lemma \ref{lemmadim} essentially states that the posterior will not reward models whose dimensionality is larger than (or equal to) $q_n$. In our following considerations, we thus condition only models with less than $q_n$ variables.

We now verify that the  complexity of overfitting models $\mS\supset\mS_0$ is not too large in the sense that their global metric entropy satisfies
\begin{equation}\label{eq:cover}
\log N\left(\varepsilon_{n,\mS}\,;\,\mF_{\mS}(K_{\mS})\,;\,\|\cdot\|_n\right)\leq n\,\varepsilon_{n,\mS}^2.
\end{equation}
First, we note that for two tree step functions  $f_{\mT,\b_1}\in\mF(\mT)$ and $f_{\mT,\b_2} \in \mF(\mT)$ that have the same partition $\mT\in\mV_{\mS}^{K_\mS}$ and different step heights  $\b_1\in\R^{K_{\mS}}$ and $\b_2\in\R^{K_{\mS}}$, we have $\{\|f_{\mT,\b_1}-f_{\mT,\b_2}\|_n\leq \varepsilon_{n,\mS}\}\supset\{\|\b_1-\b_2\|_2\leq \varepsilon_{n,\mS}\}$. Furthermore,  noting that $\mF(\mT)=\{f_{\mT,\b}:\|f_{\mT,\b}\|_\infty\leq B \}\subset \{\b\in\R^{K_{\mS}}: \|\b\|_2\leq B\, \sqrt{n} \}$ we can write
$$
N\left(\varepsilon_{n,\mS}\,;\, \mF(\mT)\, ;\,\|\cdot\|_n\right)\leq \left(\frac{3\,B\,\sqrt{n}}{\varepsilon_{n,\mS}}\right)^{K_{\mS}}\leq \left({3\,B}\,n^{3/2}/C_\varepsilon\right)^{K_{\mS}},
$$
where we used the standard  $\varepsilon_{n,\mS}$ covering number of a $K_{\mS}$-Euclidean ball of a radius $B\,\sqrt{n}$ and the fact that $1/\varepsilon_{n,\mS}\leq 1/C_\varepsilon\times n^{\alpha/(2\alpha+|\mS|)}\leq 1/C_\varepsilon\times n$. Then we can write
 $$
 N\left(\varepsilon_{n,\mS}\,;\,\mF_{\mS}(K_{\mS})\,;\,\|\cdot\|_n\right)\leq \Delta(\mV_{\mS}^{K_{\mS}})  \left({3\,B}\,n^{3/2}/C_\varepsilon\right)^{K_{\mS}}.
 $$
Using Lemma 3.1 of Rockova and van der Pas (2017), we have $\Delta(\mV_{\mS}^{K_{\mS}}) \leq (K_{\mS}\,n\,|\mS|)^{K_{\mS}}$.

The overall log-covering number is then upper-bounded with (since $|\mS|\leq q_n\leq n $)
\begin{align}\label{entropy:overall}
&K_{\mS}\log \left({3\,B}\,n^3\,n^{3/2}\right)\lesssim K_{\mS}\log n\propto n\,\varepsilon_{n,\mS}^2.
\end{align}
This verifies the model complexity condition for overfitting models.
Next, we need to verify  \eqref{eq:cover} with $\varepsilon_{n,\mS}$ replaced by $\wt\varepsilon_n$ for  ``underfitting" models  $\mS\in \G_{\mS\not\supset\mS_0}$ where $|\mS|\leq q_n$.  This follows from the  same arguments as above and the fact that $\varepsilon_{n,\mS}\leq \wt\varepsilon_n$.
Finally, the  last requirement in Assumption B4 of YP17 is verifying that 
\begin{equation}\label{almost_last_one}
\sum_{\mS\not\supset\mS_0:|\mS|\leq q_n}\e^{-C_2\,n\,\wt\varepsilon_n^2}+\sum_{\mS\supset\mS_0:|\mS|\leq q_n}\e^{-C_2\,n\,\varepsilon^2_{n,\mS}}
\leq 1 
\end{equation}
for some large constant $C_2>0$. Since $\wt\varepsilon_n\geq\varepsilon_{n,\mS}> \varepsilon_{n,\mS_0}$ for any $\mS\supset\mS_0$ such that $|\mS|\leq q_n$, we can upper-bound the left-hand side above with
$$
\sum_{q=0}^{ q_n}\sum_{\mS:|\mS|=q} \e^{-C_2\,n\,\varepsilon^2_{n,\mS_0}}
\leq  \e^{-C_2\,n\,\varepsilon^2_{n,\mS_0}} \sum_{q=0}^{ q_n}{p\choose q}\leq \left(\frac{2\,\e\,p}{q_n}\right)^{q_n+1} \e^{-C_2\,n\varepsilon^2_{n,\mS_0}}
$$
From our definition of $q_n$, we have $q_n\log p\asymp n\, \varepsilon_{n,\mS_0}^2$ and \eqref{almost_last_one} will be satisfied  for a large enough $C_2$.

\subsubsection{Prior Anticoncentration Condition}\label{sec:proof:antic}
Lastly, as one of the sufficient conditions for model selection consistency, we need to verify
\begin{equation}\label{last_one}
\sum_{\mS\supset\mS_0:|\mS|\leq q_n}\pi(\mS)\,\Pi_{\mS}\left(f_{\mT,\b}\in\mF_{\mS}(K_{\mS}): \|f_0-f_{\mT,\b}\|_n\leq M\, \varepsilon_{n,\mS}\right)\leq \e^{-H\, n\,\varepsilon_{n,\, \mS_0}^2}
\end{equation}
for some $H>0$.  Alternatively,  YP17 introduce the so-called ``anti-concentration condition"
$\Pi_{\mS}\left(f_{\mT,\b}\in\mF_{\mS}(K_{\mS}): \|f_0-f_{\mT,\b}\|_n\leq M\, \varepsilon_{n,\mS}\right)\leq \e^{-H\, n\,\varepsilon_{n,\, \mS_0}^2}$
 for overfitting models $\mS\supset\mS_0$ where $\varepsilon_{n,\mS}\geq \varepsilon_{n,\mS_0}$. This condition is needed to show that the posterior probability  of more complex models that contain the truth goes to zero.
 
 It turns out that this condition can be avoided with our choice of model weights $\pi(\mS)$ \citep{ghosal2008nonparametric}.
 We can verify \eqref{last_one} directly (without the anticoncentration condition) by upper-bounding the left hand side of \eqref{last_one} with
\begin{equation}\label{eq:anticoncentration}
\sum_{\mS\supset\mS_0:|\mS|\leq q_n}\pi(\mS)\leq \sum_{\mS\supset\mS_0:|\mS|\leq q_n} {\e^{-C\,n\,\varepsilon_{n,\mS}^2}}\leq \e^{-C\,n\,\varepsilon_{n,\mS_0}^2} \left(\frac{2\,\e\,p}{q_n}\right)^{q_n+1}.
\end{equation}
Since $q_n\log p\asymp n\, \varepsilon_{n,\mS_0}^2$, \eqref{last_one} holds for $H<C-1$.

\subsubsection{Identifiability}

Under the identifiability and irrepresentability assumptions \eqref{ass:identify} and \eqref{ass:irrep}, it turns out that we cannot approximate $f_0$ well enough with models that miss at least one covariate. This property is summarized in the following Lemma, which is a variant of  Proposition 1 of YP17.

\begin{lemma}\label{lemma:approx}
For $f_0\in\Ha_p\cap\mC(\mS_0)$, assume that $\mS_0$ is $(f_0,\varepsilon)$-{\sl identifiable} and that $\varepsilon$-irrepresentability holds. Then
$$
\inf\limits_{\mS\not\supset\mS_0}\inf_{f_{\mT,\b}\in\mF_\mS}\|f_0-f_{\mT,\b}\|_n>M\,\varepsilon.
$$
\end{lemma}
\begin{proof}
We decompose $\mS\not\supset\mS_0$ into true positives and false positives, i.e. $\mS=\mS_1\cup\mS_2$, where  $\mS_1\subset\mS_0$ and $\mS_2\cap\mS_0=\emptyset$. 
We denote with $\wh f^{\mS}$ the projection of $f_0$ onto $\mF_\mS$, omitting the subscripts $\wh\mT$ and $\wh\b$. 
With a slight abuse of notation we denote $\E(f,g)=\frac{1}{n}\sum_{i=1}^nf(\x_i)g(\x_i)$.
Then we can write
$$
\|f_0-\wh f^{\mS}\|_n^2=\|f_0-\wh f^{\mS_1} +\wh f^{\mS_1}- \wh f^{\mS}\|_n^2>\|f_0-\wh f^{\mS_1}\|_n^2-2 |\E[(f_0-\wh f^{\mS_1})(\wh f^{\mS}-\wh f^{\mS_1})]|,
$$
where
$
\E[(f_0-\wh f^{\mS_1})(\wh f^{\mS}-\wh f^{\mS_1})]
$
equals $\rho^{\mS}_n$  defined in \eqref{rho}. We note that $\delta_n^{\mS_1}$ is monotone increasing in  the number of false non-discoveries $|\mS_0\backslash \mS_1|$. The statement of the Lemma then follows from the fact that 
$
\|f_0-\wh f^{\mS}\|_n^2> \inf\limits_{\mS_1\subset\mS_0}\delta_n^{S_1}-2\sup\limits_{\mS\not\supset\mS_0}\rho_n^S>\inf\limits_{i\in\mS_0}\delta_n^{S_0\backslash i}-M\,\varepsilon> M\varepsilon.
$
\end{proof}

\subsection{Proof of Theorem \ref{thm:consist2}}\label{sec:proof:thm2}
We introduce some more notation. We denote with $\mF_{\mS}=\bigcup_{K=1}^n\mF_{\mS}(K)$  all valid trees that split on directions inside $\mS$ and we write $\Pi_{K,\mS}(\cdot)$ for the conditional prior, given $K$ and $\mS$.

Similarly as in Section \ref{sec:proof:tm1}, we verify the three conditions (Prior Concentration, Entropy, Prior Anti-concentration). 
The prior model concentration condition is again satisfied automatically from the definition of model weights in \eqref{prior:joint} and $K_{\mS_0}=\lfloor C_K/C_\varepsilon\,n\,\varepsilon_{n,\mS_0}^2/\log n\rfloor$. Namely,
\begin{equation}\label{prior_mass2}
\pi(K_{\mS_0},\mS_0)\propto \e^{-C\,\max\{C_K/C_\varepsilon n\,\varepsilon_{n,\mS_0}^2,q_0\log p\}}\geq \e^{-n\,\varepsilon_{n,\mS_0}^2},
\end{equation} 
for  $C_K<C_\varepsilon/C$, where we used the assumption $q_0\log p\leq n^{q_0/(2\alpha+q_0)}$. 
Next, the prior concentration in the parameter space associated with the true model
$$
\Pi_{K_{\mS_0},\mS_0}\left(f_{\mT,\b}\in\mF_{\mS_0}(K_{\mS_0}):\|f_{\mT,\b}-f_0\|_n\leq\varepsilon_{n,\mS_0}\right)\geq \e^{-d\,n\,\varepsilon_{n,\mS_0}^2}
$$
follows again from Section 8.2 of RP17.

For the entropy considerations, we focus only on models with up to $q_n$ covariates and up to $K_n$ splits, where $q_n=\lceil C_qn\,\varepsilon_{n,\mS_0}^2/\log p\rceil$ and $K_n=\lceil \bar C n\,\varepsilon_{n,\mS_0}^2/\log n\rceil$ were defined in Theorem \ref{thm:consist2}. This restriction is justified by the following Lemma.

\begin{lemma}\label{lemma:dim2}
Denote with $q_n=\lceil  C_q n\,\varepsilon_{n,\mS_0}^2/\log p\rceil$ and $K_n=\lceil  \bar C\, n\,\varepsilon_{n,\mS_0}^2/\log p\rceil$. Under the assumptions of Theorem \ref{thm:consist}, we have
 \begin{equation}\label{lemma:dim}
 \Pi(q\geq q_n\C\Y^{(n)})\rightarrow0\quad\text{and}\quad  \Pi(K\geq K_n\C\Y^{(n)})\rightarrow0
 \end{equation}
 in $\P_{f_0}^{(n)}$-probability as $n\rightarrow\infty$.
\end{lemma}
\begin{proof}
It suffices to show that $\Pi(q>q_n)\e^{(d+2)n\varepsilon_{n,\mS_0}^2}\rightarrow 0$ and $\Pi(K\geq K_n)\e^{(d+2)n\varepsilon_{n,\mS_0}^2}\rightarrow 0$ for $d>2$ from \eqref{prior:mass}. 
We have $q_0\leq q_n$ for $n$ large enough, since $q_0=\mathcal{O}(1)$ as $n\rightarrow\infty$, and thereby
\begin{align*}
&\Pi(q\geq q_n)\e^{(d+2)n\varepsilon_{n,\mS_0}^2}\lesssim\e^{(d+2)n\varepsilon_{n,\mS_0}^2} \sum_{q=q_n}^{p}{p\choose q}\sum_{K=1}^n
\e^{-C\,\max\{K\log n,q\log p\}}\\
&\leq  \e^{(d+2)n\varepsilon_{n,\mS_0}^2}\sum_{q=q_n}^{p}\e^{\log n+q\log(p\,\e/q)- C\, q\log p}
\leq  \e^{\log p+\log n-(C-1)\, q_n\log p+(d+2)n\,\varepsilon_{n,\mS_0}^2 }\\
&\leq  \e^{-(C-3)\, q_n\log p+(d+2)n\,\varepsilon_{n,\mS_0}^2 },
\label{sum2}
\end{align*}
where we used the fact that for $q_0\geq 2$ and $\alpha\in(0,1]$, we have $\log n\leq n^{q_0/(2\alpha+q_0)}$. Since $q_n\log p\geq C_q n\varepsilon_{n,\mS_0}^2$, the right hand side above goes to zero when $(C-3)C_q>d+2$. This will be guaranteed with $C>3$ and $C_q$  large enough.
Similarly,
we have 
\begin{align*}
&\Pi(K\geq K_n)\e^{(d+2)n\,\varepsilon_{n,\mS_0}^2}\lesssim\e^{(d+2)n\,\varepsilon_{n,\mS_0}^2} \sum_{q=0}^{p}{p\choose q}\sum_{K=K_n}^n
\e^{-C\,\max\{K\log n,q\log p\}}\\
&\leq  \e^{(d+2)n\,\varepsilon_{n,\mS_0}^2}\sum_{q=0}^{p}\sum_{K=K_n}^n\e^{- (C-1) \max\{K\log n,q\log p\}}\\
&\leq   \e^{\log (p+1)+\log n-(C-1)\, K_n\log n+(d+2)n\,\varepsilon_{n,\mS_0}^2 }
\leq  \e^{-(C-2)\, K_n\log n+(d+3)n\,\varepsilon_{n,\mS_0}^2 },
\end{align*}
where we used our assumption $\log p\leq  n^{q_0/(2\alpha+q_0)}$.
Since $K_n\geq \bar C n\,\varepsilon_{n,\mS_0}^2$, the right hand side above goes to zero when $(C-2)\bar C>d+3$.
Together with the prior mass conditions \eqref{prior:mass} and \eqref{prior_mass2}, \eqref{lemma:dim} follows from Lemma 1 of \citet{ghosal2007convergence}.
\end{proof}

This Lemma  essentially says that the posterior does not overfit in terms of both $q$ and $K$, where the mass concentrates on models with $K<K_n$ splits. Note that $K_n$ is of the same order as the optimal regularity $K_{\mS_0}$. Now, we denote with $\bG_n\subset \bG$ a sieve consisting of all models with less than $q_n$ variables and $K_n$ splits.
For the entropy bounds of overfitting and underfitting models (inside the sieve $\bG_n$),  we can use the same arguments as in Section \ref{sec:proof:tm1}.
Assume a model $(K,\mS)\in \bG_n$. Then it follows from \eqref{entropy:overall} that
$$
\log N\left(\varepsilon_{n,\mS}\,;\,\mF_{\mS}(K)\,;\,\|\cdot\|_n\right)\leq K\log (3\,B\, n^{3}n^{3/2})\lesssim  K_n \log n\lesssim n \,\varepsilon_{n,\mS_0}^2.
$$
For over-fitting models, this can be further upper-bounded with a multiple of  $n\,\varepsilon_{n,\mS}^2$, thus satisfying \eqref{eq:cover}. 
The last requirement for the entropy condition is verifying the following variant of \eqref{almost_last_one} 
\begin{equation}\label{bound11}
\sum_{(K,\mS)\in\bG_n:\,{\mS\not\supset\mS_0\cup K< K_{\mS_0}}}\e^{-C_2\,M^2\,n\,\varepsilon_{n,\mS_0}^2}+
\sum_{(K,\mS)\in\bG_n:\,{\mS\supset\mS_0\cap K\geq K_{\mS_0}}}\e^{-C_2n\,\varepsilon_{n,\mS}^2}\leq 1
\end{equation}
for some suitable $C_2>0$.
Since $n\,\varepsilon_{n,\mS_0}^2\leq n\,\varepsilon_{n,\mS}^2$ for $\mS\supset\mS_0$, we can upper-bound the left hand side with
\begin{equation}\label{bound10}
\sum_{\mS:|\mS|< q_n}\sum_{K=1}^{K_n}\e^{-C_2n\,\varepsilon_{n,\mS_0}^2}\leq \e^{-C_2n\,\varepsilon_{n,\mS_0}^2} \left(\frac{2\,\e\, p}{q_n}\right)^{q_n+1}\e^{\log K_n}
\leq  \e^{-C_2n\,\varepsilon_{n,\mS_0}^2+(q_n+1)\log p+\log K_n}. 
\end{equation}
Since $q_n\log p\asymp n\,\varepsilon_{n,\mS_0}^2$ and $\log K_n\lesssim n^{q_0/(2\alpha+q_0)}\lesssim n\varepsilon_{n,\mS_0}^2$, the right-hand side of \eqref{bound10} converges to zero for some suitably large $C_2$ as $n\rightarrow \infty$, thus satisfying \eqref{bound11}.

In place of the  anti-concentration condition (similarly as in  \eqref{eq:anticoncentration}), we need to verify that the prior probability of larger models (that contain the truth) is small in the sense that, for some $H>0$,
\begin{equation}\label{bound13}
\sum_{(K,\mS)\in\bG_n:{\{\mS\supset\mS_0\}\cap \{K\geq K_{\mS_0}\}}}\pi(\mS,K)\leq \e^{-H\,n\,\varepsilon_{n,\mS_0}^2}.
\end{equation}
We can write
\begin{align}
\sum_{\mS\supset\mS_0:|\mS|< q_n}\sum_{K=K_{\mS_0}}^{K_n}\pi(\mS,K)&\leq
\sum_{q=0}^{q_n}{p\choose q}\sum_{K=K_{\mS_0}}^{K_n}\e^{-C\,K_{\mS_0}\log n}\\
&\leq \left(\frac{2\,\e\, p}{q_n}\right)^{q_n+1}\e^{\log K_n}\e^{-C\,K_{\mS_0}\log n}.
\end{align}
Because $q_n\log p\asymp n\,\varepsilon_{n,\mS_0}^2$ and $\log K_n\lesssim n^{q_0/(2\alpha+q_0)}\lesssim n\varepsilon_{n,\mS_0}^2$
the condition \eqref{bound13} is satisfied for some $H>0$ when $C$ and $C_K$ are large enough.

\subsection{Proof of Theorem \ref{thm:consist3}}\label{sec:proof:tm3}
We modify the notation a bit. We adopt the definition of $\delta$-valid ensembles from RP17 (Definition 5.3). With $\mF_{\mS}({\bm K})$ we denote all $\delta$-valid tree ensembles $f_{\mE,\B}$  that  (a) are uniformly bounded (i.e. $\|f_{\mE,\B}\|_\infty\leq B$ for some $B>0$), (b) consist of $T$ trees with $\bm K=(K^1,\dots, K^T)'\in\N^T$ leaves and (c) that split along directions $\mS$.

We start by showing that the prior model concentration condition is satisfied.  From our assumption $q_0\log p\leq n^{q_0/(2\alpha+q_0)}$ and definition  $K_{\mS_0}< C_K/C_\varepsilon^2\, n\,\varepsilon_{n,\mS_0}^2/\log n$  and using \eqref{prior:T} and \eqref{prior:joint2},  we obtain
$$
\pi(\mS_0,E(K_{\mS_0}))\propto\sum_{T=1}^{K_{\mS_0}}\e^{-C_T\, T}\sum_{\bm K\in\N^T:\sum_{t=1}^TK^t=K_{\mS_0}} \e^{-C\,n^{q_0/(2\alpha+q_0)}\log n}\geq \e^{-(C_TC_K/(C_\varepsilon\,\log n)+C/C_\varepsilon^2) \,n\,\varepsilon_{n,\mS_0}^2}.
$$
The right-hand side can be further lower-bounded with $\e^{-n\,\varepsilon_{n,\mS_0}^2}$ for a large enough $C_\varepsilon$ and $n$.
Next, we need to show prior concentration in the parameter space under the true model equivalence class $(\mS_0,E(K_{\mS_0}))$.  All that is needed is finding a single well-approximating forest supported on one partition ensemble characterized by $(T,\bm K)$ from the equivalence class $E(K_{\mS_0})$. Such an ensemble  can be obtained by considering $T=1$ and a single $k$-$d$ tree with $K_{\mS_0}$ leaves from Lemma 3.2 of RP17.  The prior concentration condition then boils down to \eqref{prior:mass}, which has already been verified in RP17.

Next, we show that for $K_n=\lceil \bar C\,n\,\varepsilon_{n,\mS_0}^2/\log n\rceil$ we have
$$
\Pi\left((T,\bm{K}): \sum_{t=1}^TK^t\geq K_n\,\big|\,\Y^{(n)}\right)\rightarrow 0.
$$
We can write
\begin{align*}
\Pi\left((T,\bm{K}): \sum_{t=1}^TK^t\geq K_n\right)&\lesssim \sum_{T=1}^{n}\e^{-C_T\,T}\sum_{q=0}^p
{p\choose q}\sum_{Z={K_n}}^n\sum_{\bm K:\sum_{t=1}^T K^t=Z}\e^{-C \max\{Z\log n,q\log p\}}\\
&\lesssim \e^{-(C-1) K_n\log n +\log p+2\log n+\log p(n)-C_T},
\end{align*}
where $p(n)$ is the partitioning number.
According to Andrews (1976), we have
\begin{equation}\label{partition_bound}
\log p(n)\sim \pi\sqrt{\frac{2\,n}{3}}\quad\text{as}\quad n\rightarrow\infty.
\end{equation}
Under our assumptions $q_0> 2$ and $\alpha\in(0,1]$, we have $\sqrt{n}\leq n^{q_0/(2\alpha+q_0)}$ and $\log n\leq n^{q_0/(2\alpha+q_0)}$.
From $\log p\leq n^{q_0/(2\alpha+q_0)}$ and using the fact that $K_n\geq \bar C\,n\,\varepsilon_{n,\mS_0}^2/\log n$, we can then write
$$
\Pi\left((T,\bm{K}): \sum_{t=1}^TK^t\geq K_{n}\right)\e^{(d+2)n\,\varepsilon_{n,\mS_0}^2}\lesssim \e^{-[(C-1)\bar C-D\,\pi\sqrt{2/3}-d-5]n\,\varepsilon_{n,\mS_0}^2}
$$
for some $D>0$. The right hand side goes to zero for $C>1$ and $\bar C$ large enough. 
Similarly, we can show that $\Pi(q\geq q_n\C\Y^{(n)})\rightarrow 0$ as $n\rightarrow\infty$  for $q_n=\lceil C_q\,n\,\varepsilon_{n,\mS_0}^2/\log p\rceil$ by proceeding as in Lemma  \ref{lemma:dim2} in  Section \ref{sec:proof:thm2}.

Based on the previous paragraph, we narrow down attention to a subset of model indices $\bG_n\subset\bG$, consisting of models $(\mS,E(Z))$ such that $|\mS|< q_n$ and $Z< K_n$.
 We now define a sieve $\mF_n$ as follows
 $$
 \mF_n=\bigcup_{q=0}^{q_n}\bigcup_{T=1}^{K_n}\bigcup_{\sum_{t=1}^TK^t\leq K_n}\bigcup_{\mS:|\mS|=q}\mF_{\mS}({\bm K}).
 $$
It follows from the previous paragraph that $\Pi(\mF_n^c\C\Y^{(n)})\rightarrow 0$ as $n\rightarrow\infty$. For the entropy calculation we thus focus on the sieve $\mF_n$.

We first note that the metric entropy $\log N\left(\varepsilon_{n,\mS};\mF(\mE); \|\cdot\|_n\right)$, where $\mF(\mE)$ are all uniformly bounded forests supported on a $\delta$-valid partition ensemble $\mE$,
can be upper-bounded with $\left(\sum_{t=1}^TK^t\right)\log(B/\varepsilon_{n,\mS}C_1\kappa(\mE)\sqrt{n})$ (follows from equation (9.3) of RP17), 
where  $\kappa(\mE)$ is the condition number of a valid ensemble (defined in Section 9.1. of RP17).
Next, we find an upper bound for the covering number of the tree ensembles that are attached to a model $(\mS,E(Z))$, where $E(Z)$ is the equivalence class of $(T,\bm{K})$ defined in \eqref{equivalance_class}.
From Section 9.1 of RP17, and using  the fact that $\Delta(E(Z))\leq Z!p(Z)$, it follows that
 \begin{align*}
& \log N\left(\varepsilon_{n,\mS};\bigcup_{(T,\bm K)\in E(Z)} \mF_{\mS}({\bm K})\cap\mF_n; \|\cdot\|_n\right)\\
 &\quad\leq \log\Delta(E(Z))+ \log \Delta(\mV\mE_{\mS}^{\bm K})+Z\log(B/\varepsilon_{n,\mS}C_1\kappa(\mE)\sqrt{n})\\
 &\quad \lesssim Z\log Z + \sqrt{Z} + Z\log (|\mS|n^2)+  Z\log \left(n^{2+\delta/2}\sqrt{ Z}\right)
 \end{align*}
 for some $C_1>0$, where $\Delta(\mV\mE_{\mS}^{\bm K})$ is the cardinality of  $\delta$-valid ensembles $\mV\mE_{\mS}^{\bm K}$.  Inside the sieve, we have $|\mS|< q_n\leq n$ and $Z< K_n\asymp n\,\varepsilon_{n,\mS_0}^2/\log n$ and thereby we can upper bound the log entropy with a constant multiple of $n\,\varepsilon_{n,
\mS_0}^2$. For an overfitting model $(\mS,E(Z))$ such that $Z\geq K(\mS_0)$ and $\mS\supset\mS_0$, the log-covering number  is further upper-bounded with  $n\,\varepsilon_{n,\mS}^2\geq n\,\varepsilon_{n,\mS_0}^2$. Next, we verify the following  variant of condition \eqref{almost_last_one}
\begin{equation}\label{bound111}
\sum_{\bG_n\cap\bG_{\{\mS\not\supset\mS_0\}\cup \{Z<K_{\mS_0}\}}}\e^{-C_2\,M^2\,n\,\varepsilon_{n,\mS_0}^2}+
\sum_{\bG_n\cap\bG_{\{\mS\supset\mS_0\}\cap \{Z\geq K_{\mS_0}\}}}\e^{-C_2n\,\varepsilon_{n,\mS}^2}\leq 1
\end{equation}
for some $C_2>0$.
Since $n\,\varepsilon_{n,\mS}^2>n\,\varepsilon_{n,\mS_0}^2$ for $\mS\supset\mS_0$ and $M>1$, we can upper-bound the left-hand-side with
$$
\e^{-C_2\,n\,\varepsilon_{n,\mS_0}^2}\sum_{q=0}^{q_n}{p\choose q}\sum_{Z=1}^{K_n}\Delta(E(Z))\lesssim
\left(\frac{2\,\e\,p}{q_n}\right)^{q_n+1}\e^{-C_2\,n\,\varepsilon_{n,\mS_0}^2+\log q_n+\log K_n +K_n\log K_n+\pi\sqrt{2K_n/3}},
$$
where we used the fact $\Delta(E(Z))\leq Z!p(Z)$ and \eqref{partition_bound}.
Since $K_n\log K_n\lesssim n\,\varepsilon_{n,\mS_0}^2$ and $q_n\log p\asymp n\varepsilon_{n,\mS_0}^2$, the right hand side goes to zero for a large enough constant $C_2>0$.

Lastly, the anti-concentration condition is replaced with
$$
\sum_{T=K_n}^n\pi(T)\sum_{{\bG_n\cap\bG_{\{\mS\supset\mS_0\}\cap \{Z\geq K_{\mS_0}\}}}}\sum_{\bm K\in\N^T:\sum_{t=1}^T K^t=Z}\pi(\mS,\bm{K}\C T)\leq\e^{-H\,n\,\varepsilon_{n,\mS_0}^2}
$$
for some $H>0$. Using the fact $\pi(\mS,\bm{K}\C T)\gtrsim \e^{-C\, \sum K^t\log n}$, we can upper-bound the left hand side above with
\begin{align*}
&\sum_{T=1}^{K_n}\pi(T)\e^{-C\, K_{\mS_0}\log n}\sum_{{\bG_n\cap\bG_{\{\mS\supset\mS_0\}\cap \{Z\geq K_{\mS_0}\}}}}\Delta(E(Z))\\
&\quad\quad\quad\quad\lesssim \e^{-C\,K_{\mS_0}\log n}\left(\frac{2\,\e\, p}{q_n}\right)^{q_n+1}\e^{2\log K_n+K_n\log K_n+\pi\sqrt{2K_n/3}-C_T}
\end{align*}
Using similar arguments as before, and because $K_{\mS_0}\log n\geq C_K/C_\varepsilon\,n\,\varepsilon_{n,\mS_0}^2$, the condition will be satisfied for large enough $C>0$ and $C_K>0$.

\subsection{Theory for ABC}\label{sec:theory_ABC}

First, we show the following ABC posterior concentration result.

\begin{theorem}\label{thm:abc1}
Under the assumptions of Theorem 4.1 and assuming $\sigma^2=1/n$ in (1), the naive ABC posterior satisfies  with $\P_{f_0}^{(n)}$ tending to one
$$
\Pi\left[\|f-f_0\|_n>\lambda_n\C \|\Y-\bm Y^\star\|_n\leq \epsilon^T_n\right]\lesssim 1/M
$$
for $\epsilon^T_n=\sqrt{2\log n/n}$, $\lambda_n=4\epsilon^T_n/3+1/\sqrt{n}$ and for any $M>0$ large enough.
\end{theorem}

\begin{proof}
We will be working conditionally on the event $\mathcal A=\{\bm \varepsilon=(\varepsilon_1,\dots,\varepsilon_n)':\max_{1\leq i\leq n}|\varepsilon_i|\leq \sqrt{2\log n/n}\}$  whose complement has a small probability, i.e. $\P_{f_0}^{(n)}[\mathcal A^c]\leq c_0/\sqrt{2\log n}$ for some $c_0>0$ when $\varepsilon_i\iid \mathcal N(0,1/n)$. On the event $\mA$, we have
$$
\|\Y-f_0\|_n=\sqrt{\frac{1}{n}\sum_{i=1}^n\varepsilon_i^2}\leq \sqrt{2\log n/n}\equiv\epsilon^T_n.
$$
We now define a joint event
$$
\mA(\epsilon^T_n,\lambda_n)\equiv\{(\bm Y^\star,f): \|\bm Y^\star-\Y\|_n\leq \epsilon^T_n\quad\text{and}\quad \|f-f_0\|_n>\lambda_n\}.
$$
For all $(\bm Y^\star,f)\in \mA(\epsilon^T_n,\lambda_n)$ we have
$$
\|f-f_0\|_n\leq \|\bm Y^\star-\Y\|_n+\|f-\bm Y^\star\|_n+\|f_0-\Y\|_n\leq \frac{4}{3}\epsilon^T_n+\|f-\bm Y^\star\|_n.
$$
This means that $(\bm Y^\star,f)\in \mA(\epsilon^T_n,\lambda_n)$  implies 
$
\|f-\bm Y^\star\|_n>\lambda_n-\frac{4}{3}\epsilon^T_n
$
and choosing $\lambda_n\geq \frac{4}{3}\epsilon^T_n+t_\varepsilon$ leads to
$$
\Pi[\mA(\epsilon^T_n,\lambda_n)]\leq \int \P_f[\|f-\bm Y^\star\|_n>t_{\varepsilon}]d\Pi(f)
$$
and
\begin{equation}\label{eq:bound}
\Pi\left[\|f-f_0\|_n>  \frac{4}{3}\epsilon^T_n+t_\varepsilon\,\big|\, \|\Y-\bm Y^\star\|_n\leq \epsilon^T_n\right]\leq\frac{\int \P_f[\|\bm Y^\star-f\|_n>t_{\varepsilon}]d\Pi(f)}{\int\P_f[\|\bm Y^\star-\Y\|_n\leq \epsilon^T_n]d\Pi(f)}.
\end{equation}
Now, we have for a random variable $\chi^2_n$ with a chi-square distribution with $n$ degrees of freedom
$$
\P_f[\|\bm Y^\star-f\|_n>u]=\P_f\left[\frac{\chi^2_n}{n^2}>u^2\right]=\P_f\left[\e^{\chi^2_n/4}>\e^{u^2\,n^2/4}\right]\leq \frac{2^{n/2}}{\e^{u^2\,n^2/4}}.
$$
Next, for $n$ large enough we can write
\begin{align}
\int\P_f[\|\bm Y^\star-\Y\|_n\leq \epsilon^T_n]d\Pi(f)&\geq\int_{\|f-f_0\|_n\leq \epsilon^T_n/3}\P_f[\|\bm Y^\star-f\|_n\leq\epsilon^T_n/3]d\Pi(f)\\
&\geq \Pi[\|f-f_0\|_n\leq \epsilon^T_n/3]-\e^{n/2\log 2-n\,\log n/18}\\
&\geq \Pi[\|f-f_0\|_n\leq \epsilon^T_n/3]/2.
\end{align}
Next (under the assumption $q_0\log p<n^{q_0/(2\alpha+q_0)}$, we have $\pi(\mS_0)\geq \e^{-n\varepsilon_{n,\mS_0}^2}$ and (assuming $K=K_{\mS_0}\asymp n\varepsilon_{n,\mS_0}^2/\log n$ and denoting $\wh \b\in\R^K$ the steps of the $\|\cdot\|_n$ projection of $f_0$ onto trees with $K$ leafs)  for some $c>0$
\begin{align}
\Pi[\|f-f_0\|_n\leq \epsilon^T_n/3]&>\e^{-n\,\varepsilon_{n,\mS_0}^2}\Pi(\|\b-\wh\b\|_2\leq \epsilon^T_n/6)\\
&>
\e^{-n\,\varepsilon_{n,\mS_0}^2}\frac{\e^{-K\log 2-\|\wh\b\|_2^2-(\varepsilon_n^T)^2/72+K/2\log[(\varepsilon_n^T)^2/36]}}{\Gamma(K/2)K/2}>\e^{-c\,n\,\varepsilon_{n,\mS_0}^2}.
\end{align}
We can now upper-bound \eqref{eq:bound} with 
$2^{n/2}\e^{-t_\varepsilon^2\,n^2/4+cn\,\varepsilon_{n,\mS_0}^2}$ which is smaller than an arbitrary constant $M>0$ for $n$ large enough if we choose $t_\varepsilon=1/\sqrt{n}$.
\end{proof}

Given this consistency result, we can immediately conclude (using the inequality in (21) in the paper) that the ABC posterior will not reward underfitting model as long as our identifiability and irrepresentability conditions are satisfied with $\varepsilon=\lambda_n$. In other words, under the assumptions of Theorem \ref{thm:abc1}  and assuming that $\mS_0$ is $(f_0,\lambda_n)$-identifiable and that $\lambda_n$-irrepresentability holds we have, with $\P_{f_0}$ tending to one and for any $M>0$,
$$
\Pi\left[ \mS\not\supset\mS_0 \C \|\Y-\bm Y^\star\|_n\leq \epsilon^T_n\right]\lesssim 1/M.
$$

Regarding over-fitting models, we first show the following ABC analogue of Lemma 8.1.
We can write, on the event $\mA$, and for $q_n=C_q\lceil n\,\varepsilon_{n,\mS_0}^2/\log p\rceil$ (as in Lemma 8.1)
$$
\Pi_1\equiv\Pi\left[q\geq q_n\C \|\Y-\bm Y^\star\|_n\leq \epsilon^T_n\right]=\sum_{\mS:|\mS|\geq q_n}\pi(\mS)\frac{\int\P_f[\|\Y-\bm Y^\star\|_n\leq \epsilon^T_n]d\Pi(f\C\mS)}{\int\P_f[\|\Y-\bm Y^\star\|_n\leq \epsilon^T_n]d\Pi(f)}.
$$
It turns out from the proof of Theorem  \ref{thm:abc1} that 
$$
\Pi_1\leq \frac{\sum_{q\geq q_n}\sum_{\mS:|\mS|=q}\pi(\mS)}{\int\P_f[\|\Y-\bm Y^\star\|_n\leq \epsilon^T_n]d\Pi(f)}\lesssim \e^{c\,n\varepsilon_{n,\mS_0}^2}\Pi(q\geq q_n).
$$
In the proof of Lemma 1.1 we have already showed (under the assumptions of Theorem 4.1) that $\Pi(q\geq q_n)\lesssim \e^{-n\,\varepsilon_{n,\mS_0}^2C}$ for some $C>0$. Choosing $C_q$ large enough, one concludes that $\Pi_1\rightarrow 0$ as $n\rightarrow\infty$. This shows that the ABC posterior concentrates on the sieve of models $\mathcal F_n$ with up to  $q_n$ covariates. Using this result, we can focus on models of size up to $q_n$ and show that the posterior probability of over-fitting models goes to zero. Indeed, on the event $\mA$ and on $\mathcal F_n$  we have (using an inequalities (4) and (6))
\begin{align*}
\Pi\left[ \{\mS\supset\mS_0\}\cap\mathcal F_n \C \|\Y-\bm Y^\star\|_n\leq \epsilon^T_n\right]
&\leq\frac{\sum_{\mS\supset\mS_0:|\mS|\leq q_n}\pi(\mS)}{\int\P_f[\|\Y-\bm Y^\star\|_n\leq \epsilon^T_n]d\Pi(f)}\\
&\lesssim  \e^{(c-C)\,n\varepsilon_{n,\mS_0}^2}\left(\frac{2\e p}{q_n}\right)^{q_n+1}\lesssim \e^{-H\,n\varepsilon_{n,\mS_0}^2}
\end{align*}
for some $H>0$ with $C>0$ is large enough. This concludes that the ABC posterior will lead to consistent variable selection as well.

We now discuss how the theory can be extended when data-splitting is deployed in ABC. First, we discuss the case when the split is done only once before applying ABC (not internally at each iteration). Denote with $n_1$ the training sample size and with $n_2$ the validation sample size. 
In order for the consistency result in Theorem \ref{thm:abc1} to hold, we need to make sure that prior concentration holds in the sense that $\Pi[\|f-f_0\|_{n_2}\lesssim \epsilon^T_{n_2}]\geq \e^{-c\,n_2\epsilon_{n_2,\mS_0}^2}$ for some $c>0$. Leaving $n_1$ data-points for training the prior, we know (from results in RP17 under fixed $\sigma^2$) that the posterior concentrates at the optimal rate (up to a log factor), i.e.
$$
\Pi[|f-f_0\|_{n_1}\lesssim\epsilon_{n_1,\mS_0}\C \Y^{(n)}_{\mathcal I},\mS_0]\rightarrow 1 \quad\text{as $n_1\rightarrow\infty$}.
$$
Choosing $n_1$ and $n_2$ in such a way so that  $\epsilon_{n_1,\mS_0} \lesssim \epsilon^T_{n_2}\equiv \sqrt{2(\log n_2)/n_2}$ (and assuming that observed fixed covariates  in the training and testing sets are close), the prior concentration condition will be satisfied and the ABC will be consistent and concentrate at the rate $\lambda_{n_2}$. This implies variable selection consistency of our ABC method under identifiability and irrepresentabilty conditions which depend on $\lambda_{n_2}$.
A similar conclusion is obtained for the expected posterior prior \eqref{eq:expected_prior} where 
$$
\Pi[|f-f_0\|_{n_1}\lesssim\epsilon_{n_1,\mS_0}]\geq \pi(\mS_0)\frac{1}{L}\sum_{l}\Pi[\|f-f_0\|_{n_1}\lesssim\epsilon_{n_1,\mS_0}\C \Y^{(n)}_{\mathcal I_l},\mS_0]\gtrsim\pi(\mS_0).
$$
A rigorous proof of ABC consistency for the expected posterior priors would require more care and will be left for future investigation.

	\begin{table}
	\caption{\label{tab:time_BART} \small Computation time of $ 1\,000$ MCMC iterations of BART/DART in seconds (using the Friedman's datasets with $\sigma=5$ and autocorrelation of $0.9$).}
	\centering
	\scalebox{0.7}{
\begin{tabular}{l l | *{8}{c}}
\toprule
         &  & \multicolumn{4}{c}{BART}           & \multicolumn{4}{c}{DART} \\
          \cmidrule( lr){3-6}\cmidrule(lr){7-10}

 & & $T=10$ & $T=20$ & $T=50$ & $T=200$ & $T=10$ & $T=20$ & $T=50$ & $T=200$ \\
\midrule
\multirow{3}{*}{$n=100$}& $p=100$   & 0.21  & 0.32  & 0.54  & 1.86   & 0.55  & 0.64  & 0.76  & 2.27   \\
&$p=1\,000$  & 0.53  & 0.67  & 1.57  & 5.48   & 0.96  & 0.98  & 1.91  & 5.79   \\
&$p=10\,000$& 3.56  & 5.58  & 10.91 & 39.16  & 5.93  & 7.72  & 12.55 & 36.95  \\
\hline
\multirow{3}{*}{$n=250$}&$p=100$    & 0.21  & 0.34  & 0.79  & 2.99   & 0.41  & 0.56  & 0.99  & 3.19   \\
&$p=1\,000$  & 0.50  & 0.82  & 1.81  & 6.51   & 0.87  & 1.14  & 1.83  & 6.58   \\
&$p=10\,000$ & 3.70  & 5.63  & 11.38 & 40.48  & 6.29  & 8.06  & 12.97 & 39.13  \\
\hline
\multirow{3}{*}{$n=500$}&$p=100$    & 0.29  & 0.53  & 1.23  & 4.93   & 0.49  & 0.71  & 1.40  & 5.07   \\
&$p=1\,000$ & 0.63  & 1.11  & 2.36  & 8.65   & 1.01  & 1.30  & 2.29  & 7.89   \\
&$p=10\,000$ & 4.21  & 6.54  & 12.27 & 43.35  & 6.80  & 8.35  & 13.46 & 37.81  \\
\hline
\multirow{3}{*}{$n=1\,000$}&$p=100$    & 0.53  & 0.97  & 2.22  & 8.86   & 0.71  & 1.13  & 2.30  & 8.95   \\
&$p=1\,000$  & 0.91  & 1.49  & 3.32  & 12.82  & 1.24  & 1.84  & 3.54  & 12.12  \\
&$p=10\,000$ & 5.41  & 8.18  & 14.23 & 48.92  & 7.61  & 9.06  & 14.47 & 41.90  \\
\hline
 \multirow{3}{*}{$n=10\,000$}& $p=100$   & 7.17  & 10.78 & 23.25 & 82.45  & 6.37  & 12.07 & 22.33 & 92.23  \\
 & $p=1\,000$ & 13.00 & 22.12 & 34.67 & 125.98 & 12.76 & 16.95 & 40.25 & 102.72 \\
 & $p=10\,000$ & 25.35 & 31.39 & 59.71 & 218.08 & 28.59 & 39.93 & 73.99 & 171.73\\
\bottomrule
\end{tabular}}

\end{table}

\section{ABC Computational Feasibility}\label{supplement:time}

Regarding computational considerations, our sampling method deploys MCMC inside each ABC iteration but uses only on a subset of the original observations)
	(say $\frac{n}{2}$ observations) and a subset of $|\mS|<p$ variables. In addition, we only need to collect one posterior sample after a burnin period $B$.
		
	In order to understand how ABC scales with $\abs{\mS}$, $p$ and $s$, we first assess the computing time of plain  BART/DART. The timing comparisons are summarized in Table \ref{tab:time_BART}. From these computations we can conclude, for example, that running 
$M=1\,000$ BART iterations with $T=200$ trees (the default) on a dataset with $p=10\,000$ variables and $n=500$ observations takes $43.35$ seconds which roughly amounts to running $43.35\times 5/0.5=433.5 $ ABC iterations with $B=200$ burnin MCMC  iterations, $T=10$ trees and with $s=n/2$, assuming that the sparsity prior is such that $|\mS|\approx 1\,000$.
Under the same settings but a stricter sparsity prior such that $|\mS|\approx 100$, we obtain $43.35\times 5/0.21= 1032.14$ ABC iterations for the same time as $1\,000$ BART iterations. These computing times, however, do not take into account autocorrelation in BART samples, where $M=1\,000$  BART MCMC iterations do not necessarily yield  $1\,000$  {\em effective} samples. One advantage of  ABC sampling over MCMC is that it is embarrassingly parallel and that it does not incur correlation. This provides an opportunity for large speedups using parallel computing.

\section{Spike-and-Forests: MCMC Variant}\label{sec:mcmc}
As a precursor to ABC Bayesian Forests, we first implemented an MCMC algorithm for joint sampling from a posterior $\Pi(\mS,\mE\C\Y^{(n)})$ over the space of models and tree ensemble partitions. We refer to this algorithm as Spike-and-Forests. The sampling  follows a Metropolis-Hasting scheme, exploiting the additive structure of forests by sampling each tree individually from conditionals in a Gibbs manner within each Metropolis step (Bayesian backfitting by \citet{chipman2010bart}). 
The key is assigning a joint proposal distribution  $pr(\mS,\mE\C\mS_m,\mE_m)= pr(\mS\C\mS_m)pr(\mE\C\mS,\mE_m)$ over variable subsets $\mS$ and  partition ensembles $\mE$, where $\mS_m$ and $\mE_m$ are current MCMC states.

We explain the proposal mechanism using a single tree and write $\mT$ instead of $\mE$.
First, a model proposal $\mS^\star$ is sampled from $pr(\mS\C\mS_m)$ which consists of the following three options:  $\texttt{add},\texttt{delete}$ and $\texttt{stay}$ for adding/deleting one (or none) of the variables. These three steps are chosen with probabilities $0.4, 0.4$ and $0.2$, respectively. Candidate variables for deletion/addition are chosen from a uniform distribution.
Given the newly suggested model $\mS^\star$, the proposal distribution $pr(\mT\C\mS^\star,\mT_m)$ consists of various moves, described below, depending on the status of $\mS^\star$. 

If 
$\mS^\star$ was obtained from $\mS_m$ by $\texttt{adding}$ a variable,  the proposal $pr(\mT\C\mS^\star=\texttt{add},\mT_m)$ consists of two steps: \texttt{birth} and \texttt{replace}.
 In the \texttt{birth} step,  a bottom node  is added to $\mT_m$ and in the  $\texttt{replace}$ step one of the variables that occurs more than once inside  $\mT_m$ is replaced with the new variable. The birth step increases the size of the tree, while the replace step does not.  The two steps are chosen with  probabilities
\begin{align*}
\pi_\text{birth,add}=0.7\min\left\{\frac{\pi(K+1)}{\pi(K)}, 1\right\}, \pi_\text{birth, replace}=1-\pi_\text{birth,add},
\end{align*}
where $K$ is the number of bottom nodes in $\mT_m$ and $\pi(K)$ is a prior on the number of bottom nodes.
 If no variable appears more than once in the tree, then \texttt{replace} is invalid and $\pi_\text{birth, replace}$  is set to $0$.

If $\mS^\star$ is obtained from $\mS_m$ by \texttt{deleting} a variable,   the proposal $pr(\mT\C\mS^\star=\texttt{delete},\mT_m)$ consists of two steps: \texttt{death} and \texttt{replace}.
If the variable chosen for deletion occurs in a bottom node, it can be removed from a tree $\mT_m$ with a \texttt{delete} step that erases the bottom node. If the variable occurs inside the tree, it can be deleted by replacing it with other variables in the \texttt{replace} step. If both of these moves are eligible,  we pick one of them with probabilities
\begin{align*}
\pi_\text{death,delete}=0.7\min\left\{\frac{\pi(K-1)}{\pi(K)},1 \right\}, \pi_\text{death,replace}=1-\pi_{\text{death,delete}}.
\end{align*}
If the variable suggested for deletion is not in a bottom node, then $\pi_\text{death,delete}=0$.

If the pool of variables stays the same, i.e. $\mS^\star=\mS_m$,  the proposal $pr(\mT\C\mS^\star=\texttt{stay},\mT_m)$ consists of 4 moves: \texttt{add, delete, replace} and \texttt{rule}.
All proposal moves, and their probabilities, are adopted from Bayesian CART of \citet{denison1998bayesian}. These steps only modify the tree configuration without  adding/deleting variables.

Regarding the prior distributions for our MCMC implementation, we assume the beta-binomial prior on the variable subsets. Namely, for binary indicators $\gamma_j\in\{0,1\}$, for whether or not $x_j$ is active, we assume $\P(\gamma_j=1\C\theta)=\theta$ and $\theta\sim\mathcal{B}(a,b)$. The prior distribution on trees consists of (a) the truncated Poisson distribution on the number of bottom leaves, (b) uniform prior over trees with the same number of leaves  and (c) standard Gaussian prior on the step sizes. This is the Bayesian CART prior proposed by \citet{denison1998bayesian} and analyzed theoretically by \citet{rockova2017posterior}. In the computation of MH acceptance ratios, we leverage the fact that the bottom leave parameters can be integrated out to obtain a conditional marginal likelihood, given each partition. 


The MCMC sampling routine can be extended to spike-and-forests, altering each tree inside the forests one by one through Bayesian backfitting \citep{chipman2010bart}.
One big advantage of the Bayesian forest representation is that it accelerates mixing since most trees are shallow and thereby more easily modified throughout MCMC (see \citet{pratola2016efficient}).

\section{Sensitivity Analysis}

Our sensitivity analysis focuses on two aspects. First, we want to  assess how the choices of $M$ (the number of ABC samples), $T$ (the number of trees in each forest), $B$ (the number of burn-in iterations inside each ABC iteration) and  $\epsilon$ (tolerance for ABC acceptance) collaboratively impact ABC variable importance. Second, we want to investigate the impact of different data splitting strategies, including varying choices of $s$ (proportion of data used in training) and pre-determined data splitting versus internal data splitting.  There is an obvious tradeoff between $s$ and $M$, where   small  $s$    will yield fewer ABC pseudo-observations that are compatible with the observed data and $M$  will thereby  have to be larger.
We have considered the following combinations  
		\[
		M\in \{1\,000, 10\,000\} \times T \in\{10, 25, 50\} \times B\in\{200,1\,000\}  \times \epsilon \in \{top \, 1\%, 5\%, 10\% \}
		\]
		
These comparisons are conducted using the Friedman's simulation setup with $p\in \{100,1\,000\}, \rho=0.9 \text{(autoregressive)}$ and $\sigma=5$, assuming $s=n/2$ {and internal splitting} for ABC.  We also include various sample sizes $n\in\{100,500,1\,000\}$ for each $p$. For each setting, we show  ABC  inclusion probabilities (\texttt{ip})  for the first $30$ variables of which only the first $5$ are active (Figure \ref{fig:ip_friedman}). We denote the parameters for each ABC setup by $T\star B$ where, for example, $20\star200$ means each forest consists of $T=20$ trees and uses $B=200$ MCMC iterations as a burnin.

 		\begin{figure}[!t]
				\centering
\includegraphics[width=\textwidth]{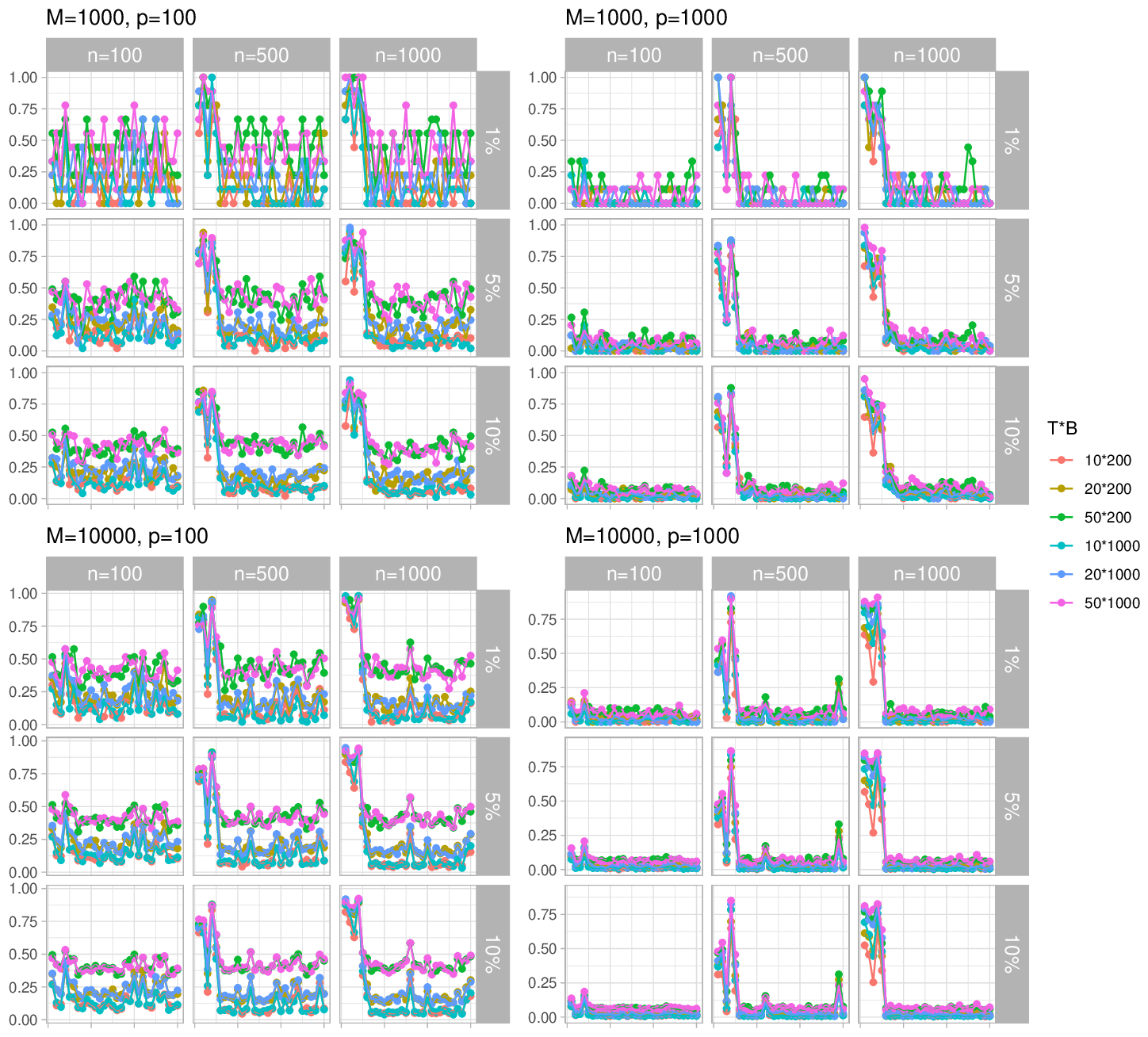}
			\caption{ABC inclusion probabilities of the first 30 variables over different $\epsilon$. 
			Each panel corresponds to a different combination of $p\in\{100,1\,000\}$ and $M\in\{1\,000,10\,000\}$. Each row indicates a different model averaging strategy based on a different $\epsilon$ value. Each column corresponds to a different sample size. The legend represents various combinations of $T\star B$. For example, $20*200$ means each forest consists of $T=20$ trees and  $B=200$ MCMC iterations as burnin. Note that we use $s=n/2$ here.}\label{fig:ip_friedman}
		\end{figure}

		From the figures we can see that ABC is more sensitive to the choice of $T$ than to the choice of $B$. This is not entirely unexpected. As suggested in \citet{chipman2010bart} and \citet{bleich2014variable},  a large value of $T$ allows for increased flexibility in fitting the model while smaller $T$ should be adopted for the purpose of variable selection. The variables must compete with each other to be included when $T$ is small.  In terms of a median probability model, the model tends to have more power and higher false discoveries when $T$ is large, and less power and fewer false discoveries when $T$ is small.

\begin{figure}[!t]
				\centering
\includegraphics[width=\textwidth]{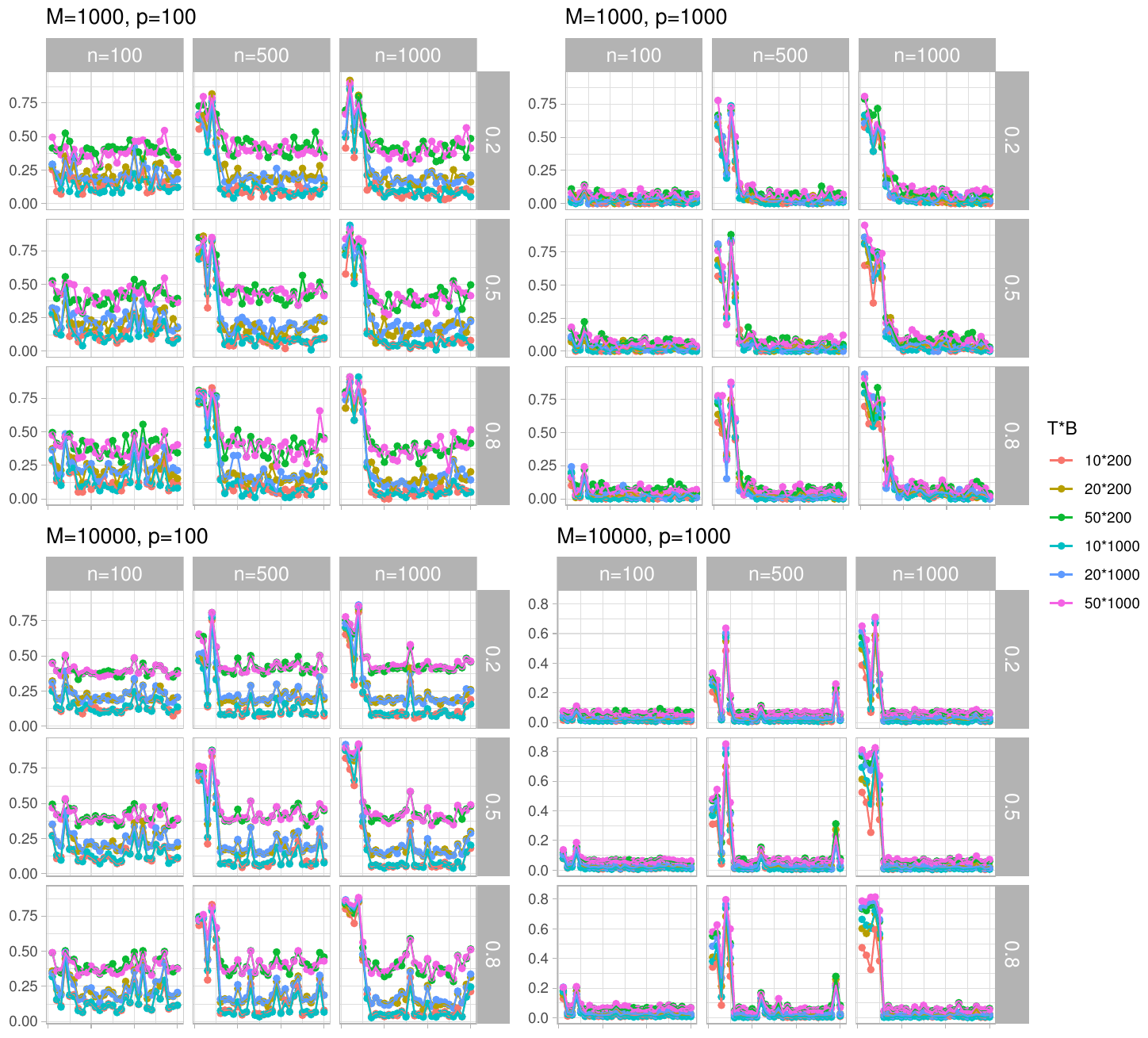}
			\caption{ABC inclusion probabilities of the first 30 variables over different $s$. 
			Each panel corresponds to a different combination of $p\in\{100,1\,000\}$ and $M\in\{1\,000,10\,000\}$. Each row indicates a different model averaging strategy based on a different ratio of $s$ over $n$. Each column corresponds to a different sample size. The legend represents various combinations of $T\star B$. For example, $20*200$ means each forest consists of $T=20$ trees and  $B=200$ MCMC iterations as burnin. Note that we use $\epsilon$=\{top 10\%\} here.}\label{fig:s_comp}
		\end{figure}

		Regarding $\epsilon$, although the trends are similar for top 1\%, 5\% and 10\% selected model, higher variance is observed for smaller tolerance when $M$ is not large enough, especially for $M=1\,000$ with top 1\% models accepted. This is, again, not entirely unexpected.

		The comparisons in Figure \ref{fig:ip_friedman} were done assuming $s=n/2$. We now consider a similar simulation study, but for $\epsilon=\{ \text{top } 10\%\}$ and  various $s$  by considering 
		\[
		M\in \{1\,000, 10\,000\} \times T \in\{10, 25, 50\} \times B \in\{200,1\,000\}  \times s \in \{n/5,n/2, 4n/5\}.
		\]

The results are displayed in Figure \ref{fig:s_comp}.
The posterior inclusion probabilities do not seem to vary much with respect to $s$. This suggests that even $s=0.2n$ provides reasonable prior guesses for ABC regarding variable selection.  
Based on this sensitivity analysis, we choose $T=20, B=200, M=1\,000, s=n/2, \epsilon=\{ \text{top } 10\%\}$ as the default parameters for our ABC model. 
 		
The last part of the sensitivity analysis we want to investigate the differences between  pre-determined data splitting and internal data splitting. 
Customarily \citep{berger2004training}, the subsample size $s$  is chosen as the minimal number of samples needed to convert an improper prior into a proper one. Our situation, however, is different in at least three aspects: (a)  we are  converting a proper uninformative prior into an informative one, (b)  our model is entirely non-parametric and (c) we aim to enhance ABC acceptance rate rather than using non-informative priors for model selection with Bayes factors.  As pointed out in \citet{berger2004training}, defining any optimal training sample is very challenging and one needs to exercise statistical judgment  to select from among various strategies. While \citet{berger2004training} argue that: ``Judgments involved in choosing good training samples will typically be much less than the judgments needed to implement an actual subjective Bayesian analysis", we argue that entertaining some reasonable form of the data splitting (even if not optimal) will provide better results than naive ABC strategy in our context. 
 The following simulated example shows that the variable selection performance with internal splitting is at least as good as with pre-determined splitting. We still use the Friedman's dataset with $n=500$, $p=100$ and $p=1000$, $\sigma=5$ and autocorrelation $0.9$. The ABC settings are $T=20, \theta=0.5, s=0.5n,\epsilon=$ top 10\%. The  inclusion probabilities are averaged over $10$ datasets and plotted in Figure \ref{fig:fsrs_comp}.
\begin{figure}[!ht]
\centering
\includegraphics[width=0.8\textwidth]{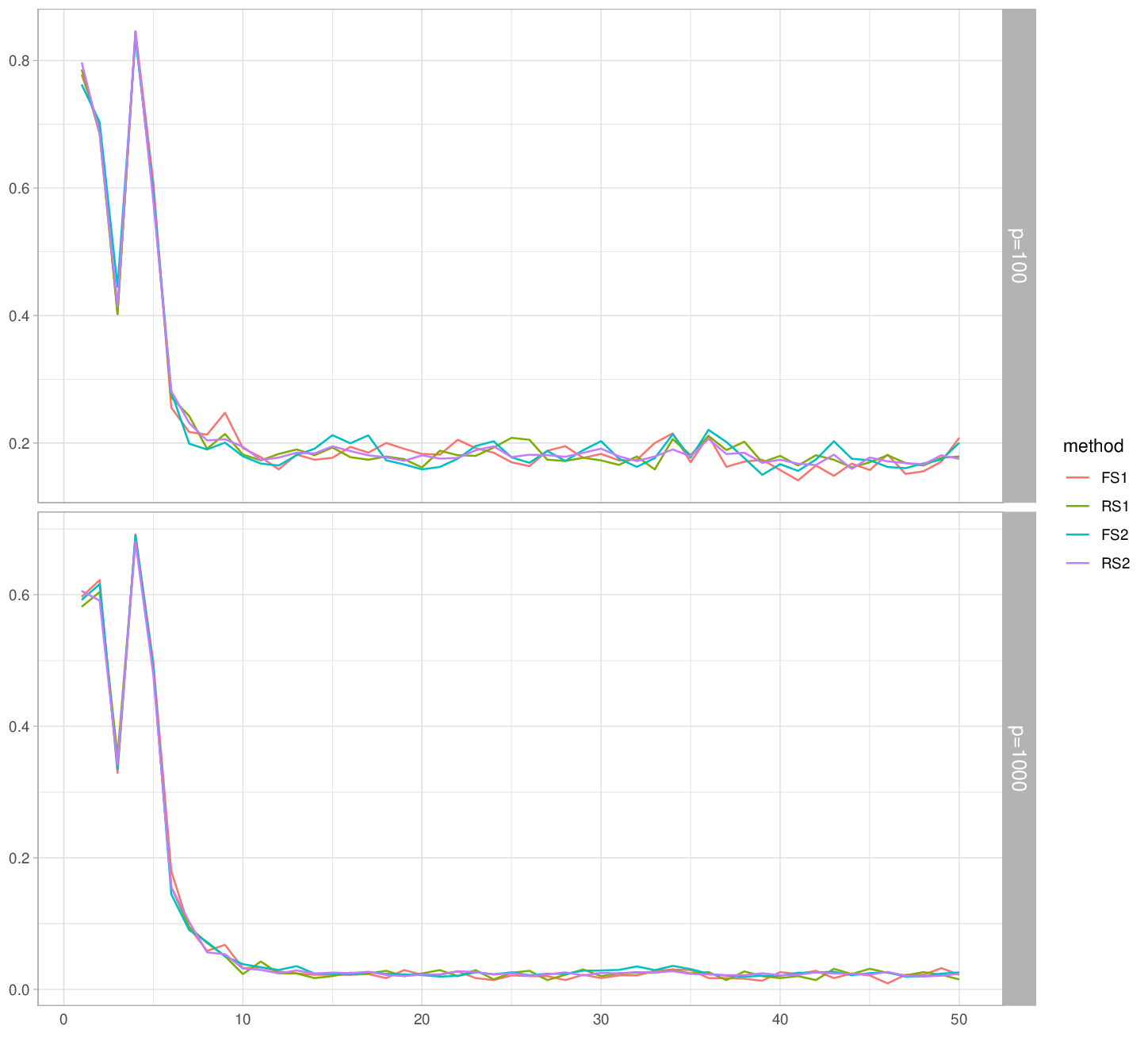}
\caption{Comparison of Inclusion Probabilities of Pre-determined Splitting (FS) and Internal Splitting (RS). The inclusion probabilities are averaged over $10$ independent Friedman's datasets ($n=500, \sigma=5$, autocorrelation = $0.9$). FS1/RS1 are built with $M=1\,000$, and FS2/RS2 are built with $M=10\,000$ }\label{fig:fsrs_comp}
\end{figure}
From  Figure \ref{fig:fsrs_comp}, we can see that the differences in inclusion probabilities of pre-determined splitting and internal splitting are small. This could be explained by the fact the data have been generated with Gaussian noise without any outliers which could potentially affect quality of splits. 

Combining our findings from all of the sensitivity analyses above, we recommend the following default settings for the  parameters: $s=0.5n, T=20, \text{burnin}=200, \epsilon= \text{top }10\%$, $M=1,000$ and internal data splitting.

\section{Full HIV Data Analysis}\label{supplement:hiv}
In this section, we  provide a summary of our results on the entire dataset from \citet{barber_controlling_2015}. The summary statistics of the data are reported in Table \ref{tab:HIV_stats}.
Comparisons are made between ABC Bayesian Forests, BART, DART and Random Forests. BART and DART are run with $50$ trees for $20\,000$ MCMC iterations  (taking the first $10\,000$ as a burn-in).  Random Forests are implemented with the default number of  $500$ trees. 

\begin{table}
		\caption{\label{tab:HIV_stats} Basic summary statistics of the HIV dataset. DS refers to the decrease in susceptibility of the drug once the mutations has occurred.}
	\centering
\begin{tabular}{lllll}
	HIV Virus Life Cycle   & Drug Class  & Mean Log DS & Number of Features & Number of Samples \\ \hline
	\multirow{7}{*}{PI}    & APV         & 0.75     & 201                & 767               \\
						   & ATV         & 1.59     & 147                & 328               \\
						   & IDV         & 1.33     & 206                & 825               \\
						   & LPV         & 1.74     & 184                & 515               \\
						   & NFV         & 2.00     & 207                & 842               \\
						   & RTV        & 1.72     & 205                & 793               \\
						   & SQV        & 1.22     & 206                & 824               \\ \hline
	\multirow{6}{*}{NRTI}  & X3TC       & 3.10     & 283                & 629               \\
						   & ABC        & 1.14     & 283                & 623               \\
						   & AZT        & 1.55     & 283                & 626               \\
						   & D4T        & 0.43     & 281                & 625               \\
						   & DDI        & 0.43     & 283                & 628               \\
						   & TDF        & 0.22     & 215                & 351               \\ \hline
	\multirow{3}{*}{NNRTI} & DLV        & 0.98     & 305                & 730               \\
						   & EFV        & 1.08     & 312                & 732               \\
						   & NVP        & 1.80     & 313                & 744               \\ \hline
	\end{tabular}

\end{table}

To summarize the results, we adopted $2$ cutoff selection criteria. 
The first selection  threshold is adaptive and is chosen as the maximum importance measure of a non-experimentally validated mutation.
This cutoff point corresponds to zero false discoveries. Next, we use   an automatic criterion for each method.
 For ABC Bayesian Forests (run with $T=20$ trees and $M=200$ burnin iterations, $10\,000$ ABC samples and top $100, 500$ and $1\,000$ samples with the smallest discrepancy), we adopted the median probability model with the $0.5$ cutoff. For DART and BART, we choose variables which have been split on at least once on average.  For Random forest, the RFE approach (as described in \cite{linero2018bayesian})
 is used to find the variables. Similarly as in \cite{barber_controlling_2015}, we report the number true positions discovered and the number of false positions. To further study the separation power, we also report AUC of each method.  The results are shown in Table \ref{tab:result_PI}, \ref{tab:result_NRTI} and \ref{tab:result_NNRTI}.

 Across all the drugs, we notice that ABC Bayesian Forest has a strong separation power, as  is indicated by the performance of AUC scores. Random Forests with RFE tends to overfit  by selecting too many mutations. BART and DART are performing well in this case but ABC is seen to have better AUC while  being overall more conservative.


\begin{table}
\caption{\label{tab:result_PI} The table summarizes results for a drug class PI. There are three performance criteria. For the adaptive cutoff, we report the number of true discoveries since the number of false discoveries is 0. For the automatic cutoff, we report both the number of false and true discoveries. Finally, we report a cutoff-free metric AUC. The best performance in each row is in bold font.}
	\centering
		
\scalebox{0.8}{\begin{tabular}{llllllll}
		\hline
		\multicolumn{8}{c}{APV}                                                                                                                                 \\ \hline
		\multicolumn{2}{l}{\multirow{2}{*}{Methods}} & \multicolumn{3}{c}{ABC} & \multirow{2}{*}{BART} & \multirow{2}{*}{DART} & \multirow{2}{*}{Random Forest} \\
		\multicolumn{2}{l}{}   & 100    & 500    & 1000  &                       &                       &                                \\ \hline
		Adaptive cut-off       & True Discoveries    & 17     & \textbf{19}     & \textbf{19}    & 14                    & 15                    & 15                             \\
		Automatic cut-off      & False Discoveries   & \textbf{0}      & \textbf{0}      & \textbf{0}     & \textbf{0}                     & 7                     & 31                             \\
		& True Discoveries    & 13     & 11     & 11    & 14                    & 20                    & \textbf{34 }                            \\
		AUC                    &                     & 0.69   & 0.75   & \textbf{0.77}  & 0.65                  & 0.65                  & 0.61                           \\ \hline
		\multicolumn{8}{c}{ATV}                                                                                                                                 \\ \hline
		Adaptive cut-off       & True Discoveries    & \textbf{23}     & \textbf{23}     & \textbf{23}    & 19                    & 19                    & 13                             \\
		Automatic cut-off      & False Discoveries   & \textbf{0}      & \textbf{0}      & \textbf{0}     & \textbf{0}                     & 3                     & \textbf{0}                              \\
		& True Discoveries    & 16     & 15     & 15    & 18                    & \textbf{21}                    & 19                             \\
		AUC                    &                     & 0.77   & 0.78   & \textbf{0.79}  & 0.62                  & 0.65                  & 0.71                           \\ \hline
		\multicolumn{8}{c}{IDV}                                                                                                                                 \\ \hline
		Adaptive cut-off       & True Discoveries    & 8      & 9      & 9     & 6                     & 11                    & \textbf{13 }                            \\
		Automatic cut-off      & False Discoveries   & \textbf{1}      & \textbf{1}      & \textbf{1}     & 2                     & 5                     & 32                             \\
		& True Discoveries    & 14     & 14     & 14    & 18                    & 18                    & \textbf{34}                             \\
		AUC                    &                     & 0.73   & \textbf{0.75}   & \textbf{0.75}  & 0.65                  & 0.63                  & 0.62                           \\ \hline
		\multicolumn{8}{c}{LPV}                                                                                                                                 \\ \hline
		Adaptive cut-off       & True Discoveries    & 14     & 14     & 14    & \textbf{15}                    & 13                    & 9                              \\
		Automatic cut-off      & False Discoveries   & \textbf{0}      & \textbf{0}      & \textbf{0}     & \textbf{0}                     & 7                     & 31                             \\
		& True Discoveries    & 13     & 13     & 13    & 14                    & 17                    & \textbf{34}                             \\
		AUC                    &                     & 0.72   & 0.74   & \textbf{0.75}  & 0.56                  & 0.57                  & 0.62                           \\ \hline
		\multicolumn{8}{c}{NFV}                                                                                                                                 \\ \hline
		Adaptive cut-off       & True Discoveries    & 8      & 10     & 10    & 11                    & \textbf{16}                    & 15                             \\
		Automatic cut-off      & False Discoveries   & \textbf{1}      & \textbf{1}      & \textbf{1}     & \textbf{1}                     & 5                     & 32                             \\
		& True Discoveries    & 15     & 15     & 14    & 17                    & 20                    & \textbf{34}                             \\
		AUC                    &                     & 0.73   & \textbf{0.74}   & \textbf{0.74}  & 0.65                  & 0.64                  & 0.65                           \\ \hline
		\multicolumn{8}{c}{RTV}                                                                                                                                 \\ \hline
		Adaptive cut-off       & True Discoveries    & 10     & 10     & 9     & \textbf{13 }                   & 11                    & 11                             \\
		Automatic cut-off      & False Discoveries   & 2      & \textbf{1}      &\textbf{ 1}     & 3                     & 4                     & 31                             \\
		& True Discoveries    & 13     & 11     & 11    & 14                    & 20                    & \textbf{34}                             \\
		AUC                    &                     & 0.72   & 0.74   & \textbf{0.75}  & 0.62                  & 0.60                  & 0.67                           \\ \hline
		\multicolumn{8}{c}{SQV}                                                                                                                                 \\ \hline
		Adaptive cut-off       & True Discoveries    & 15     & 15     & 15    & 3                     & \textbf{17}                    & 10                             \\
		Automatic cut-off      & False Discoveries   & \textbf{0 }     & \textbf{0}      & \textbf{0}     & 3                     & 6                     & 31                             \\
		& True Discoveries    & 15     & 15     & 14    & 16                    & 17                    & \textbf{34}                             \\
		AUC                    &                     & 0.74   & 0.77   & \textbf{0.78}  & 0.64                  & 0.62                  & 0.57                           \\ \hline
\end{tabular}}

\end{table}

\begin{table}
\caption{ \label{tab:result_NRTI} The table summarizes results for a drug class NRTI. There are three performance criteria. For the adaptive cutoff, we report the number of true discoveries since the number of false discoveries is 0. For the automatic cutoff, we report both the number of false and true discoveries. Finally, we report a cutoff-free metric AUC. The best performance in each row is in bold font.}

	\centering

\scalebox{0.8}{\begin{tabular}{llllllll}
		\hline
		\multicolumn{8}{c}{X3TC}                                                                                                                                \\ \hline
		\multicolumn{2}{l}{\multirow{2}{*}{Methods}} & \multicolumn{3}{c}{ABC} & \multirow{2}{*}{BART} & \multirow{2}{*}{DART} & \multirow{2}{*}{Random Forest} \\
		\multicolumn{2}{l}{}   & 100    & 500    & 1000  &                       &                       &                                \\ \hline
		Adaptive cut-off       & True Discoveries    & 6      & \textbf{9}      & \textbf{9}     & 4                     & 5                     & 6                              \\
		Automatic cut-off      & False Discoveries   & \textbf{0}      & \textbf{0}      & \textbf{0}     & 4                     & 3                     & 6                              \\
		& True Discoveries    & 6      & 5      & 5     & 7                     & 12                    & \textbf{15}                             \\
		AUC                    &                     & \textbf{0.70}   & \textbf{0.70}   & \textbf{0.70}  & 0.62                  & 0.64                  & 0.66                           \\ \hline
		\multicolumn{8}{c}{ABC}                                                                                                                                 \\ \hline
		Adaptive cut-off       & True Discoveries    & 8      & 8      & 7     & 7                     & 10                    & \textbf{12}                             \\
		Automatic cut-off      & False Discoveries   & 2      & \textbf{1}      & 1\textbf{}     & \textbf{1}                     & 7                     & 2                              \\
		& True Discoveries    & 10     & 10     & 10    & 11                    & 14                    & \textbf{16}                             \\
		AUC                    &                     & 0.74   & 0.73   & \textbf{0.76}  & 0.66                  & 0.71                  & 0.74                           \\ \hline
		\multicolumn{8}{c}{AZT}                                                                                                                                 \\ \hline
		Adaptive cut-off       & True Discoveries    & 7      & 7      & 7     & 3                     & 10                    & \textbf{13}                             \\
		Automatic cut-off      & False Discoveries   & 2      & \textbf{1}      & \textbf{1}     & 6                     & 8                     & 2                              \\
		& True Discoveries    & 12     & 11     & 11    & 14                    & \textbf{16}                    & 15                             \\
		AUC                    &                     & 0.71   & 0.72   & \textbf{0.73}  & 0.70                  & 0.69                  & 0.75                           \\ \hline
		\multicolumn{8}{c}{D4T}                                                                                                                                 \\ \hline
		Adaptive cut-off       & True Discoveries    & \textbf{9}      & 8      & \textbf{9}     & 5                     & 0                     & 8                              \\
		Automatic cut-off      & False Discoveries   & 2      & \textbf{1}      & \textbf{1}     & 3                     & 12                    & 80                             \\
		& True Discoveries    & 12     & 12     & 11    & 12                    & 14                    & \textbf{24}                             \\
		AUC                    &                     & \textbf{0.75}   & \textbf{0.75}   & \textbf{0.75}  & 0.70                  & 0.70                  & 0.73                           \\ \hline
		\multicolumn{8}{c}{DDI}                                                                                                                                 \\ \hline
		Adaptive cut-off       & True Discoveries    & 5      & 5      & 6     & 7                     & 3                     & \textbf{10}                             \\
		Automatic cut-off      & False Discoveries   & \textbf{1}      & \textbf{1}      & \textbf{1}     & 2                     & 11                    & 81                             \\
		& True Discoveries    & 8      & 7      & 7     & 8                     & 13                    & \textbf{24}   \\
		AUC                    &                     & 0.71   & 0.73   &\textbf{ 0.74}  & 0.68                  & 0.66                  & 0.72                           \\ \hline
		\multicolumn{8}{c}{TDF}                                                                                                                                 \\ \hline
		Adaptive cut-off       & True Discoveries    & 4      & \textbf{9}      & \textbf{9}     & 3                     & 7                     & 2                              \\
		Automatic cut-off      & False Discoveries   & 2      & \textbf{1}      & \textbf{1}     & 4                     & 11                    & 8                              \\
		& True Discoveries    & 9      & 9      & 9     & 10                    & 18                    & \textbf{15}                             \\
		AUC                    &                     & 0.69   & 0.72   & 0.72  & 0.72                  & \textbf{0.75}                  & 0.73                           \\ \hline
\end{tabular}}

\end{table}

\begin{table}
	\caption{\label{tab:result_NNRTI} The table summarizes results for a drug class NNRTI. There are three performance criteria. For the adaptive cutoff, we report the number of true discoveries since the number of false discoveries is 0. For the automatic cutoff, we report both the number of false and true discoveries. Finally, we report a cutoff-free metric AUC. The best performance in each row is in bold font.}
	
		\centering
\scalebox{0.8}{\begin{tabular}{llllllll}
		\hline
		\multicolumn{8}{c}{DLV}                                                                                                                                 \\ \hline
		\multicolumn{2}{l}{\multirow{2}{*}{Methods}} & \multicolumn{3}{c}{ABC} & \multirow{2}{*}{BART} & \multirow{2}{*}{DART} & \multirow{2}{*}{Random Forest} \\
		\multicolumn{2}{l}{}   & 100    & 500    & 1000  &                       &                       &                                \\ \hline
		Adaptive cut-off       & True Discoveries    & \textbf{4}      & \textbf{4}      & \textbf{4}     & 3                     & 3                     & 3                              \\
		Automatic cut-off      & False Discoveries   & \textbf{3}      & \textbf{3}      & \textbf{3}     &\textbf{ 3}                     & 8                     & 96                             \\
		& True Discoveries    & 7      & 7      & 7     & 9                     & 10                    & \textbf{14 }                            \\
		AUC                    &                     & 0.84   &\textbf{ 0.87}   & \textbf{0.87}  & 0.73                  & 0.70                  & 0.81                           \\ \hline
		\multicolumn{8}{c}{EFV}                                                                                                                                 \\ \hline
		Adaptive cut-off       & True Discoveries    & \textbf{5}      & \textbf{5}      & \textbf{5}     &\textbf{ 5}                     & 4                     & 4                              \\
		Automatic cut-off      & False Discoveries   & 5      & \textbf{4}      & \textbf{4}     & 5                     & 6                     & 9                              \\
		& True Discoveries    & 8      & 7      & 6     & 9                     & 9                     & \textbf{10}                             \\
		AUC                    &                     & 0.80   & 0.83   & \textbf{0.84}  & 0.74                  & 0.73                  & 0.78                           \\ \hline
		\multicolumn{8}{c}{NVP}                                                                                                                                 \\ \hline
		Adaptive cut-off       & True Discoveries    & 6      & 6      & 6     & 8                     & 6                     & \textbf{14}                             \\
		Automatic cut-off      & False Discoveries   & 3      & 3      & \textbf{2}     & \textbf{2}                     & 9                     & 97                             \\
		& True Discoveries    & 6      & 6      & 5     &\textbf{ 7}                     & 6                     & 5                              \\
		AUC                    &                     & 0.79   & 0.79   & 0.79  & 0.71                  & 0.66                  &\textbf{ 0.82 }                          \\ \hline
\end{tabular}}

\end{table}

\clearpage


\end{document}